\documentclass[11pt]{article} 
\usepackage[margin=1in]{geometry}

\usepackage[margin=1cm]{caption}

\usepackage{bm, amsmath, amsfonts, bbm, amsthm}
\usepackage{graphicx}
\usepackage{subcaption}
\usepackage{multirow}
\usepackage{algorithm}
\usepackage{algpseudocode}

\usepackage[round]{natbib}
\setcitestyle{citesep={,}}

\usepackage{hyperref}
\hypersetup{
    colorlinks,
    linkcolor={blue},
    citecolor={blue},
    urlcolor={blue}
}

\usepackage[page]{appendix}

\newcommand{\Var}{\mathrm{Var}}
\newcommand{\Cov}{\mathrm{Cov}}
\newcommand{\bx}{\bm{x}}
\newcommand{\by}{\bm{y}}
\newcommand{\bOne}{\bm{1}}
\newcommand{\bd}{\bm{\delta}}
\newcommand{\Rm}{R(m^N)}
\newcommand{\pw}{\tilde{p}_{\text{pred}}}
\newcommand{\xic}{\xi^{\text{conj}}}

\newcommand{\bxN}{\bm{x}^{N}}
\newcommand{\byN}{\bm{y}^{N}}
\newcommand{\xN}{x^{N}}
\newcommand{\yN}{y^{N}}
\newcommand{\bvN}{\bm{v}^{N}}
\newcommand{\vN}{v^{N}}
\newcommand{\dN}{\delta^{\mN}}
\newcommand{\bdxN}{\bm{\delta}_x^{\mN}}
\newcommand{\bdyN}{\bm{\delta}_y^{\mN}}
\newcommand{\Wy}{W_y(m^{N})}
\newcommand{\Wx}{W_x(m^{N})}
\newcommand{\W}{W(m^{N})}
\newcommand{\mN}{m^{N}}
\newcommand{\nxN}{n_x^{N}}
\newcommand{\nyN}{n_y^{N}}

\newcommand{\gc}{g^{\text{con}}}
\newcommand{\msa}{m_\text{stop}^{\text{asym}}}

\newtheorem{theorem}{Theorem}
\newtheorem{definition}{Definition}
\newtheorem{corollary}{Corollary}
\newtheorem{lemma}{Lemma}
\newtheorem{proposition}{Proposition}

\title{Fast Approximation of Small p-values in Permutation Tests by Partitioning the Permutations}
\author{Brian Segal,$^*$ Thomas Braun, Michael Elliott, and Hui Jiang\\
Department of Biostatistics, University of Michigan\\
\href{mailto:bdsegal@umich.edu}{bdsegal@umich.edu}
}

\begin{document}

\maketitle

\begin{quote}
Researchers in genetics and other life sciences commonly use permutation tests to evaluate differences between groups. Permutation tests have desirable properties, including exactness if data are exchangeable, and are applicable even when the distribution of the test statistic is analytically intractable. However, permutation tests can be computationally intensive. We propose both an asymptotic approximation and a resampling algorithm for quickly estimating small permutation p-values (e.g. $<10^{-6}$) for the difference and ratio of means in two-sample tests. Our methods are based on the distribution of test statistics within and across partitions of the permutations, which we define. In this article, we present our methods and demonstrate their use through simulations and an application to cancer genomic data. Through simulations, we find that our resampling algorithm is more computationally efficient than another leading alternative, particularly for extremely small p-values (e.g. $<10^{-30}$). Through application to cancer genomic data, we find that our methods can successfully identify up- and down-regulated genes. While we focus on the difference and ratio of means, we speculate that our approaches may work in other settings.\\

\end{quote}

% \tableofcontents
% \addtocontents{toc}{\protect\setcounter{tocdepth}{1}}

\section{Introduction and Motivation \label{sec_intro}}

Many researchers in the life sciences use permutation tests, for example, to test for differential gene expression \citep{doerge1996, morley2004, stranger2005, stranger2007, raj2014}, and to analyze brain images \citep{nichols2002, bartra2013, simpson2013}. These tests are useful when the sample size is too small for large sample theory to apply, or when the distribution of the test statistic is analytically intractable. Permutation tests are also exact, meaning that they control the type I error rate exactly for finite sample size \citep{lehmann2006testing}. However, permutation tests can be computationally intensive, especially when estimating small p-values for many tests. In this paper, we present computationally efficient methods for approximating small permutation p-values (e.g. $<10^{-6}$) for the difference and ratio of means in two-sample tests, though we speculate that our methods will also work for other smooth function of the means.

We denote the two groups of sample data as $\bm{x}=(x_1,\ldots,x_{n_x})'$ and $\bm{y}=(y_1,\ldots,y_{n_y})'$, with respective sample sizes $n_x$ and $n_y$. We denote the full data as $\bm{z}=(\bm{x}',\bm{y}')'$, with total sample size $N=n_x + n_y$. Writing $\bm{z}=(z_1,\ldots, z_N)'$, we have that $z_i=x_i, i=1,\ldots,n_x$, and $z_{n_x +j} = y_j, j=1,\ldots,n_y$. In our setting, $z_i$ are scalar values for all $i=1,\ldots,N$. We use $\pi$ to denote a permutation of the indices of $\bm{z}$, i.e. $\pi : \{1,\ldots,N \} \rightarrow \{1,\ldots,N \}$ is a bijection, and we denote the permuted dataset corresponding to $\pi$ as $\bm{z}^*=(z^*_1,\ldots,z^*_N )'$, where $z^*_{\pi(i)}=z_i, i=1,\ldots,N$. We use the term \textit{correspondence} throughout this paper, so for clarity, we define our use of the term in Definition \ref{defCorrespondence}.
\begin{definition}[Correspondence]
Let $\bm{z}=(z_1,\ldots,z_N)'$ be the $N$-dimensional vector of observed data, and let $\pi: \{1,\ldots,N\} \rightarrow \{1,\ldots,N\}$ be a bijection (permutation) of the indices of $\bm{z}$. We say that the $N$-dimensional vector $\bm{z}^*=(z^*_1, \ldots, z^*_N)'$ \textit{corresponds} to permutation $\pi$ if $z^*_{\pi(i)}=z_i$ for all $i=1,\ldots,N$.
\label{defCorrespondence}
\end{definition}

It will also be useful to write the permuted dataset as $\bm{z}^*=({\bm{x}^*}', {\bm{y}^*}')'$, where $\bm{x}^*=(z^*_1,\ldots,z^*_{n_x})'$ and $\bm{y}^*=(z^*_{n_x+1},\ldots,z^*_N)'$ are the permuted group samples.

Let $T$ be a test statistic, such that larger values are more extreme, and let $t=T(\bx, \by)$ be the observed test statistic. Similar to \citet[][p. 636]{lehmann2006testing}, we denote the permutation p-value as $\hat{p} = \Pr(T \ge t | \bm{z}) = |\Psi|^{-1} \sum_{\pi \in \Psi} I [T(\bx^*, \by^*) \ge t]$, where $\Psi$ is the set of all permutations of the indices of $\bm{z}$ (also the symmetric group of order $N!$), $|\Psi| = N!$ is the number of elements in $\Psi$, $I$ is an indicator function, and for each $\pi$, $({\bx^*}', {\by^*}')'$ is the corresponding permuted dataset. The randomization hypothesis \citep[][Definition 15.2.1]{lehmann2006testing} asserts that under the null hypothesis, the distribution of $T$ is invariant under permutations $\pi \in \Psi$. This allows, for example, for the null hypothesis $H_0: z_i \overset{\text{iid}}{\sim} P, i = 1,\ldots, N$, or more generally, for exchangeability, $H_0: P(Z_1 =  z_1, \ldots Z_N = z_n) = P(Z_1 = z^*_1, \ldots, Z_N = z^*_N)$ for all permuted datasets $\bm{z}^*$.

The set $\Psi$ is typically too large to evaluate fully, so Monte Carlo methods are usually used to approximate $\hat{p}$. When resampling with replacement, also known as simple Monte Carlo resampling, the Monte Carlo estimate of $\hat{p}$ is $\tilde{p}=(B+1)^{-1} \left( \sum_{b=1}^B I \left[T_b \ge t \right]+1 \right)$, where $B$ is the number of resamples, and $T_b=T(\bx^*, \by^*)$ for $({\bx^*}', {\by^*}')'$ corresponding to the $b^{th}$ randomly sampled permutation $\pi_b$. We refer to the above estimate as the adjusted $\tilde{p}$, because it adjusts the estimate to ensure it stays within its nominal level \citep{lehmann2006testing, phipson2010permutation}. However, for simplicity and to be consistent with other computationally efficient methods, particularly that of \citet{yu2011}, we use the unadjusted $\tilde{p}$, in which we remove the `+1' from the numerator and denominator.

While there may be many reasons for obtaining accurate small p-values, perhaps they are most often obtained in multiple testing settings, which are common in genetics. For example, in the analysis we present in Section \ref{sec_app}, we analyze 15,386 genes for differential expression. With a Bonferroni correction and a type I error rate of $\alpha=0.05$, to control the family-wise error rate (FWER), we would need to estimate $p \text{-values}<0.05/15,386 \approx 3.25 \times 10^{-6}$. While one might want to use a different correction to control the FWER, false discovery rate (FDR), or other criteria, we would still need to calculate small p-values before implementing typical step-up or step-down procedures (for example, \citet{holm1979simple} to control FWER, or \citet{benjamini1995controlling} to control FDR). These p-values, in combination with content area expertise and other statistical quantities, such as effect size, can be useful for prioritizing genes for further laboratory and statistical analysis. 

As noted by \citet{kimmel2006fast} and \citet{yu2011}, with simple Monte Carlo resampling, to estimate p-values on the order of $\hat{p}=10^{-6}$ with a precision of $\sigma_{\hat{p}}=\hat{p}/10$, we need on the order of $B=10^8$ iterations when using simple Monte Carlo resampling. For example, to estimate 5,000 permutation p-values that are each on the order of $10^{-6}$, we would need a total of $5,000 \times 10^8 = 5 \times 10^{11}$ iterations.

Several researchers have developed methods for reducing the computational burden of permutation tests, including \citet{robinson1982, mehta1983network, booth1990, kimmel2006fast, conneely2007, li2008, han2009rapid, knijnenburg2009fewer, pahl2010permory, Zhang2011, jiang2012statistical}, and \citet{zhou2015hypothesis}. For comparisons with our method, we focus on the stochastic approximation Monte Carlo (SAMC) algorithm developed by \citet{liang2007} and tailored to p-value estimation by \citet{yu2011}. Of the available methods, we found that SAMC was the most appropriate comparison, because: 1) we could directly apply it to the test static in our motivating application (see Section \ref{sec_app}), 2) it is intended for very small p-values, and 3) it does not require difficult derivations, so is more likely to be used in practice.

In this article, we propose alternative methods for quickly approximating small permutation p-values for the difference and ratio of the means in two-sample tests. Our approaches partition the permutations such that $\tilde{p}$ has a predictable trend across the partitions. Taking advantage of this trend, we develop both a closed form asymptotic approximation to the permutation p-value, as well as a computationally efficient resampling algorithm.

We find through simulations that our resampling algorithm is more computationally efficient than the SAMC algorithm, which in turn is 100 to 500,000 times more computationally efficient than simple Monte Carlo resampling \citep{yu2011}. However, SAMC is a more general algorithm, and can be used for a greater variety of statistics. The increase in efficiency is most notable for our algorithm when estimating extremely small p-values (e.g. $<10^{-30}$). Our asymptotic approximation tends to be less accurate than our resampling algorithm, but does not require resampling.
 
Before presenting our methods, we briefly explain the underlying properties that make them possible. The two basic components underlying our methods are 1) the partitions, which we define, and the distribution of permutations across these partitions, and 2) the limiting behavior of test statistics within each partition, and the trend in p-values across the partitions. We address the first component in Section \ref{sec_partPermSpac}, and the second in Section \ref{sec_statDist}.

In Section \ref{sec_alg}, we introduce methods for estimating permutation p-values that take advantage of the properties discussed in Sections \ref{sec_partPermSpac} and \ref{sec_statDist}. In Section \ref{sec_simulations}, we investigate the behavior of these methods through simulations and compare against the SAMC algorithm (additional simulations and comparisons against other methods are in the Appendices). Then in Section \ref{sec_app}, we use our proposed methods to analyze cancer genomic data. In Section \ref{discussion}, we end with a discussion of limitations and possible extensions. As noted under Supplementary material, we have implemented our methods in the \textsf{R} package \verb|fastPerm|.

\section{Partitioning the permutations \label{sec_partPermSpac}}

\subsection{Defining the partitions \label{sec_dist}}

Let the smaller of the two sample sizes be $n_{\min} = \min(n_x, n_y)$. We define the distance between permutation $\pi$ and the observed ordering of the indices $(1,2,3,\ldots,N)$ as the number of observations that are exchanged between $\bx$ and $\by$ under the action of $\pi$. To be precise, let $\omega(\pi)$ be the set of indices that $\pi$ places in one of the first $n_x$ positions, i.e. $\omega(\pi) = \{i \in \{1,\ldots,N\} : \pi(i) \le n_x \}$. Then we define the distance, denoted as $d(\pi)$, between permutation $\pi$ and the observed ordering, as
\begin{equation}
d\left(\pi \right)= n_x - |\omega(\pi) \cap \{1,2,\ldots,n_x\}| \label{d}.
\end{equation}

We define partition $m$, denoted as $\Pi(m)$, as the set of all permutations a distance of $m$ away from the observed ordering, i.e. $\Pi(m)=\left\{\pi : d \left( \pi \right) = m \right\}$, $m = 0,1,\ldots,n_{\min}.$ As described below, our proposed methods focus on the permutation distributions of test statistics when resampling is restricted to permutations from a single partition.

To see why this definition of distance is useful, and to foreshadow our method, suppose that $\mu_x \ne \mu_y$, and note that as observations are exchanged between $\bm{x}$ and $\bm{y}$, the empirical distributions of the permuted samples $\bm{x}^*$ and $\bm{y}^*$ tend to become more similar. Consequently, test statistics that measure changes in the mean tend to become less extreme. For example, suppose that $n=n_x=n_y$ with $n$ even, and let $\bm{z}^*=({\bm{x}^*}', {\bm{y}^*}')'$ be a permuted dataset corresponding to a permutation $\pi \in \Pi(n/2)$. Then half of the observations in $\bm{x}^*$ are from $\bm{x}$ and half are from $\bm{y}$, and the same is true for $\bm{y}^*$. Consequently, we would expect $\bar{x}^* \approx \bar{y}^*$, where $\bar{x}^*$ and $\bar{y}^*$ are the means of the permuted samples.

To make this explicit, and again assuming that $n = n_x = n_y$, let $\bm{\delta}^{\pi}_x=(\delta^{\pi}_{x,1},\ldots, \delta^{\pi}_{x, n})'$ and $\bm{\delta}^{\pi}_y=(\delta^{\pi}_{y,1},\ldots, \delta^{\pi}_{y,n})'$ be $n \times 1$ indicator vectors designating which observations are exchanged between $\bm{x}$ and $\bm{y}$ under the action of permutation $\pi$:
\begin{align*}
\delta^{\pi}_{x,i}=
\begin{cases}
1 \text{ if } \pi(i) >n \\
0 \text{ if } \pi(i) \le n
\end{cases}, i=1,\ldots,n, &&
\delta^{\pi}_{y,j}=
\begin{cases}
1 \text{ if } \pi(n + j) \le n \\
0 \text{ if } \pi(n + j) > n
\end{cases}, j=1, \ldots, n.
\end{align*}
Under the action of permutation $\pi$, $\bar{x}^*=n^{-1}\left[(\bm{1} - \bm{\delta}^{\pi}_x)'\bm{x} + \left(\bm{\delta}^{\pi}_y \right)'\bm{y} \right]$, where $\bm{1}$ is an $n \times 1$ vector of ones. Assuming uniform distribution of the permutations $\pi$, $\mathbb{E} \left[ \bm{\delta}^{\pi}_{x} | \pi \in \Pi(m) \right]=(m/n)\bm{1}$, an $n \times 1$ vector with all elements equal to $m/n$. Consequently, $\mathbb{E}[\bar{x}^* | \pi \in \Pi(m), \bm{x}, \bm{y}] = \bar{x} + (m/n)(\bar{y} - \bar{x})$ and $\mathbb{E}[\bar{y}^* | \pi \in \Pi(m), \bm{x}, \bm{y}] = \bar{y} + (m/n)(\bar{x} - \bar{y})$.

Then, for example, with the test statistic $T=\bar{x} - \bar{y}$, we have that $\mathbb{E}[T(\bx^*, \by^*)|\pi \in \Pi(m), \bx, \by] = (\bar{x} - \bar{y})(1-2m/n)$, where $\bx^*, \by^*$ are the permuted samples corresponding to a permutation $\pi \in \Pi(m)$, $m=0,\ldots,n$. This shows that the expected value of $T$ is zero when, for both $\bm{x}^*$ and $\bm{y}^*$, half of the observations are from $\bm{x}$ and half are from $\bm{y}$, i.e. in the $m=n/2$ partition. Similarly, the magnitude of $T$ is $|\bar{x}-\bar{y}|$ when either none or all of the observations are exchanged between $\bm{x}$ and $\bm{y}$ (partitions $m=0$ and $m=n$, respectively). This example demonstrates that test statistics tend to be less extreme when the permuted group samples, $\bm{x}^*$ and $\bm{y}^*$, each contain a mixture of elements from the observed group samples, $\bm{x}$ and $\bm{y}$. Similar results hold for unbalanced sample sizes.

\subsection{Distribution of the partitions \label{sec_subOrbDist}}

Uniform sampling of the permutations $\pi$ leads to a non-uniform distribution of the partitions $\Pi(m)$. The probability of drawing a permutation from partition $m$ under uniform sampling, which we denote as $f(m), m =1,\ldots, n_{\min}$, is given by
\begin{align*}
f \left(m \right)
  &\propto \left| \Pi (m) \right| && (\pi \sim \text{Uniform}) \\
  &=\binom{n_x}{m} \binom{n_y}{m},
\end{align*}
where the last line follows directly from the definition of $\Pi(m)$. The normalizing constant is $\sum_{j=0}^{n_{\min}} \binom{n_x}{j} \binom{n_y}{j}=\binom{N}{n_{\min}}$, so
\begin{equation}
f \left(m \right)=\binom{N}{n_{\min}}^{-1} \binom{n_x}{m} \binom{n_y}{m}. \label{pmf}
\end{equation}
As described in Section \ref{sec_alg}, in our proposed methods, we use $f$ to weight the partition-specific p-values in order to obtain an overall p-value.

We note that in practice, directly using (\ref{pmf}) to calculate $f(m)$ is not possible for large $n_x$ and $n_y$, because the binomial coefficients become too large to represent on most computers. However, by noting the relationship between the gamma function and factorials, we can compute (\ref{pmf}) for large sample sizes with the equivalent form:
\begin{align*}
f \left(m \right) &=  \exp \{ \log \Gamma(n_x + 1) - \log \Gamma(n_x - m +1) \\
  & + \log \Gamma(n_y + 1) - \log \Gamma(n_y - m +1) -2 \log \Gamma(m+1)\\
  & - \log \Gamma(N +1) + \log \Gamma(N - n_{\max} +1) + \log \Gamma(n_{\max} +1) \},
\end{align*}
where $\log \Gamma$ is the log gamma function.

\section{Trend in p-values across the partitions \label{sec_statDist}}

In this section, we describe the trend in p-values across the partitions, both with asymptotic and simulated results. The results described in this section are given in greater detail in Appendix \ref{proofs}, and are the basis for our proposed methods.

Let $T$ be a two-sided test statistic that is a function of the means, such that larger values are more extreme. In particular, we study $T=|\bar{x}-\bar{y}|$ and $T=\max(\bar{x}/\bar{y},\bar{y}/\bar{x})$. $T$ is a random variable, and we could calculate its value for all permutations of the data to get its permutation distribution. To be explicit, we define the random variable $T(m)$ such that $\Pr \left( T(m) > t | \bm{z} \right) = \Pr \left( T(\bx^*, \by^*) >t | \bm{z}, \pi \in \Pi(m) \right)$, i.e., $T(m)=T(\bx^*, \by^*)$ restricted to permutations in partition $m$. To be concrete, we could, in principle, compute the permutation p-value, $\Pr(T(m) > t | \bm{z})$, as $\hat{p}(m)=|\Pi(m)|^{-1} \sum_{\pi \in \Pi(m)} I [ T(\bx^*, \by^*) \ge t ]$, where for each $\pi \in \Pi(m)$, $({\bx^*}', {\by^*}')'$ is the corresponding permuted dataset.

Regarding notation, if there are two vector-valued arguments to $T$, e.g. $T(\bx, \by)$ then $T$ is the test statistic computed with data $\bx, \by$. If the argument to $T$ is a single scalar, e.g. $T(m)$, then $T$ is a test statistic computed with some permuted dataset $\bm{z}^*$, where $\bm{z}^*$ corresponds to a permutation $\pi \in \Pi(m)$. This notation facilitates further analysis in Appendix \ref{proofs}.

While we are primarily interested in two-sided statistics $T$ in this paper, it helps to first note results for their one-sided counterparts, which we denote by $R$. In particular, $R=\bar{x} - \bar{y}$ and $R=\bar{x}/\bar{y}$. Similar to before, let $R(m)=R(\bx^*, \by^*)$ restricted to permutations in partition $m$. As shown in Corollary \ref{Tnormal} of Appendix \ref{proofs}, under certain regularity conditions and sufficiently large sample sizes, $R(m) \sim N(\nu(m), \sigma^2(m))$, where $\nu(m)$ and $\sigma^2(m)$ are functions of the partition $m$, as well as the sample means and variances of $\bm{x}$ and $\bm{y}$. The regularity conditions are standard assumptions for finite sample central limit theorems and the delta method, requiring that the tails of the distributions of the data are not too large, and that the derivative of $R$ exists at the means.

As described in Corollary \ref{cor1} of Appendix \ref{proofs}, a direct consequence of the limiting normality of $R(m)$ is that for $n_x$ and $n_y$ sufficiently large,
\begin{equation}
\Pr \left( T(m) \ge t | \bm{z} \right) \approx 2-\Phi \left[\xi \left(\min \left\{ m, 2 m_{\max} - m \right\} \right) \right] - \Phi \left[\xic \left(\min \left\{ m, 2 m_{\max} - m \right\} \right) \right],
\label{approxPval}
\end{equation}
where $\Phi$ is the standard normal cumulative density function (CDF), $m_{\max} = \arg \max_m f (m)$, and $\xi$ and $\xic$ are functions of the partition $m$ and data $\bm{z}$, whose form depends on the statistic $T$. The functions $\xi$ and $\xic$ are identical in form, but reverse the role of the means of the permuted samples, $\bar{x}^*$ and $\bar{y}^*$. This accounts for the two-sided form of $T$. Equation \ref{approxPval} is the basis for our asymptotic approximation, which is described in Section \ref{asymSec}.

The proof of (\ref{approxPval}) involves the fact that $\Pr \left( T(m) \ge t | \bm{z} \right)$, as a function of $m$, is approximately symmetric about $m_{\max}$. This symmetry is exact when $n_x=n_y$, and less accurate as the group sample sizes become imbalanced. Consequently, the accuracy of the approximation in (\ref{approxPval}) is best for equal group sample sizes, and worsens as the group sample sizes become more imbalanced.

The result in (\ref{approxPval}) and the form for $\xi$ and $\xic$ shown in Appendix \ref{proofs} for $T=\max(\bar{x}/\bar{y}, \bar{y}/\bar{x})$ give the smooth pattern shown in Figure \ref{pApproxTheo} for $n_x=n_y=100, \mu_x=\sigma^2_x=4$, and $\mu_y=\sigma^2_y=2$. In the case where $n_x \ne n_y$, the center of the trend shifts, but is otherwise similar.

The smooth trend shown in Figure \ref{pApproxTheo} is primarily an observation, though it holds with striking similarity for both $T=|\bar{x} - \bar{y}|$ and $T=\max(\bar{x}/\bar{y}, \bar{y}/\bar{x})$ for a wide range of group sample sizes and parameter values. This observation is the basis for our resampling algorithm, described in Section \ref{resampAlgo}. 

Figure \ref{logp_sim} shows simulated results with $B = 10^3$ iterations within each partition for data coming from the following distributions with $n_x=n_y=100$: Poisson with rates $\lambda_x=4$ and $\lambda_y=2$; exponential with rates $\lambda_x = 2$ and $\lambda_y=1$; log normal with means $\mu_x=2$ and $\mu_y = 1$, and variances $\sigma^2_x=\sigma^2_y=1$, where $\mu$ and $\sigma^2$ are the means and variances of the log; and negative binomial with size $r_x=r_y=3$, and probability of success $p=r/(r+\mu)$, where the means are $\mu_x=4$ and $\mu_y=2$. For visual comparison between theoretical and simulated results, Figure \ref{pApproxFig_ratio_cutoff} shows the theoretical values cut off at $10^{-3}$.

Note that the p-value for the $m=0$ partition is always 1, as the only permutation in that partition is the observed test statistic. The same holds for partition $m=n_{\min}$ when $n_x=n_y$.
\begin{figure}[htbp]
\centering
\begin{subfigure}{0.33\textwidth}
  \centering
  \includegraphics[width=\linewidth]{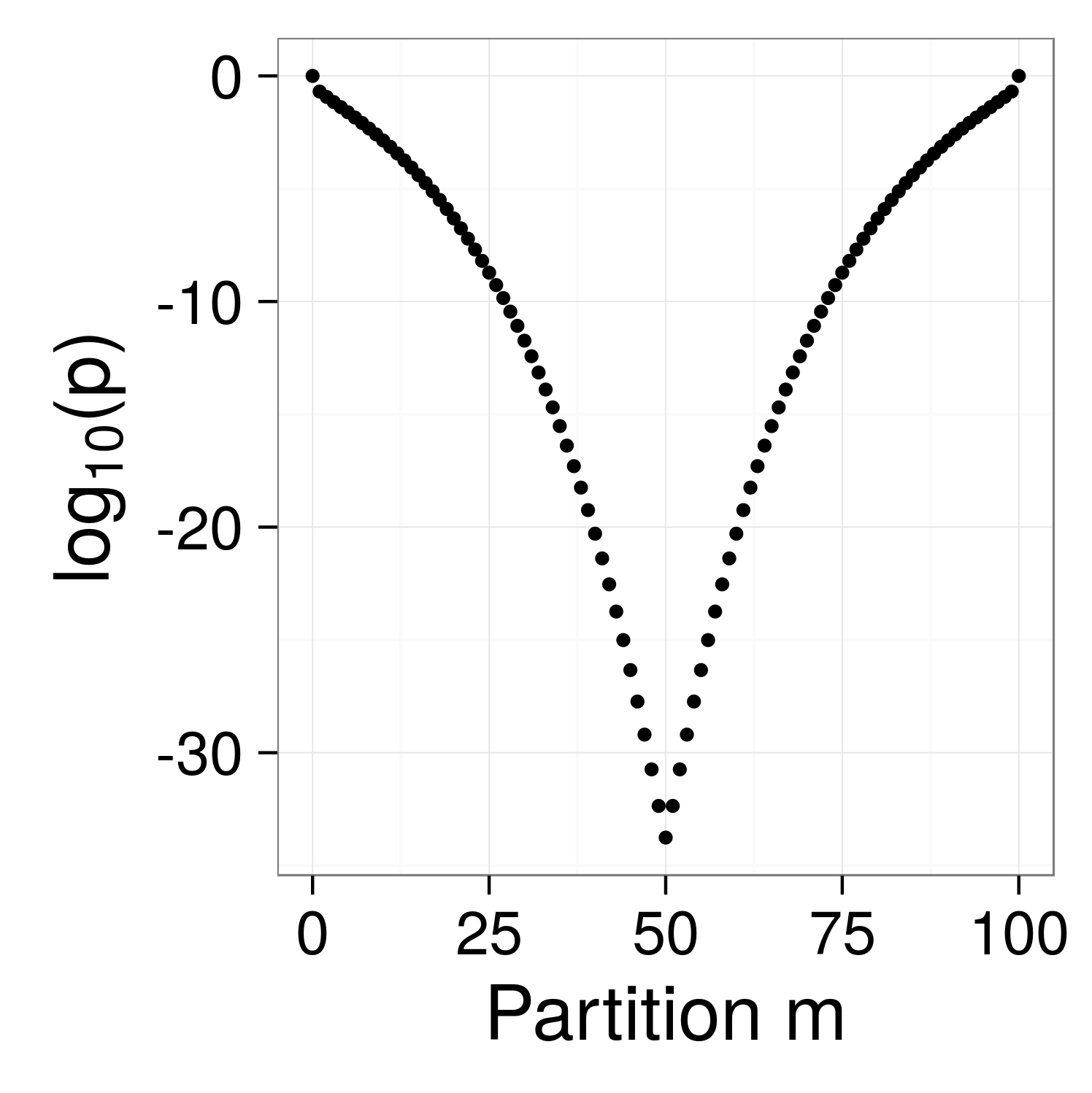}
  \caption{Theoretical trend}
  \label{pApproxFig_ratio_100_100}
\end{subfigure}
\begin{subfigure}{0.33\textwidth}
  \includegraphics[width=1\linewidth]{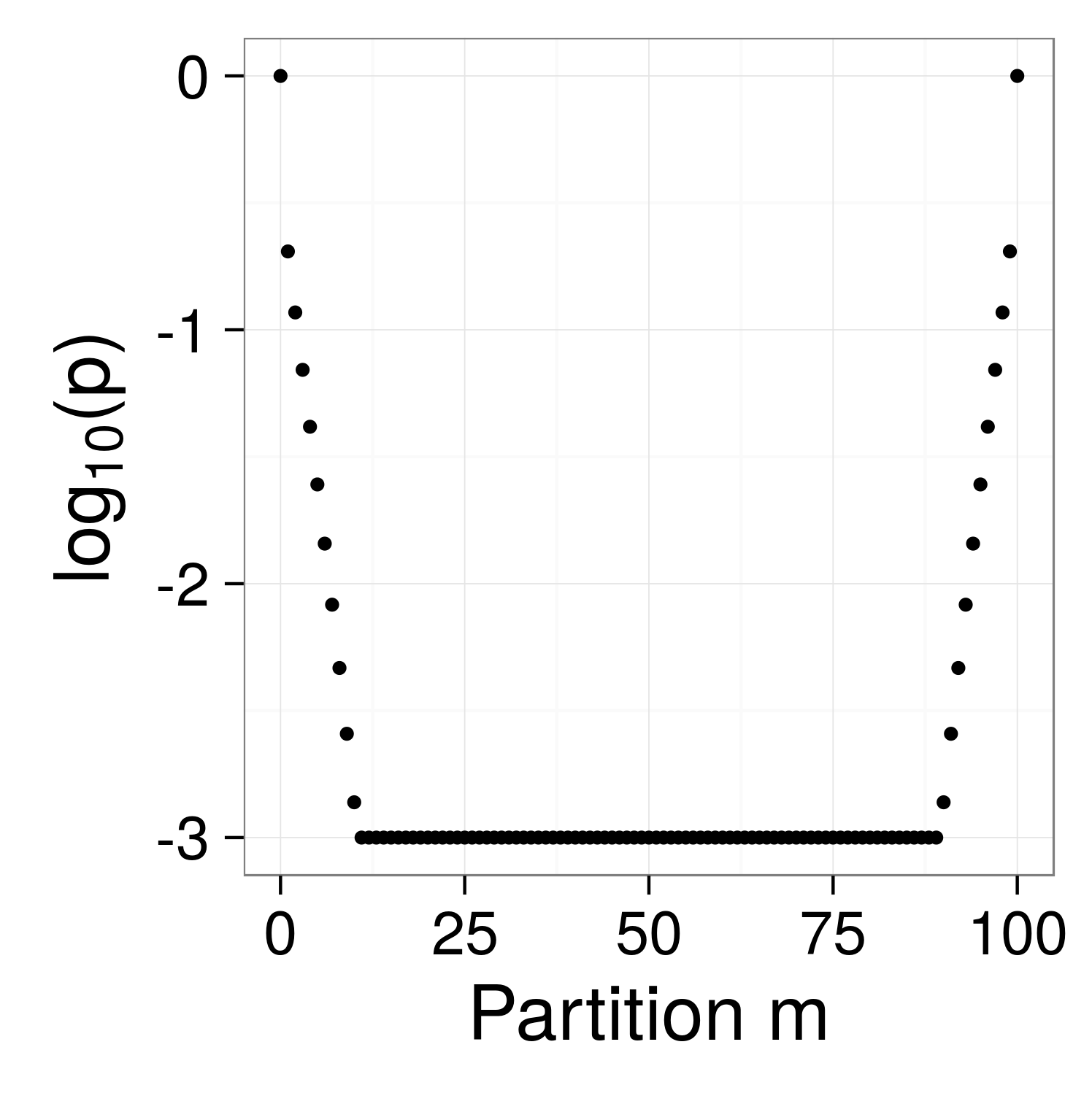}
  \caption{Theoretical trend cut off at $10^{-3}$}
  \label{pApproxFig_ratio_cutoff}
\end{subfigure}
\caption{Theoretical trend in p-values with $T=\max(\bar{x}/\bar{y},\bar{y}/\bar{x})$ for $n_x=n_y=100, \mu_x=\sigma^2_x=4$, and $\mu_y=\sigma^2_y=2$.}
\label{pApproxTheo}
\end{figure}

\begin{figure}[htbp]
\centering
\includegraphics[scale=0.44]{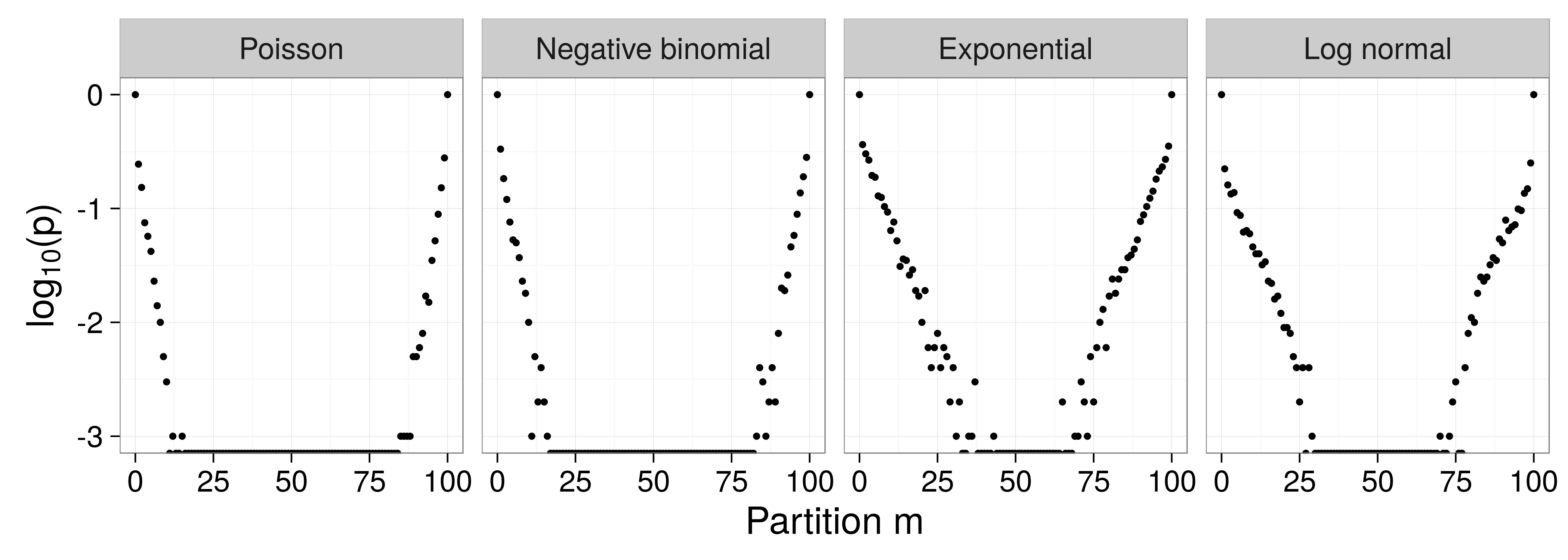}
\caption{Simulated trend in p-values with $B=10^3$ iterations within each partition and $T=\max(\bar{x}/\bar{y},\bar{y}/\bar{x})$}
\label{logp_sim}
\end{figure}

\section{Proposed methods \label{sec_alg}}
In this section, we propose two methods for approximating small permutation p-values: 1) a closed-form asymptotic approximation, and 2) a computationally efficient resampling algorithm. First, we note that we can express the permutation p-value as
\begin{equation}
\Pr (T \ge t | \bm{z}) =\sum_{m=0}^{n_{\min}} \Pr \left( T(m) \ge t \, | \bm{z} \right) f \left( m \right).
\label{p_w}
\end{equation}
Both the asymptotic and resampling-based approaches involve approximations for the\\ $\Pr \left( T(m) \ge t \, | \bm{z} \right)$ terms in (\ref{p_w}). The asymptotic approach uses (\ref{approxPval}) to approximate these terms, whereas the resampling-based algorithm uses the trend across the partitions to predict the terms.

If multiplicity corrections are needed, researchers can apply step-up or step-down procedures to the p-values produced by our method (for example, \citet{holm1979simple} to control FWER, or \citet{benjamini1995controlling} to control FDR).

\subsection{Asymptotic approximation \label{asymSec}}

Our asymptotic approximation to the permutation p-value is given by $\hat{p}_{\text{asym}} = \sum_{m=0}^{n_{\min}} h(m) f(m)$, where
\begin{align*}
h(0)&=1 \\
h(m) &= 2-\Phi \left[\xi(\min \left\{ m, 2 m_{\max} - m \right\}) \right] - \Phi \left[\xic(\min \left\{ m, 2 m_{\max} - m \right\}) \right],  m  \in [1, n_{\min} -1] \\
h(n_{\min}) &=
\begin{cases}
1 &\text{ if } n_x = n_y \\
2-\Phi \left[\xi(\min \left\{ m, 2 m_{\max} - m \right\}) \right] - \Phi \left[\xic(\min \left\{ m, 2 m_{\max} - m \right\}) \right] &\text{ otherwise}
\end{cases}
\end{align*}
To see why $h(0)=1$ always, and $h(n_{\min})=1$ when $n_x=n_y$, note that the p-value is always 1 in the $m=0$ partition, because this partition only contains the observed permutation. The same is true for the $n_{\min}$ partition when $n_x=n_y$, as $T$ is a two-sided statistic.

Regarding notation, we use a hat in $\hat{p}_{\text{asym}}$, as opposed to a tilde, to emphasize that we are not using Monte Carlo methods.

\subsection{Resampling algorithm \label{resampAlgo}}

As noted in Section \ref{sec_statDist}, we could, in principle, estimate each $\Pr (T(m) \ge t | \bm{z})$ term in (\ref{p_w}) with Monte Carlo methods, but this would be more computationally intensive than directly estimating $\Pr (T \ge t | \bm{z})$ without conditioning on the partition. This is because for small p-values, $\Pr (T(m) \ge t | \bm{z})$ terms for $m$ near $m_{\max}$ (the middle partition when $n_x=n_y$) are very small, and so we would need to use an extremely large number of resamples to estimate these values. For example, see Figure \ref{pApproxFig_ratio_100_100}.

However, by taking advantage of the trend in p-values across the partitions, we can avoid directly calculating $\Pr (T(m) \ge t | \bm{z})$ for $m$ near $m_{\max}$. Instead, we use simple Monte Carlo resampling to estimate $\Pr (T(m) \ge t | \bm{z})$ sequentially for $m=1,2,\ldots, m_{\text{stop}}$, where $m_{\text{stop}}$ is the stopping partition, which, as described below, is determined dynamically. We then use a Poisson model to predict the $\Pr (T(m) \ge t | \bm{z})$ terms for the remaining partitions (as well as for partitions $m=1,\ldots, m_{\text{stop}}$), under the assumption that the log of the partition-specific p-values is linear in $m$.

We then take a weighted sum across the predicted partition-specific p-values, as in (\ref{p_w}), to obtain an overall p-value. We denote the resulting p-value as $\tilde{p}_{\text{pred}}$, where the tilde emphasizes the use of Monte Carlo methods, and the subscript emphasizes that the estimate is based on predicted counts within each partition.

As described in Algorithm \ref{alg_two_sided_general}, we set the number of Monte Carlo iterations within partitions at $B_\text{pred}$ (e.g., we use $B_\text{pred}=10^3$), and estimate $\Pr(T(m)>t|\bm{z})$ for $m=1,\ldots,m_{\text{stop}}$, where $m_{\text{stop}}$ is the first partition in which none of the resampled statistics are larger than the observed statistic.

We stop at partition $m_{\text{stop}}$ because the exponential decrease in p-values across the partitions, shown in Figure \ref{pApproxFig_ratio_100_100}, makes it nearly certain that we would not obtain a p-value greater than zero in partitions larger than $m_{\text{stop}}$ using only $B_\text{pred}=10^3$ iterations. In other words, it would be a waste of resources to continue sampling from additional partitions. Furthermore, since the trend is symmetric about $m_{\max}$, we can estimate the p-values in partitions $m = m_{\max}+1, \dots, n_{\min}$ using the p-values in partitions $m=1,\ldots, m_{\max}$.

Regarding the Poisson model, this is a natural choice for count data (the number of resampled statistics larger than the observed statistic within each partition), and also enforces a log-linear trend. Furthermore, we found that Poisson regression worked best in the simulations. In addition to our current approach of using a slope and intercept term in the Poisson model, we also experimented with using higher order polynomials and B-splines, and selecting the optimal order or degrees of freedom based on AIC. However, we found that this approach was too sensitive to noise in the data and sometimes gave highly erroneous results (e.g. p-values $> 1$).

In Algorithm \ref{alg_two_sided_general}, we represent vector indices by square brackets $[\cdot]$, and begin the index at zero because our partitions begin at $m=0$.  We use the vector $\bm{c}$ to store the count of permuted test statistics in each partition that are as large or larger than the observed test statistic, as obtained with simple Monte Carlo resampling, and use $\bm{c}_{\text{pred}}$ to store predicted counts based on a fitted model. We use $B_{\text{pred}}$ to denote that number of resamples within each partition.

\begin{algorithm}
\caption{$\tilde{p}_{\text{pred}}$}
\begin{algorithmic}[1]
\State set $m \gets 1$ and $\bm{c}[0] \gets B_{\text{pred}}$
\While {($m\le m_\text{max}$ and $\bm{c}[m-1]>0$)}
  \State for $b=1,\ldots,B_{\text{pred}}$, sample $\pi_b \in \Pi(m)$ uniformly and calculate $T_b(m)=T(\bx^*, \by^*)$ for $\bx^*, \by^*$ corresponding to $\pi_b$
  \State set $\bm{c}[m] \gets \sum_b I [ T_b(m) \ge t ]$ and update $m \gets m + 1$
\EndWhile
\State set $m_{\text{stop}} \gets m-1$ and $m_{\text{reg}} \gets \max_m \left\{ m \in \{ 1 \ldots,m_{\max} \} : \bm{c}[m] >0 \right\}$
\State regress $\bm{c}[0:m_{\text{reg}}]$ on $(0,\ldots, m_{\text{reg}})$ using a Poisson model with slope and intercept terms
\State predict $\bm{c}_\text{pred}$ for $m=1,\ldots,n_{\min}$ with fitted model, $s.t.$ $\bm{c}_\text{pred}$ is symmetric about $m_\text{max}$
\State set $\bm{c}_\text{pred}[0] \gets B_{\text{pred}}$, and if $n_x = n_y$, then set $\bm{c}_\text{pred}[n_x] \gets B_{\text{pred}}$
\State return $\pw \equiv (1/B_{\text{pred}}) \sum_{m=0}^{n_{\min}} \bm{c}_\text{pred}[m] f (m)$
\end{algorithmic}
\label{alg_two_sided_general}
\end{algorithm}

Our proposed algorithm runs in $O(B_{\text{pred}} m_{\text{stop}})$ time. In our current implementation, we set $B_{\text{pred}}$ a priori. Regarding $m_{\text{stop}}$, we obtain the following approximation for small p-values, in which we assume that $1-\Phi(\xi(m)) \gg 1 - \Phi(\xic(m))$. From Algorithm \ref{alg_two_sided_general},
\begin{align}
m_{\text{stop}}&= \min_m \left\{ m \in \{1,\ldots, m_{\max} \} : \bm{c}[m] < 1\right\} \nonumber \\
  & \approx \min_m \left\{ m \in \{1,\ldots, m_{\max} \} : \Pr(T(m) \ge t | \bm{x}, \bm{y}) < 1/B_{\text{pred}} \right\}  && \text{(for large $B_{\text{pred}}$)} \nonumber \\
  & \approx \min_m \left\{ m \in \{1,\ldots, m_{\max} \} : 1- \Phi(\xi(m)) < 1/B_{\text{pred}} \right\}  \label{xiDom} \\
  & \approx \min_m \left\{ m \in \{1,\ldots, m_{\max} \} : \Phi^{-1}(1-1/B_{\text{pred}}) < \xi(m) \right\} && \text{(for large $n_x, n_y$)} \label{mStopEst} \\
  & \equiv m^{\text{asym}}_{\text{stop}}, \nonumber
\end{align}
where (\ref{xiDom}) follows from (\ref{approxPval}) and the assumption that $1-\Phi(\xi(m)) \gg 1 - \Phi(\xic(m))$.

In the \textsf{R} package \verb|fastPerm|, we provide functions for computing $m^{\text{asym}}_{\text{stop}}$, which can help an analyst to approximate run-time before running the algorithm. We emphasize that $m^{\text{asym}}_\text{stop}$ is based on asymptotic approximations, and may not be the same as the actual stopping partition; $m^{\text{asym}}_{\text{stop}}$ is not used in Algorithm \ref{alg_two_sided_general}.

\section{Simulations \label{sec_simulations}}

To investigate the behavior of our proposed methods, we conducted simulations with the statistics $T=|\bar{x}-\bar{y}|$ and $T= \max (\bar{x}/\bar{y}, \bar{y}/\bar{x})$. We use the former statistic because the true permutation p-value can be approximated well with the p-value from a $t$-test, which provides a baseline for comparison, and the latter because it is the statistic of interest in our motivating application (Section \ref{sec_app}).

We simulated data under the alternative hypothesis, and given the extremely small p-values in our simulations, it was not feasible to compute the true permutation p-values for comparison. Instead, we used asymptotically equivalent p-values and large sample sizes.

In Appendix \ref{C}, we show results from additional simulations for 1) small sample sizes, and 2) data generated under the null hypothesis, in which case we approximated the true permutation p-value with simple Monte Carlo resampling, and 3) data generated as Gamma random variables. In Appendix \ref{D}, we also show simulations with the moment-corrected correlation (MCC) method of \citet{zhou2015hypothesis} using the statistic $T = |\bar{x} - \bar{y}|$, and compare our method with saddlepoint approximations \citep{robinson1982} by analyzing two small datasets ($n_x = n_y = 8$ and $n_x = 7, n_y = 10$), also using the statistic $T = | \bar{x} - \bar{y}|$. In Appendix \ref{E}, we show simulation results using our method with a studentized statistic to test null hypotheses regarding a single parameter as opposed to the full distribution, as described by \citet{chung2013exact}. The results in Appendices \ref{C} and \ref{D} show that the accuracy of our method is comparable to alternative methods, and the results in Appendix \ref{E} show that by using a studentized statistic, our method can be extended to null hypotheses specifying equality in the means ($H_0: \mu_x = \mu_y$), as opposed to equality in the entire distributions ($H_0: P_x = P_y$).

\subsection{Difference in means \label{sec_algoEval}}

In this section, we consider the test statistic $T=|\bar{x}-\bar{y}|$ with normally distributed data of equal variance. Since the t-test is asymptotically equivalent to the permutation test in this setting \citep[][p. 642-643]{lehmann2006testing}, we used the t-test as a baseline for comparison. We simulated data with both equal and unequal sample sizes ($n_x = n_y$ and $n_x \ne n_y$). In both cases, we generated data $x_{i}, i=1,\ldots,n_x$ and $y_j, j=1,\ldots,n_y$ as realizations of the respective random variables $X_i\overset{\text{iid}}{\sim} N(\mu_x,1)$ and $Y_{j} \overset{\text{iid}}{\sim} N(\mu_y,1)$, for various parameter values. For each combination of parameter values, we generated 100 datasets.

For equal sample sizes, we set $n=n_x=n_y=100$, 500, or 1,000. For unequal sample sizes, we set $n_y=500$, and $n_x=50$, 200, or 350. In both cases we set $\mu_y=0$ and $\mu_x=0.75$ or 1. For each dataset, we applied our methods and did a t-test with the \verb|t.test| function in \textsf{R} \citep{R} (two-sided with equal variance). For our resampling algorithm, we used $B_{\text{pred}}=10^3$ iterations in each partition.

For comparison, we also ran the SAMC algorithm using the \textsf{R} package \verb|EXPERT| written by \citet{yu2011}. We set the number of iterations in the initial round at $5 \times 10^4$, and the number of iterations in the final round at $10^6$. Following the advice of \citet{yu2011}, we set the gain factor sequence to begin decreasing after the $1,000^{th}$ iteration, the proportion of data to be updated at each iteration at 0.05, and the number of regions at 101 for the initial run and 301 for the final run.

Results are shown in Figures \ref{sim_sym} and \ref{sim_nonSym}. In the Figures, $p_t$ denotes the p-value from a two-sided t-test with equal variance, and $p$ denotes the p-value from either our methods or SAMC. The dashed line has a slope of 1 and intercept of 0, and indicates agreement between methods. The SAMC algorithm did not produce values for smaller p-values due to numerical problems, and so these points are missing from Figures \ref{sim_sym} and \ref{sim_nonSym} (385 missing points in Figure \ref{sim_sym}, and 179 missing points in Figure \ref{sim_nonSym}). In order to estimate these points with the \verb|EXPERT| implementation of the SAMC algorithm, we would need to increase the number of iterations.

\begin{figure}[htbp]
\centering
  \begin{subfigure}{0.58\textwidth}
  \centering
  \includegraphics[width=1\linewidth]{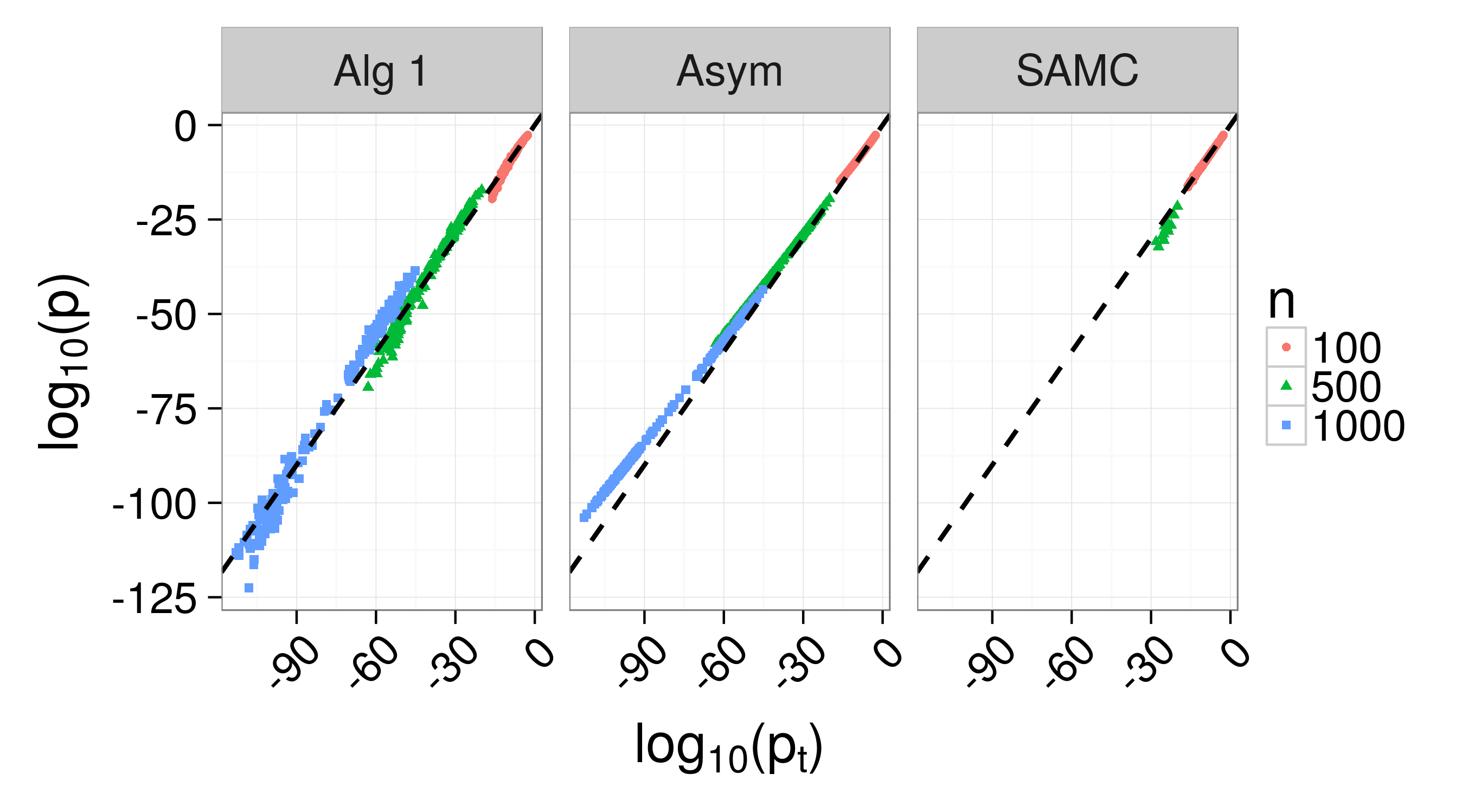}
  \caption{p-values}
  \end{subfigure}
  \begin{subfigure}{0.38\textwidth}
  \centering
  \includegraphics[width=1\linewidth]{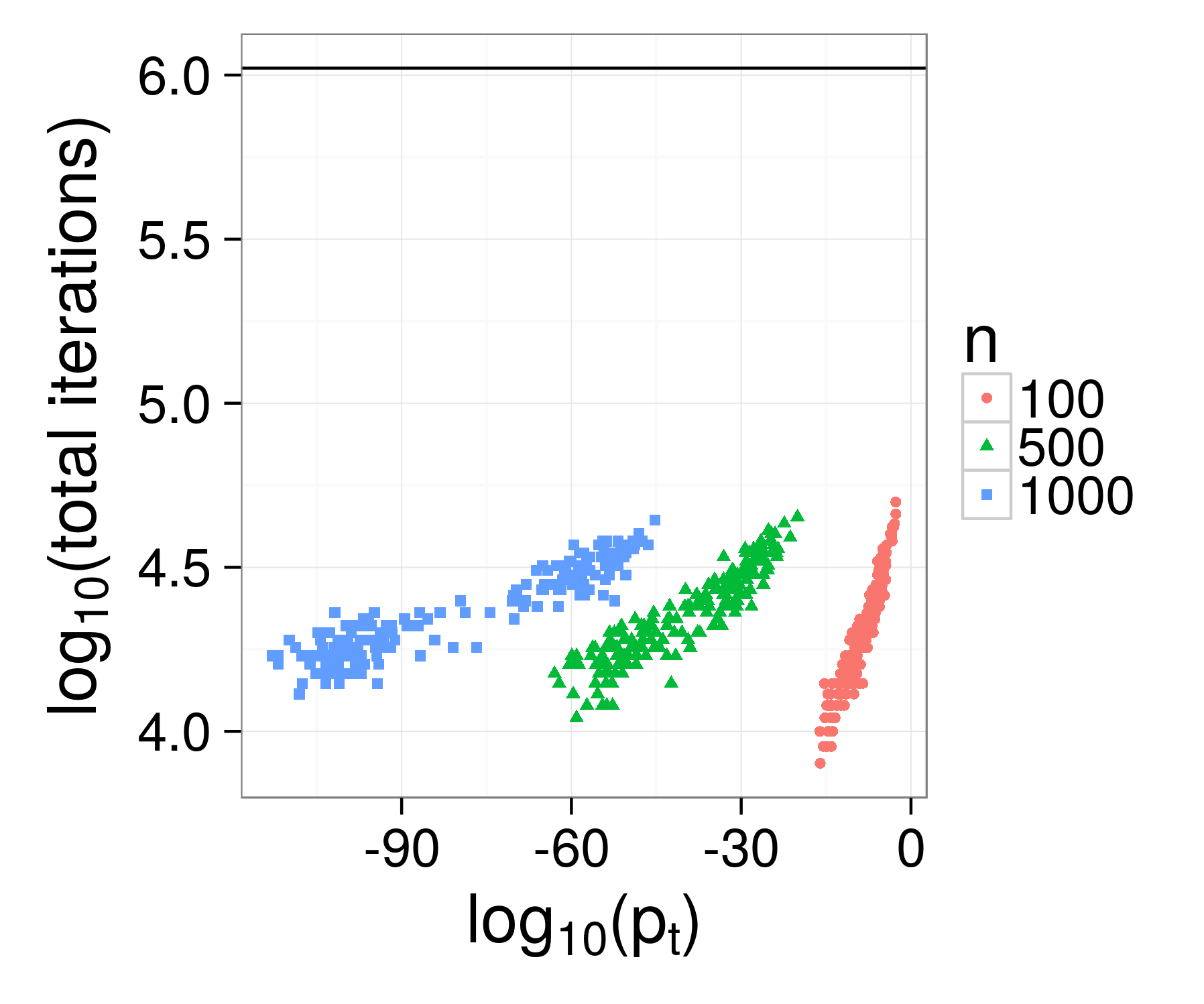}
  \caption{Iterations in resampling algorithm}
  \label{simDiff_sym_iter}
  \end{subfigure}
\caption{Simulation results using the statistic $T=|\bar{x}-\bar{y}|$, with equal sample sizes of $n=n_x=n_y=100$, 500, 1,000. \textit{Alg 1} is our resampling algorithm with $B_{\text{pred}}=10^3$ iterations in each partition, \textit{Asym} is our asymptotic approximation, \textit{SAMC} is the SAMC algorithm, and $p_t$ is a two-sided t-test with equal variance. The diagonal dashed line has slope of 1 and intercept of 0, and indicates agreement between methods. The horizontal line in \ref{simDiff_sym_iter} shows the number of permutations used in the SAMC algorithm (set in advance, and independent of p-value). The SAMC algorithm did not produce values for 385 tests (points missing).}
\label{sim_sym}
\end{figure}

\begin{figure}[htbp]
\centering
  \begin{subfigure}{0.58\textwidth}
  \centering
  \includegraphics[width=1\linewidth]{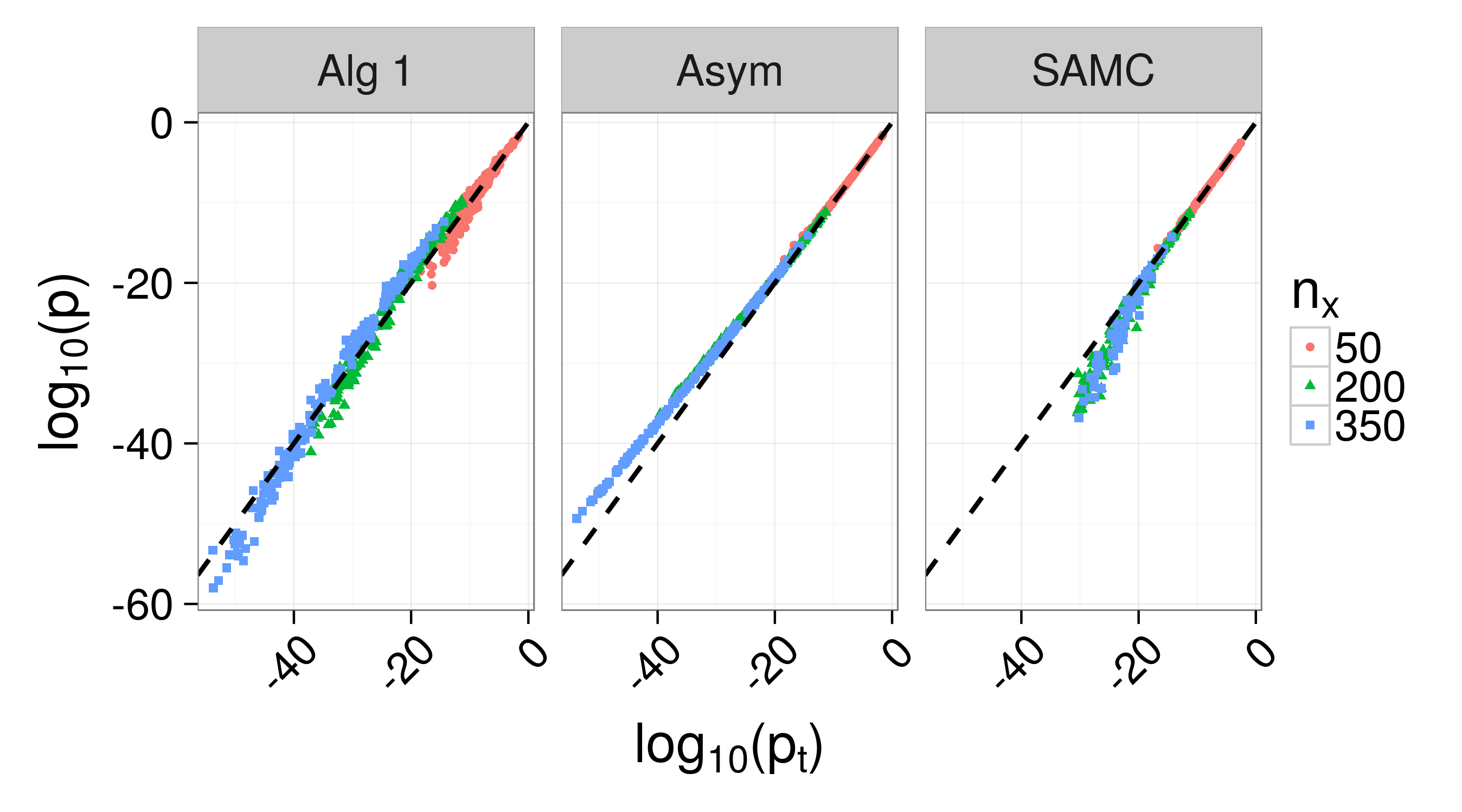}
  \caption{p-values}
  \end{subfigure}
  \begin{subfigure}{0.38\textwidth}
  \centering
  \includegraphics[width=1\linewidth]{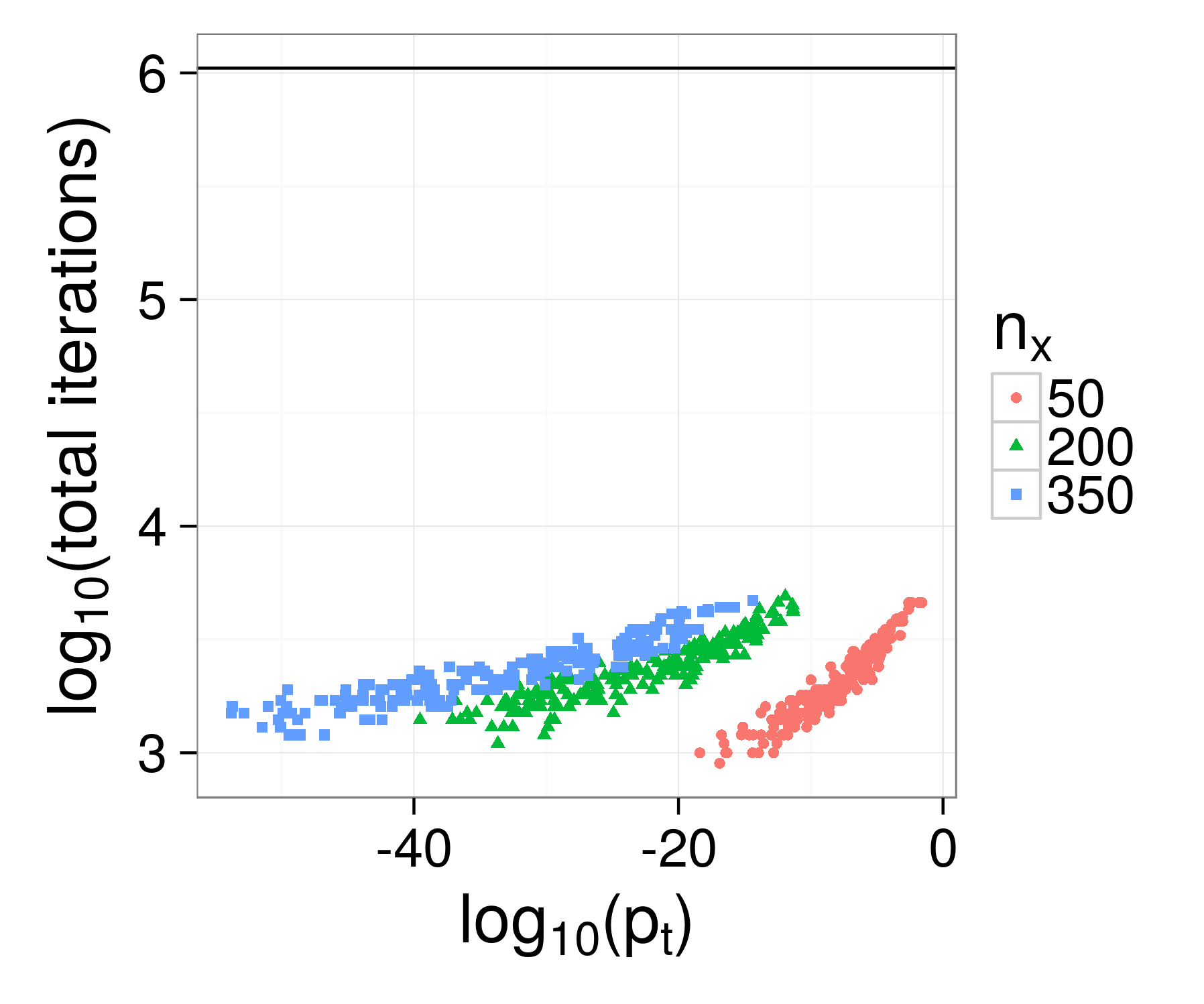}
  \caption{Iterations in resampling algorithm}
  \label{simDiff_nonSym_iter}
  \end{subfigure}
\caption{Simulation results using the statistic $T=|\bar{x}-\bar{y}|$, with unequal sample sizes, where $n_y=500$ and $n_x=50, 200, 350$. \textit{Alg 1} is our resampling algorithm with $B_{\text{pred}}=10^3$ iterations in each partition, \textit{Asym} is our asymptotic approximation, \textit{SAMC} is the SAMC algorithm, and $p_t$ is a two-sided t-test with equal variance. The diagonal dashed line has slope of 1 and intercept of 0, and indicates agreement between methods. The horizontal line in \ref{simDiff_nonSym_iter} shows the number of permutations used in the SAMC algorithm (set in advance, and independent of p-value). The SAMC algorithm did not produce values for 179 tests (points missing).}
\label{sim_nonSym}
\end{figure}

As Figures \ref{sim_sym} and \ref{sim_nonSym} show, our resampling algorithm and asymptotic approximation are able to estimate extremely small p-values, which the SAMC algorithm is not able to estimate even though we set it to use approximately two orders of magnitude more iterations than our resampling algorithm. While our asymptotic approximation has less variance than our resampling algorithm, the asymptotic approximation appears to have more bias. We note that the scale of the p-values is not the same in Figures \ref{sim_sym} and \ref{sim_nonSym}, but in both cases, they are smaller than what would typically be estimated with resampling methods. Figures \ref{sim_sym} and \ref{sim_nonSym} also show that p-values from the delta method (see Appendix \ref{G}) are not reliable, even for large sample sizes.

Figures \ref{simDiff_sym_iter} and \ref{simDiff_nonSym_iter} also demonstrate that our algorithm uses fewer permutations when estimating smaller p-values than when estimating larger p-values. This occurs because the trend in partition-specific p-values across the partitions tends to be steeper for smaller overall p-values, which leads to earlier stopping times.

\subsection{Ratio of means \label{sec_SAMC_compare}}

In this section, we consider the test statistic $T=\max(\bar{x}/\bar{y}, \bar{y}/\bar{x})$, both for $n_x=n_y$ and $n_x \ne n_y$. We generated data $x_i, i=1,\ldots,n_x$ and $y_j, j=1,\ldots,n_y$ as realizations of the respective random variables $X_i \overset{\text{iid}}{\sim} \text{Exp}(\lambda_x)$ and  $Y_j \overset{\text{iid}}{\sim} \text{Exp}(\lambda_y)$, where $\text{Exp}(\lambda)$ is an exponential distribution with rate $\lambda$, i.e. $\mathbb{E}[X_i]=1/\lambda_x$. We chose this setup because 1) having data with non-negative support ensures non-zero denominators in the ratio statistic, and 2) the resulting ratio statistic follows a beta prime distribution, also called a Pearson type VI distribution \citep[][p. 248]{johnson2002continuous}, which provides an approximate baseline for comparison (see Appendix \ref{B}).

For equal sample sizes, we set $n_x=n_y=100$, 500, or 1,000. For unequal sample sizes, we set $n_y=500$, and $n_x=50$, 200, or 350. In both cases we set $\lambda_x=1$ and $\lambda_y=1.75$ or 2.25. For all parameter combinations, we generated 100 datasets.

For each dataset, we applied our methods and computed the p-value from the beta prime distribution, using the \verb|PearsonDS| package for \textsf{R} \citep{Becker2016}. For our resampling algorithm, we used $B_{\text{pred}}=10^3$ iterations in each partition. We also computed p-values using the delta method (see Appendix \ref{G}), and ran the SAMC algorithm, with the same specifications as described in Section \ref{sec_algoEval}.

Results are shown in Figures \ref{simExp_sym} and \ref{simExp_nonSym}. In the Figures, $p_{\beta}$ denotes the p-value from the beta prime distribution, and $p$ denotes the p-value from either our methods, the delta method (see Appendix \ref{G}), or SAMC. The dashed line has a slope of 1 and intercept of 0, and indicates agreement between methods. As before, the SAMC algorithm did not produce values for smaller p-values, and so these points are missing from Figures \ref{simExp_sym} and \ref{simExp_nonSym} (246 missing points in Figure \ref{simExp_sym}, and 33 missing points in Figure \ref{simExp_nonSym}).

\begin{figure}[htbp]
\centering
  \begin{subfigure}{0.58\textwidth}
  \centering
  \includegraphics[width=1\linewidth]{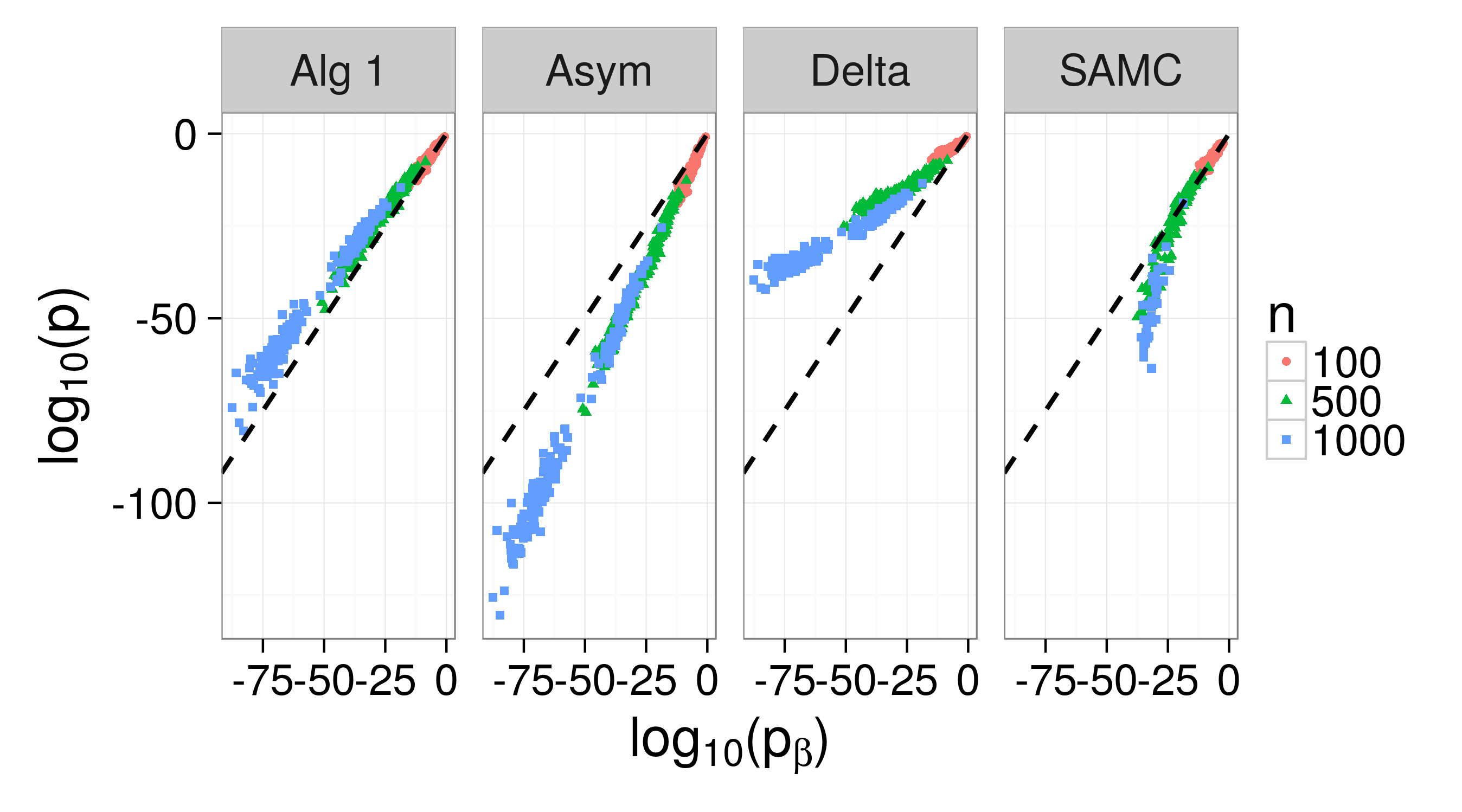}
  \caption{p-values}
  \end{subfigure}
  \begin{subfigure}{0.38\textwidth}
  \centering
  \includegraphics[width=1\linewidth]{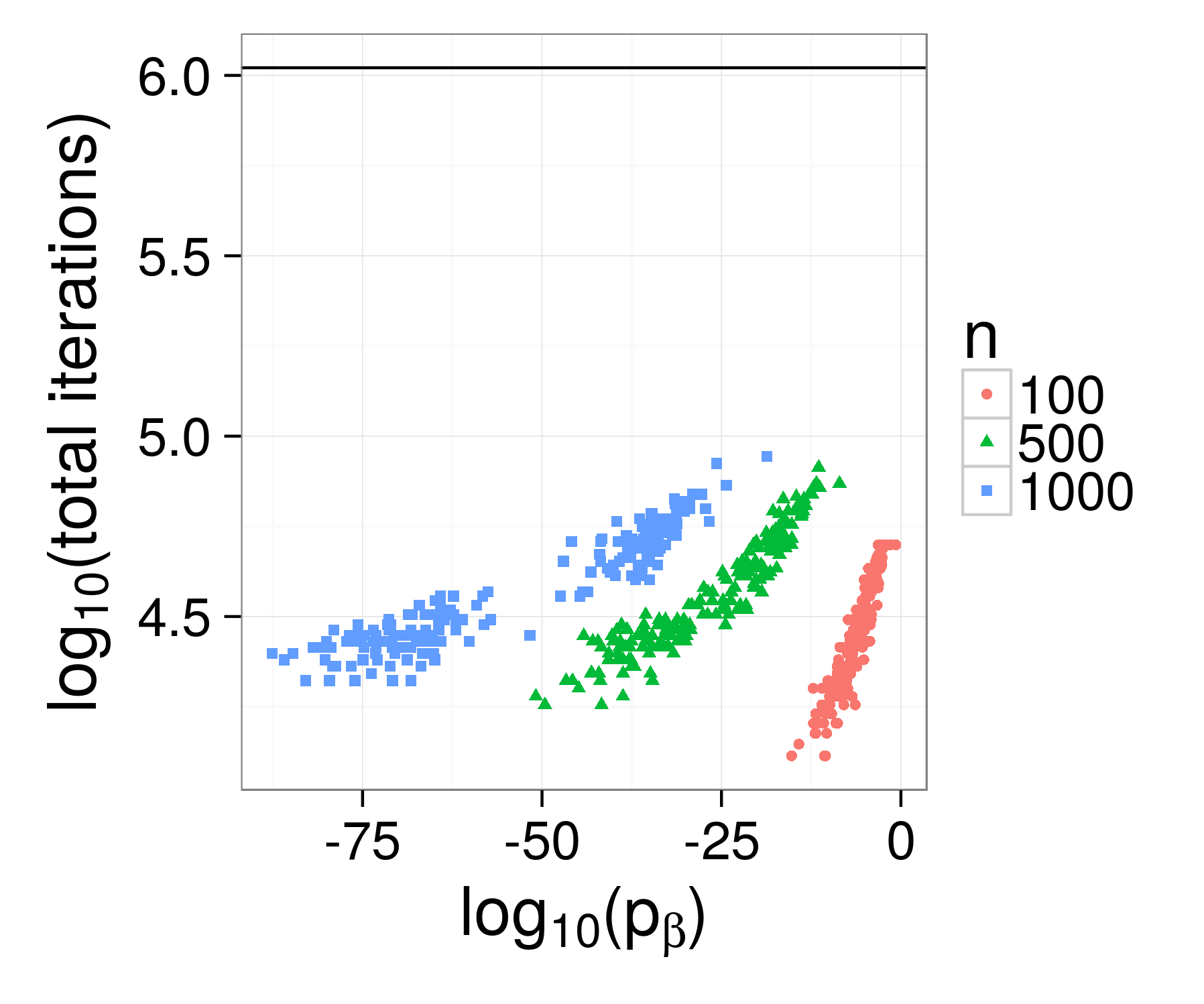}
  \caption{Iterations in resampling algorithm}
  \label{simExp_sym_iter}
  \end{subfigure}
\caption{Simulation results using the statistic $T= \max(\bar{x} / \bar{y}, \bar{y} / \bar{x})$, with equal sample sizes of $n=n_x=n_y=100$, 500, 1,000. \textit{Alg 1} is our resampling algorithm with $B_{\text{pred}}=10^3$ iterations in each partition, \textit{Asym} is our asymptotic approximation, \textit{Delta} is the delta method, and \textit{SAMC} is the SAMC algorithm, and $p_{\beta}$ is the two-sided p-value from the beta prime distribution. The diagonal dashed line has slope of 1 and intercept of 0, and indicates agreement between methods. The horizontal line in \ref{simExp_sym_iter} shows the number of permutations used in the SAMC algorithm (set in advance, and independent of p-value). The SAMC algorithm did not produce values for 246 tests (points missing).}
\label{simExp_sym}
\end{figure}

\begin{figure}[htbp]
\centering
  \begin{subfigure}{0.58\textwidth}
  \centering
  \includegraphics[width=1\linewidth]{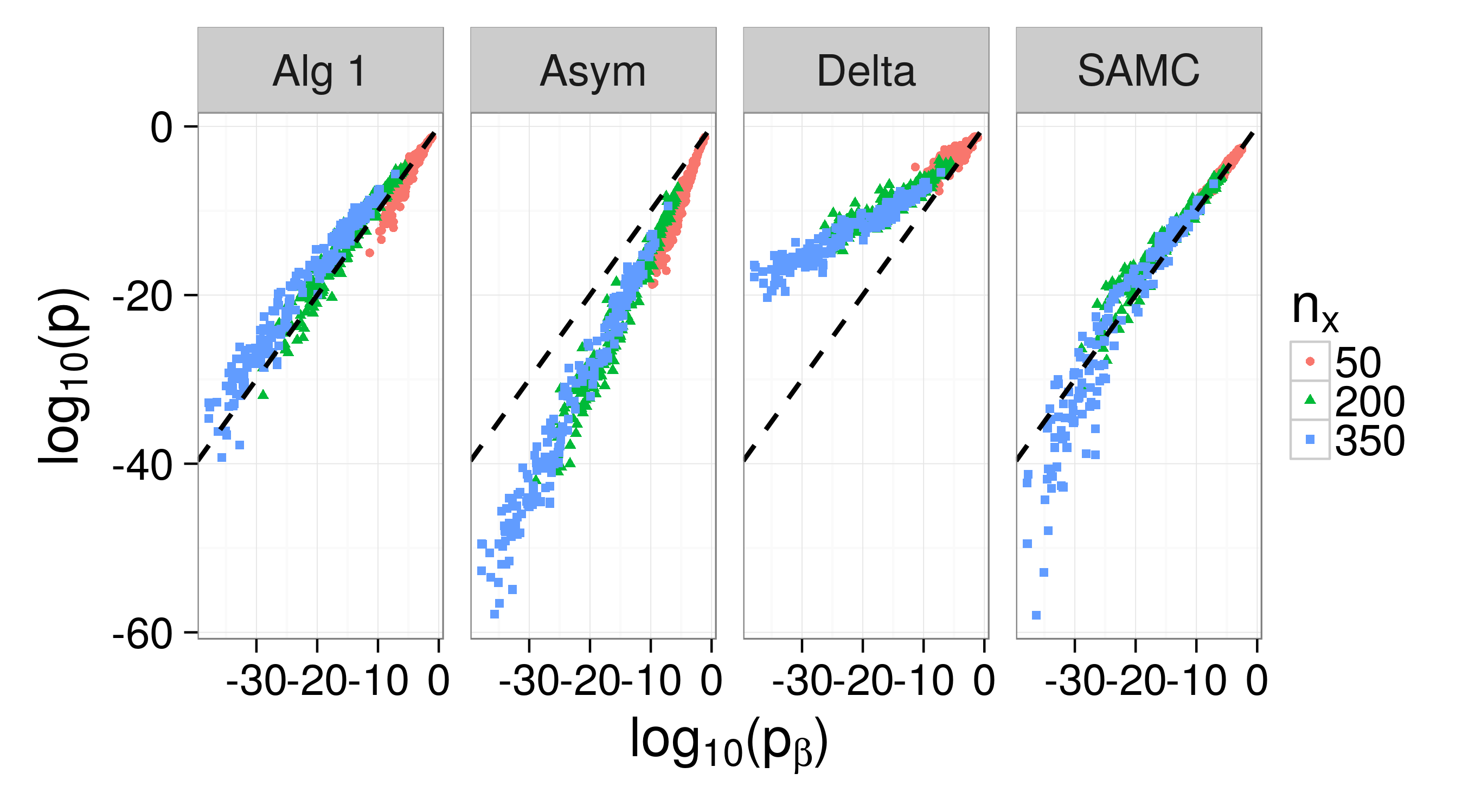}
  \caption{p-values}
  \end{subfigure}
  \begin{subfigure}{0.38\textwidth}
  \centering
  \includegraphics[width=1\linewidth]{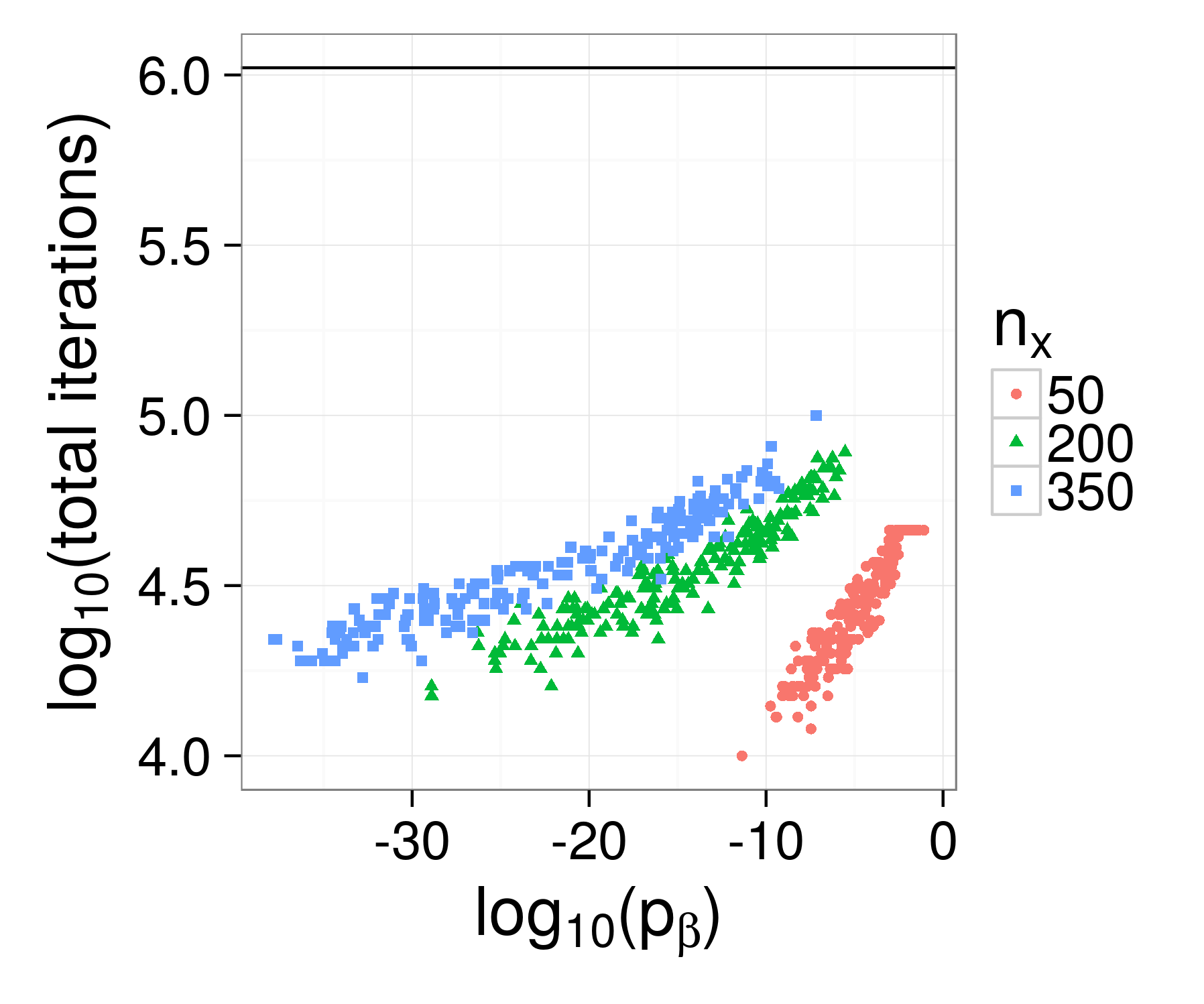}
  \caption{Iterations in resampling algorithm}
  \label{simExp_nonSym_iter}
  \end{subfigure}
\caption{Simulation results using the statistic $T= \max(\bar{x} / \bar{y}, \bar{y} / \bar{x})$ and unequal sample sizes, where $n_y=500$ and $n_x=50, 200, 350$. \textit{Alg 1} is our resampling algorithm with $B_{\text{pred}}=10^3$ iterations in each partition, \textit{Asym} is our asymptotic approximation, \textit{Delta} is the delta method, and \textit{SAMC} is the SAMC algorithm, and $p_{\beta}$ is the two-sided p-value from the beta prime distribution. The diagonal dashed line has slope of 1 and intercept of 0, and indicates agreement between methods. The horizontal line in \ref{simExp_nonSym_iter} shows the number of permutations used in the SAMC algorithm (set in advance, and independent of p-value). The SAMC algorithm did not produce values for 33 tests (points missing).}
\label{simExp_nonSym}
\end{figure}

As Figures \ref{simExp_sym} and \ref{simExp_nonSym} show, both our resampling algorithm and asymptotic approximation appear to have more bias in this setting than for the difference in means, though in this case, the asymptotic approximation is biased downward instead of upward. Our resampling algorithm tends to be biased upward for equal group sizes ($n_x = n_y$), and downward for highly imbalanced group sizes (e.g. $n_x=50$ and $n_y=500$).

As before, the SAMC algorithm had trouble estimating extremely small p-values with the number of iterations we allowed it. In the case of the equal sample size simulation, the SAMC algorithm began to have problems for p-values around $10^{-30}$. In the case of unequal sample size, the SAMC algorithm appears to have performed similarly to our resampling algorithm, albeit with one to two orders of magnitude more iterations.

Similar to Section \ref{sec_algoEval}, Figures \ref{simExp_sym_iter} and \ref{simExp_nonSym_iter} show that our resampling algorithm uses fewer iterations for smaller p-values. Also, as before, the scale of the p-values is not the same in Figures \ref{simExp_sym} and \ref{simExp_nonSym}, but in both cases, they are smaller than what would typically be estimated with resampling methods.

\section{Application to cancer genomic data \label{sec_app}}

To further demonstrate our methods, we analyzed RNA-seq data collected as part of The Cancer Genome Atlas (TCGA) \citep{TCGA2015}. In particular, we were interested in identifying genes that were differentially expressed in two different types of lung cancers: lung adenocarcinoma (LUAD), and lung squamous cell carcinoma (LUSC).

We downloaded normalized gene expression data from the TCGA data portal\\
(\verb|https://tcga-data.nci.nih.gov/tcga|). As described by TCGA, to produce the normalized gene expression data, tissue samples from patients with LUSC and LUAD were sequenced using the Illumina RNA Sequencing platform. The raw sequencing reads from all patient samples were processed and analyzed using the SeqWare Pipeline 0.7.0 and MapspliceRSEM workflow 0.7 developed by the University of North Carolina. Sequencing reads were aligned to the human reference genome using MapSplice \citep{wang2010mapsplice}, and gene level expression values were estimated using RSEM \citep{li2011rsem} with gene annotation file GAF 2.1. For each sample, RSEM gene expression estimates were normalized to set the upper quartile count at 1,000 for gene level estimates. For the analyses in this section, we used the normalized RSEM gene expression estimates.

For both LUAD and LUSC, TCGA contains normalized expression estimates for 20,531 genes (the same genes for both cancers). There were 548 subjects with LUAD observations, and 541 with LUSC observations. To ensure that our results would be biologically meaningful, we restricted our analysis to genes for which at least 50\% of the subjects had expression levels above the $25^{th}$ percentile of all normalized gene expression levels (6.57). This reduced our analysis to 15,386 genes.

Let $P_{x,g}$ and $P_{y,g}$ be the underlying distributions that generated the normalized expression levels in LUAD and LUSC, respectively, for gene $g$. To test the two-sided hypothesis of $H_0: P_{x,g} = P_{y,g}$ versus the alternative $H_1: \mu_x / \mu_y \ne 1$, we used the fold-change statistic $T= \max(\bar{x}_g/\bar{y}_g,\bar{y}_g/\bar{x}_g)$. Here, $\mu_x$ and $\mu_y$ are the means of $P_{x,g}$ and $P_{y,g}$, respectively.

First, we conducted simple Monte Carlo permutation tests on all 15,386 genes with $B=10^3$ iterations. This left us with 10,302 genes with p-values less than $10^{-3}$, the minimum estimate possible with only $B=10^3$ iterations. We then used our resampling algorithm to estimate p-values for the 10,302 genes that passed our preliminary screen.

Figure \ref{pw_hist} shows the distribution of the resulting p-values. The dashed red line indicates the cutoff value from the preliminary screen ($10^{-3}$). Figure \ref{p_all_hist} shows all 15,386 p-values, where the p-value is taken from the initial screen if the p-value was larger than $10^{-3}$ and from our algorithm otherwise. The non-uniform shape of Figure \ref{p_all_hist} provides strong evidence against the null hypothesis of no differential expression.

\begin{figure}[htbp]
\centering
  \begin{subfigure}{0.35\textwidth}
  \centering
  \includegraphics[width=1\linewidth]{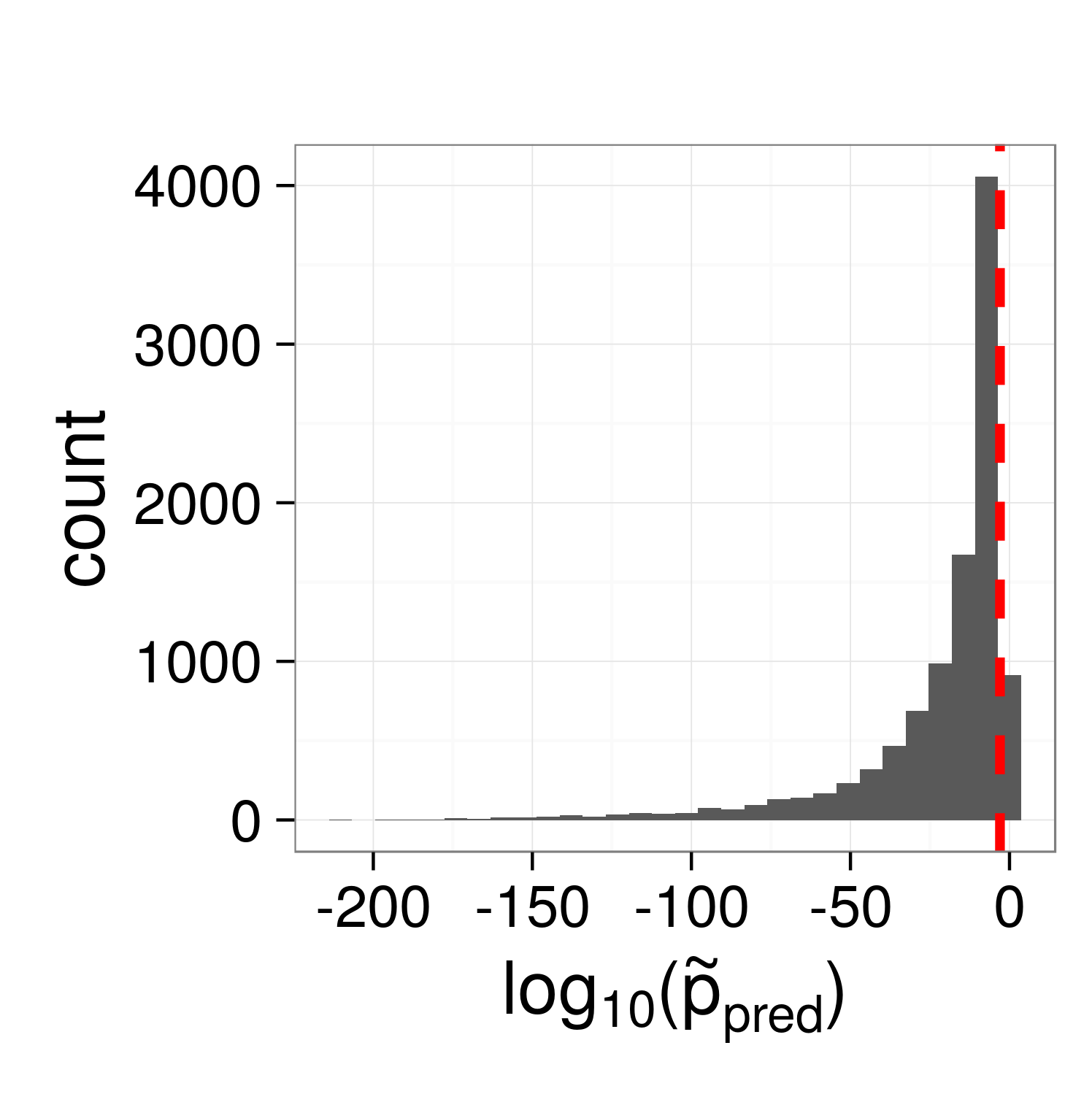}
  \caption{Genes with $\tilde{p} \le 10^{-3}$}
  \label{pw_hist}
  \end{subfigure}
  \begin{subfigure}{0.35\textwidth}
  \centering
  \includegraphics[width=1\linewidth]{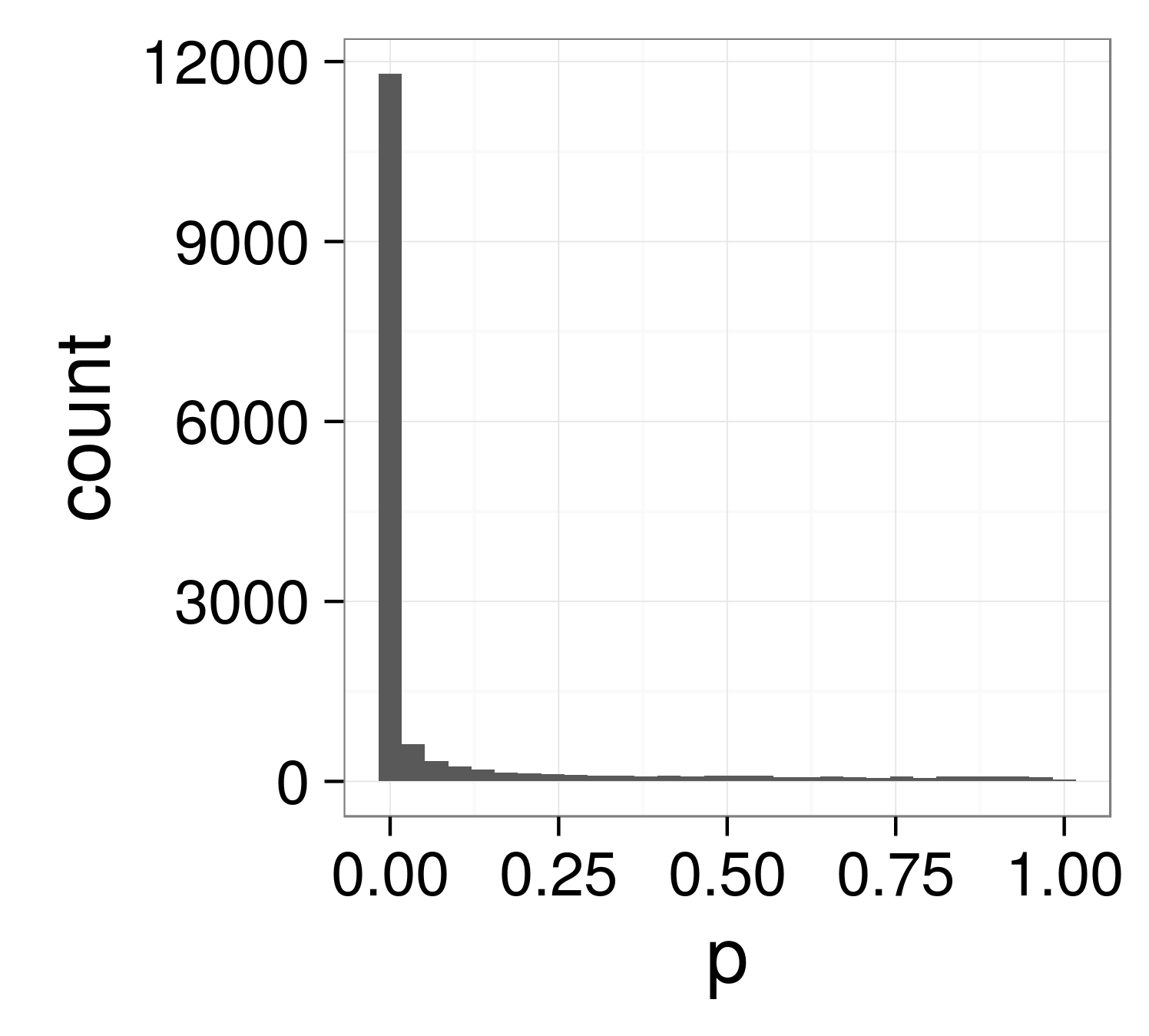}
  \caption{All genes}
  \label{p_all_hist}
  \end{subfigure}
\caption{Histogram of $\tilde{p}_{\text{pred}}$ with $B_{\text{pred}}=10^3$ iterations in each partition from the 10,302 genes that passed the initial screen (log scale), and of all 15,386 p-values (original scale, with values from either our resampling algorithm or the initial screen, depending on the size of the p-value). The dashed red line in \ref{pw_hist} indicates the cutoff from the preliminary screen ($10^{-3}$).}
\label{p_histograms}
\end{figure}

We do not show results with the asymptotic approximation or the beta prime distribution, but we note that the results from the asymptotic approximation were similar to those from the resampling algorithm, though as in Section \ref{sec_SAMC_compare}, $\hat{p}_{\text{asym}}$ tended to be smaller than $\tilde{p}_{\text{pred}}$. The results from the beta prime distribution were not similar to those from the resampling algorithm, which is not surprising, since we do not expect the normalized expression levels to follow an exponential distribution. Results using the delta method are shown in Appendix \ref{G}, and appear to have a similar bias as in the simulations.

Figure \ref{cancer_iter} shows the total number of iterations that our algorithm used for each test, and Figure \ref{mStopPlot} compares $m_{\text{stop}}^{\text{asym}}$, which can be computed beforehand, with the actual stopping partitions $m_\text{stop}$. In this analysis $m_{\text{stop}}^{\text{asym}}$ appears to be biased upward, but we think that it is a reasonable approximation of $m_{\text{stop}}$ for the purposes of obtaining a general estimate of computing time before running the resampling algorithm.

\begin{figure}[htbp]
\centering
  \begin{subfigure}{0.35\textwidth}
  \centering
  \includegraphics[width=1\linewidth]{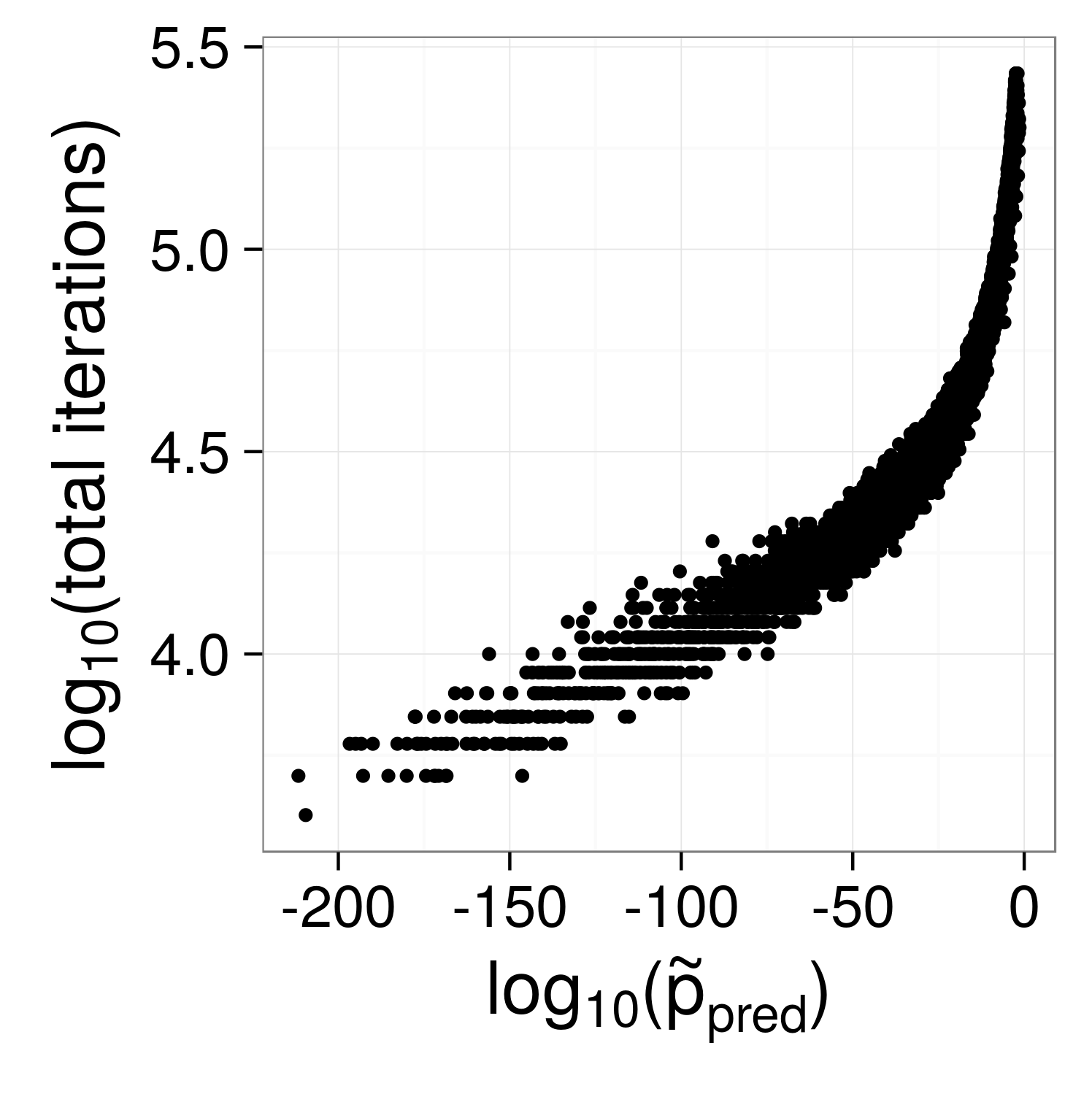}
  \caption{Total iterations}
  \label{cancer_iter}
  \end{subfigure}
  \begin{subfigure}{0.35\textwidth}
  \centering
  \includegraphics[width=1\linewidth]{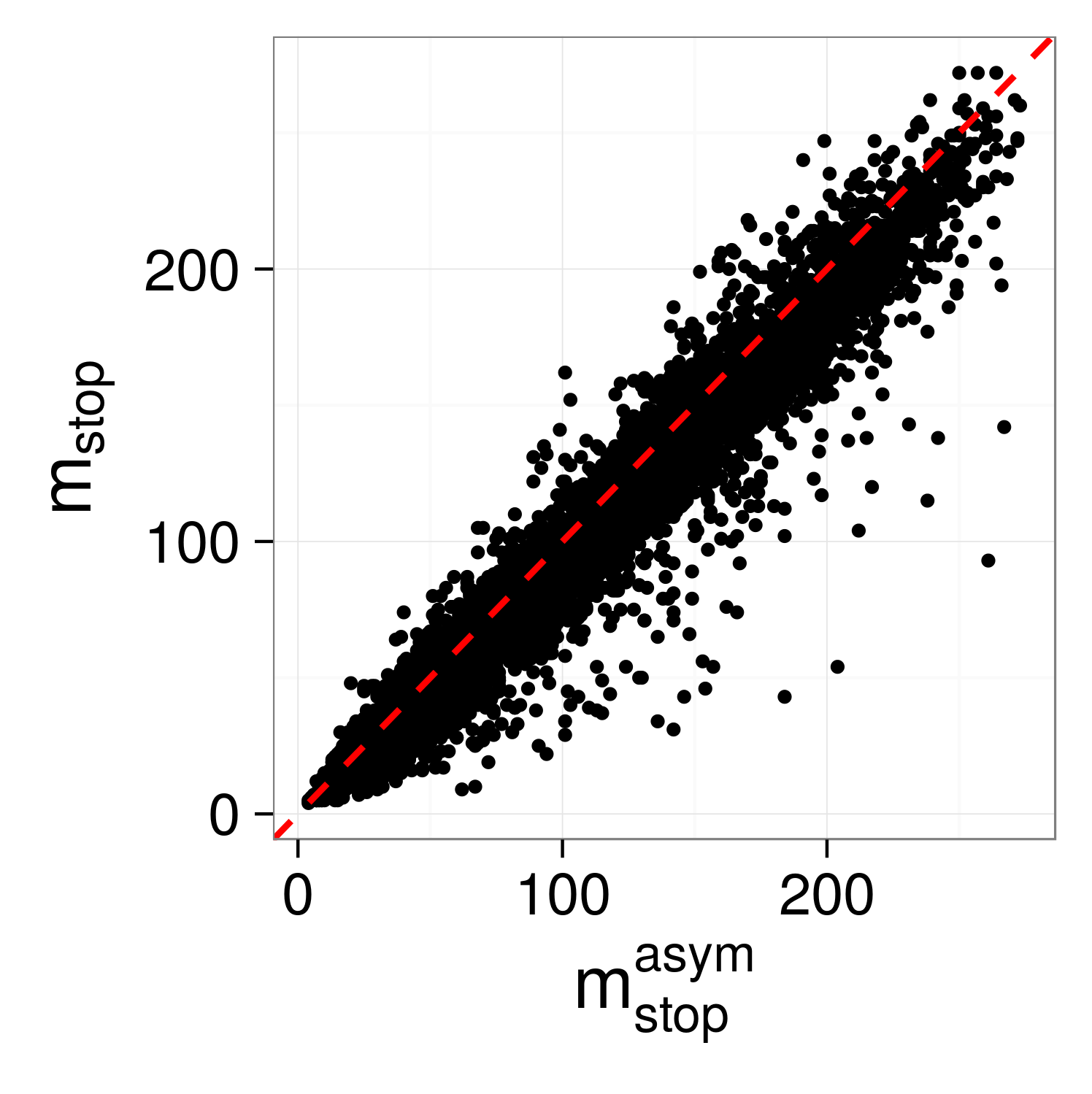}
  \caption{$m_{\text{stop}}^{\text{asym}}$ vs. $m_{\text{stop}}$}
  \label{mStopPlot}
  \end{subfigure}
\caption{Total number of iterations to conduct each test, and comparison between $m_{\text{stop}}^{\text{asym}}$ and $m_{\text{stop}}$. $m_{\text{stop}}$ is the actual stopping partition, which our resampling algorithm determines dynamically. $m_{\text{stop}}^{\text{asym}}$ is our estimate of the stopping partition based on asymptotic approximations, and can be computed before running the algorithm.}
\label{cancer_iter_plots}
\end{figure}

Table \ref{cancerResults} shows the results for the fifteen genes with the smallest p-values, as well as the deviance and AIC from the Poisson regression fit during the resampling algorithm. We report both the estimate from the initial, single run of our algorithm, as well as the $10^{th}, 50^{th}$, and $90^{th}$ quantiles from an additional 1,000 runs. Note that Table \ref{cancerResults} reports the observed ratio of mean(LUAD)/mean(LUSC), and not the max of the ratios that we used in the permutation test. Of the top 15 genes, none had elevated levels in LUAD. Point estimates for all genes are available as supplementary material.

Eleven of the these fifteen genes, shown in bold (\textit{DSG3, KRT5, DSC3, CALM3, TP63, ATP1B3, KRT6B, TRIM29, PVRL1, FAT2}, and \textit{KRT6C}), were also identified by \citet{zhan2015} as being among the most effective genes for distinguishing between LUAD and LUSC. Like us, \citet{zhan2015} used the TCGA dataset, though they based their analysis on the area under the curve from a Wilcoxon rank-sum test. 

\begin{table}[htbp]
\centering
\caption{Fifteen genes with the smallest p-values, and other output from our algorithm with $B_{\text{pred}}=10^3$ iterations in each partition. \textit{Single run} is the value of $\log_{10}(\tilde{p}_{\text{pred}})$ from the initial run of our resampling algorithm, and the quantiles are from 1,000 replicates. For the single run, $m_\text{stop}$ is the partition at which our algorithm stopped, and deviance and AIC are from the Poisson regression fit during the algorithm. Genes shown in bold were identified by \citet{zhan2015} as being among the most effective genes for distinguishing between LUAD and LUSC using the area under the curve from a Wilcoxon rank-sum test.}
\begin{tabular}{ccccccc}
\hline \hline
& \multicolumn{2}{c}{$\log_{10}(\tilde{p}_{\text{pred}})$} &\\
\cline{2-3}
Gene name & Single run & Quantiles ($10^{th}, 50^{th}, 90^{th}$) & $\frac{\text{mean(LUAD)}}{\text{mean(LUSC)}}$ & $m_\text{stop}$ & Deviance & AIC \\
\hline
\textbf{\textit{DSG3}} & -212 & (-217, -208, -200) & 0.0100 & 5 & 40.1 & 68.1 \\
\textbf{\textit{KRT5}} & -210 & (-223, -214, -205) & 0.0107 & 4 & 12.5 & 38.2 \\
\textbf{\textit{DSC3}} & -197 & (-212, -205, -197) & 0.0175 & 6 & 41.5 & 72.1 \\
\textbf{\textit{CALML3}} & -195 & (-198, -188, -179) & 0.0138 & 6 & 57.8 & 90 \\
\textbf{\textit{TP63}} & -193 & (-199, -192, -186) & 0.0308 & 6 & 24.2 & 55.1 \\
\textbf{\textit{ATP1B3}} & -193 & (-196, -188, -181) & 0.225 & 5 & 28.6 & 57.7 \\
\textit{S1PR5} & -190 & (-190, -181, -173) & 0.0775 & 6 & 98.4 & 131 \\
\textbf{\textit{KRT6B}} & -185 & (-189, -181, -173) & 0.0173 & 5 & 45.4 & 76.1 \\
\textbf{\textit{TRIM29}} & -183 & (-188, -181, -174) & 0.0788 & 6 & 39.3 & 72 \\
\textit{JAG1} & -180 & (-186, -179, -172) & 0.170 & 5 & 60.7 & 92.2 \\
\textbf{\textit{PVRL1}} & -180 & (-183, -177, -171) & 0.110 & 6 & 8.33 & 39.2 \\
\textit{CLCA2} & -178 & (-188, -180, -172) & 0.0138 & 7 & 51.6 & 86.8  \\
\textit{BNC1} & -178 & (-197, -188, -181) & 0.0244 & 7 & 76.8 & 112  \\
\textbf{\textit{FAT2}} & -177 & (-186, -179, -173) & 0.0339 & 7 & 53.5 & 89  \\
\textbf{\textit{KRT6C}} & -177 & (-188, -181, -174) & 0.0183 & 6 & 84.8 & 119
\end{tabular}
\label{cancerResults}
\end{table}

We emphasize that in presenting Table \ref{cancerResults}, we are not trying to promote the use of p-values as the sole source of information for making scientific decisions, such as ranking the importance of genes. Instead, we present Table \ref{cancerResults} and make comparisons with the findings of \citet{zhan2015} as a way of verifying the reasonableness of our results. \citet{zhan2015} used different methods to analyze the TCGA data, so we do not expect our results to be exactly the same, but it is encouraging that our results appear to agree to some extent.

We also want to point out that our resampling algorithm can approximate extremely small p-values, but that in doing so, there is a large amount variability in the estimates. However, we think these estimates could still be used as an approximation of the order of magnitude, and note that they would be infeasible to estimate with existing Monte Carlo methods, including the SAMC algorithm. 

\section{Discussion \label{discussion}}

As we have demonstrated through simulations and an application to cancer genomic data, our methods can quickly approximate small permutation p-values (e.g. $<10^{-6}$) for two-sample tests, where the test statistic is the difference or ratio of means. The computational efficiency of our resampling algorithm is particularly notable when estimating extremely small p-values, (e.g. $<10^{-30}$).

As is suggested in the example of Section \ref{sec_partPermSpac}, our method can only detect mean shifts.

As shown in the Simulations and Appendices, the accuracy of our resampling method is comparable to alternative methods, such as SAMC and MCC, though SAMC and MCC are applicable in situations where our methods are not. In particular, MCC can handle any statistic that can be expressed as, or is permutationally equivalent to, an inner product. In addition to these methods, researchers may want to consider the method of \citet{fieller1954} for obtaining confidence intervals for the ratio of means, and the approaches described by \citet{cui2003statistical} for using t-tests and ANOVA to analyze the mean log ratio.

While the reliability of our resampling algorithm will vary based on the empirical distribution of the data, in general, we recommend having at least 15-20 observations in each group for p-values near $1 \times 10^{-6}$, and at least 70-90 observations in each group for p-values near $1 \times 10^{-30}$ (see Appendix \ref{F}). As demonstrated in Section \ref{sec_app}, there can be considerable variability in estimating extraordinarily small p-values, e.g. $1 \times 10^{-200}$. For these extraordinarily small p-values, we recommend that our method be used only to approximate the order of magnitude of the permutation p-value.

In choosing between our resampling algorithm and asymptotic approximation, we recommend using the resampling algorithm when possible for small p-values, as it appears to perform better in simulations. However, as demonstrated in the appendix, our asymptotic method may be preferable for large p-values, as it appears to be more conservative under the null. Both approaches work best for equal sample sizes, and we suggest caution when using with small and highly imbalanced samples.

Depending on a researcher's needs, our algorithm could be useful as a fast approximation of small p-values. This might be helpful, for example, in a screening study involving many genes, in which a researcher wants to quickly get a sense for which genes have p-values that are likely to be below a small threshold. It might also be helpful as a preliminary analysis to approximate the order of magnitude of a p-value, which could help a researcher to determine whether it would be feasible to follow-up with other Monte Carlo methods, such as SAMC, and if so, how many iterations they would need to use. For some situations, such as our analysis in Section \ref{sec_app}, this could save considerable time and resources.

We want to emphasize that our methods are most useful for approximating small permutation p-values. For large p-values, our resampling algorithm is less computationally efficient than simple Monte Carlo sampling. In the context of genomics data, before using our methods, we recommend that researchers use simple Monte Carlo resampling with a small number of resamples (e.g. $10^3$) to identify which genes have p-values below a certain threshold (e.g., $10^{-3}$). However, this is not a requirement.

This paper focuses on two-sample tests, and we plan to explore extensions to multiple samples in future work. As one way to handle multiple samples, we could conduct a union-intersection test \citep[][p. 380]{casella2002statistical}. For example, say we have $k$ samples $\bm{x}_1, \ldots, \bm{x}_k$, and we wish to test the hypothesis $H_0: \cap_{i \ne j}P_{x_i} = P_{x_j}$ versus the alternative $H_1: \cup_{i\ne j} \mu_{x_i} \ne \mu_{x_j}$, where $\mu_i$ is the mean of $P_{x_i}$. Then we could use Algorithm \ref{alg_two_sided_general} to compute p-values for all pairwise differences (or all pairwise ratios), and then take the minimum p-value. As another alternative, we could extend Algorithm \ref{alg_two_sided_general} to use an omnibus statistic, similar to the ANOVA F-test, and use a multi-sample version of (\ref{pmf}). For example, we might use $T = \sum_i n_i |\bar{x}_i - \bar{x} | / n$ where $\bar{x}_i$ and $n_i$ are the mean and sample size, respectively, for group $i$, $\bar{x}$ is the overall mean, and $n = \sum_i n_i$. However, the extension of (\ref{pmf}) to multiple samples is non-trivial. It is also unclear whether the p-values from the multi-sample case would follow the same trends across the partitions as in the two-sample case.

Returning to the two-sample case, while we have focused on the difference and ratio of the means, preliminary efforts to explain the nearly log-linear trend in p-values across the partitions suggests that the same pattern might hold for other smooth functions of the means. In future work, we plan to explore this further. We also plan to investigate potential diagnostics for assessing the reliability of the algorithm's output, possibly based on the AIC from the Poisson regression. Finally, we note that alternative Monte Carlo methods could be incorporated into our resampling algorithm. For example, the SAMC algorithm could be used in place of simple Monte Carlo resampling within each partition. This might further reduce run-time and increase accuracy.

\section*{Supplementary material}

We have implemented our method in the \text{R} package \verb|fastPerm|, available at\\
 \url{https://github.com/bdsegal/fastPerm}. All code for the simulations and analyses in this paper will be available at \url{https://github.com/bdsegal/code-for-fastPerm-paper}.

\section*{Acknowledgments}
We would like to thank the associate editor and two referees for their insightful comments and suggestions.

\bibliographystyle{plainnat}
\bibliography{fastPerm_bib}

\begin{appendices}

\renewcommand{\thetable}{S\arabic{table}}   
\renewcommand{\thefigure}{S\arabic{figure}}
\setcounter{figure}{0}    
\setcounter{table}{0}    

\section{Proofs \label{proofs}}

In this appendix, we find the limiting distribution of $T = \max(\bar{x}/\bar{y}, \bar{y}/\bar{x})$ and $T = |\bar{x} - \bar{y}|$ within each partition, and note the corresponding trend in p-values across the partitions. In the process, we prove the results discussed in Section 3. We structure this appendix around the statistic $T=\max(\bar{x}/\bar{y}, \bar{y}/\bar{x})$ to help to motivate our discussion, and then extend our results to the statistic $T = |\bar{x} - \bar{y}|$.

As before, we denote the total sample size as $N$, and we require that $N \ge 2$ to allow for at least one observation in each sample. Let $\{ \mN \}_{N=2}^{\infty}$, $\{\nxN\}_{N=2}^{\infty}$, and $\{\nyN\}_{N=2}^{\infty}$, be sequences, such that $\mN/N \rightarrow \tau$ and $\nxN / N \rightarrow \lambda$ as $N \rightarrow \infty$, and for all $N$, $\nyN = N - \nxN$. We require that, for all $N$, $0 < \mN \le \nxN \le \nyN < N$, and similarly, $0 < \tau \le \lambda \le 1-\lambda < 1$. We denote the observed data as $\bxN$ and $\byN$, which are $\nxN \times 1$ and $\nyN \times 1$ vectors, respectively.

Let $\bd^{\mN}_x = (\delta^{\mN}_{x,1}, \ldots, \delta^{\mN}_{x, \nxN})'$ and $\bd^{\mN}_y = (\delta^{\mN}_{y,1}, \ldots, \delta^{\mN}_{y, \nyN})'$ be $\nxN \times 1$ and $\nyN \times 1$ indicator vectors, respectively, with 1's corresponding to indices of $\bxN$ and $\byN$ that are exchanged for a particular permutation $\pi$ and zero elsewhere. To be specific, for a permutation $\pi \in \Pi(\mN)$, we define $\delta^{\mN}_{x,i}$ and $\delta^{\mN}_{y,j}$ as
\begin{align*}
\delta^{\mN}_{x,i} &=
\begin{cases}
1 \text{ if } \pi(i) >\nxN \\
0 \text{ if } \pi(i) \le \nxN
\end{cases} && i=1,\ldots,\nxN \\
\delta^{\mN}_{y,j} &=
\begin{cases}
1 \text{ if } \pi(\nxN + j) \le \nxN \\
0 \text{ if } \pi(\nxN + j) > \nxN
\end{cases} &&j=1, \ldots, \nyN.
\end{align*}
For completeness, we note that for fixed $m$ and $i \ne j$, and dropping dependence on $N$,
\begin{align*}
\mathbb{E}[\delta^m_{x,i}] &= m/n_x && \mathbb{E}[\delta^m_{y,i}] = m/n_y\\
\Var(\delta^m_{x,i}) &= \frac{m}{n_x} \left(1 - \frac{m}{n_x} \right)  && \Var(\delta^m_{y,i}) = \frac{m}{n_y} \left(1 - \frac{m}{n_y} \right)\\
\Cov(\delta^m_{x,i}, \delta^m_{x,j}) &= \frac{-m(n_x-m)}{n_x^2(n_x-1)} && \Cov(\delta^m_{y,i}, \delta^m_{y,j}) = \frac{-m(n_y-m)}{n_y^2(n_y-1)}
\end{align*}

We denote the ratio of means as $R=\bar{x}/\bar{y}$. With the permutation test, for each permutation $\pi$ in partition $\mN$, we calculate the statistic (ignoring, for now, the max function used earlier)
$$
\Rm = \frac{\frac{1}{\nxN} [ (\bOne-\bdxN)'\bxN + {\bdyN}' \byN]}{\frac{1}{\nyN}[{\bdxN}' \bxN + (\bOne-\bdyN)' \byN]}.
$$
As for all permutation tests, $\Rm$ is conditional on the data. The random quantities are $(\bdxN, \bdyN$), which indexed by $N$, form a triangular array of identically distributed, dependent random variables. We can rewrite $\Rm$ as
\begin{align}
\Rm
&=\frac{\nyN}{\nxN}\left( \frac{ \nxN \bar{x} + \left( \sum_{j=1}^{\nyN} \dN_{y,j} \yN_j -\sum_{i=1}^{\nxN} \dN_{x,i} \xN_i \right)}{ \nyN \bar{y} -\left( \sum_{j=1}^{\nyN} \dN_{y,j} \yN_j - \sum_{i=1}^{\nxN} \dN_{x,i} \xN_i  \right)} \right) \nonumber \\
&= g \underset{\W}{\left( \underbrace{ \sum_{j=1}^{\nyN} \dN_{y,j} \yN_j - \sum_{i=1}^{\nxN} \dN_{x,i} \xN_i } \right)} . \label{gTmax}
\end{align}

Writing $\Rm$ as a function of $\W$ will make it straightforward to generalize our results. We note that conditional on the observed data $\bxN$ and $\byN$, all terms in $\Rm$ are constant except for $\W$.

We can further split $\W$ into
\begin{equation}
\W = \underset{\Wy}{\underbrace{\sum_{j=1}^{\nyN} \dN_{y,j} \yN_j}} - \underset{\Wx}{\underbrace{\sum_{i=1}^{\nxN} \dN_{x,i} \xN_i}} \label{T_scaled}
\end{equation}
Following Theorem 2.8.2 in \citet[][p. 116]{lehmann1999}, restated in Theorem \ref{Wlem} below, under certain conditions both $\Wy$ and $\Wx$ in (\ref{T_scaled}) converge to normal random variables, in which case $\W$ also converges to a normal random variable.

We make a few observations before stating Theorem \ref{Wlem}. The following statements focus on $\Wy$, but equivalent statements apply to $\Wx$. First, we note that conditional on $\byN$, $\Wy$ is the sum of a random sample without replacement of $\mN$ elements from a finite population $\byN = (\yN_1, \ldots, \yN_{\nyN})'$. We consider a sequence of populations of increasing size, $\bm{y}^{N}, N = 2, 3, \ldots$, and random samples $\bvN = (\vN_1,\ldots, \vN_{\mN})'$ from each $\byN$. To be specific, for fixed $\bdyN$, let $\mathcal{K} = \{j: \dN_{y,j} = 1\}$ be the set of indices corresponding to the selected elements of $\byN$. Then writing $\mathcal{K} = \{ k_1, \ldots, k_{\mN} \}$, we have $\bvN = (\yN_{k_1}, \ldots, \yN_{k_{\mN}})'$.

Let $\bar{v}_{\mN} = (1/\mN) \sum_{k = 1}^{\mN} \vN_k$, and $\bar{y}_{\nyN} = (1/\nyN) \sum_{j=1}^{\nyN} \yN_j$. Then as shown by \citet[][p. 116-117]{lehmann1999},
\begin{align*}
\mathbb{E}[\bar{v}_{\mN} | \byN] &= \bar{y}_{\nyN} \\
\Var(\bar{v}_{\mN} | \byN) &= \frac{\nyN - \mN}{\mN (\nyN - 1)} \frac{1}{\nyN} \sum_{j = 1}^{\nyN} (\yN_j - \bar{y}_{\nyN})^2.
\end{align*}
We can now state Theorem \ref{Wlem}.
\begin{theorem}[Theorem 2.8.2, \citet{lehmann1999}]
\begin{equation*}
\frac{\bar{v}_{\mN} - \mathbb{E}[\bar{v}_{\mN} | \byN]}{\sqrt{\Var(\bar{v}_{\mN} | \byN)}} \rightarrow N(0, 1)
\end{equation*}
provided that $\mN \rightarrow \infty$ and $\nyN - \mN \rightarrow \infty$ as $N \rightarrow \infty$, and either of the following two conditions is satisfied:\\
i) $\mN / \nyN$ is bounded away from 0 and 1 as $N \rightarrow \infty$, and 
\begin{equation*}
\frac{\max(\yN_j - \bar{y}_{\nyN})^2}{\sum_j (\yN_j - \bar{y}_{\nyN})^2} \rightarrow 0
\end{equation*}
or\\
ii)
\begin{equation*}
\frac{\max(\yN_j - \bar{y}_{\nyN})^2}{\sum_j(\yN_j - \bar{y}_{\nyN})^2 / \nyN}
\end{equation*}
remains bounded as $N \rightarrow \infty$.
\label{Wlem}
\end{theorem}
For a proof, please see \citet{lehmann1999} and references therein, particularly the corollary to Lemma 4.1 in \citet{hajek1961some}, as well as Example 4.1 and Section 5 in \citet{hajek1961some}. Our constraints on $\mN, \nxN, \nyN$ imply that $\mN \rightarrow \infty$ and $\nyN - \mN \rightarrow \infty$ as $N \rightarrow \infty$. The other conditions in Theorem \ref{Wlem} require that the contribution of each deviance to the sum of deviances becomes negligible as the sample size becomes large. This excludes data coming from distributions with a non-finite variance, such as the Cauchy distribution.

Applying Theorem \ref{Wlem} to $\W$ we get Corollary \ref{Wlem2}.
\begin{corollary}
Conditional on $\bxN$ and $\byN$, and assuming the conditions in Theorem \ref{Wlem} hold,
\begin{equation*}
\frac{\W - \mu(\mN)}{\sqrt{V(\mN)}} \rightarrow N(0,1),
\end{equation*}
where  $\mu(\mN) = \mu_y(\mN) - \mu_x(\mN)$ and $V(\mN) = V_y(\mN) + V_x(\mN)$, with 
\begin{align*}
\mu_y(\mN) &= \mathbb{E}[\Wy | \byN] = \mN \bar{y}_{\nyN} \\
\mu_x(\mN) &= \mathbb{E}[\Wx | \bxN] = \mN\bar{x}_{\nxN} 
\end{align*}
and
\begin{align*}
V_y(\mN) &= \Var(\Wy | \byN) = \mN \frac{\nyN -\mN }{\left(\nyN - 1\right) \nyN} \sum_{j = 1}^{\nyN} (\yN_j - \bar{y}_{\nyN})^2  \\
V_x(\mN) &= \Var(\Wx | \bxN) = \mN \frac{\nxN -\mN}{\left(\nxN - 1\right) \nxN} \sum_{i = 1}^{\nxN} (\xN_i - \bar{x}_{\nxN})^2.
\end{align*}
\label{Wlem2}
\end{corollary}

Before proving Corollary \ref{Wlem2}, we state Lemma \ref{Wind}.
\begin{lemma}
For all $m$ and $N$, $\Cov \left(W_x(m^N),  W_y(m^N) | \bx, \by \right) = 0$.
\label{Wind}
\end{lemma}
\begin{proof}
First note that for all $m, N, i$, and $j$, $\dN_{x,i} \perp \dN_{y,j}$. This is a direct consequence of the sampling procedure implied by the permutation, in which we condition on the number of elements to exchange ($m$), and then randomly select $m$ elements of $\bx$ and $m$ elements of $\by$. Therefore, dropping dependence on $N$, 
\begin{align*}
\mathbb{E}\left[W_x(m) W_y(m) | \bx, \by \right] &= \mathbb{E}\left[\left(\sum_i \delta^m_{x,i} x_i \right) \left(\sum_j \delta^m_{y,j} y_j\right) | \bx, \by \right] \\
 &= \mathbb{E} \left[\sum_{i} \sum_j \delta^m_{x,i} x_i \delta^m_{y,j} y_j | \bx, \by\right] \\
 &= \sum_i \sum_j x_i y_j \mathbb{E}\left[\delta^m_{x,i} \delta^m_{y,j} \right] \\
 &= \sum_i x_i \mathbb{E} \left[\delta^m_{x,i}\right] \sum_j y_j \mathbb{E}\left[\delta^m_{y,j}\right] && (\delta^m_{x,i} \perp \delta^m_{y,j}) \\
 &= \mathbb{E}\left[W_x(m) | \bx \right] \mathbb{E}\left[W_y(m) | \by \right].
\end{align*}
Therefore,  
\begin{align*}
\Cov \left( W_x(m^N),  W_y(m^N) | \bx, \by \right) &=\mathbb{E} \left[ W_x(m^N) W_y(m^N) |\bx, \by \right] - \mathbb{E} \left[W_x(m^N) | \bx \right] \mathbb{E}\left[W_y(m^N) | \by \right] \\
 &= 0
\end{align*}
which proves the lemma.
\end{proof}

Now we prove Corollary \ref{Wlem2}.
\begin{proof} (Corollary \ref{Wlem2}) Working with the first term in (\ref{T_scaled}), we have
\begin{align*}
\Wy =\sum_{j=1}^{\nyN} \dN_{y,j} \yN_j = \mN \bar{v}_{\mN}
\end{align*}
Therefore, as shown by \citet[][p. 116-117]{lehmann1999},
\begin{align*}
\mu_y(\mN) = \mathbb{E}[\Wy | \byN] &= \mN \bar{y}_{\nyN}
\end{align*}
and 
\begin{align*}
V_y(\mN) = \Var(\Wy | \byN) &= \left(\mN \right)^2 \frac{\nyN - \mN}{\mN (\nyN - 1)} \frac{1}{\nyN} \sum_{j = 1}^{\nyN} (\yN_j - \bar{y}_{\nyN})^2. \\
&= {\mN} \frac{\nyN - \mN}{(\nyN - 1)} \frac{1}{\nyN} \sum_{j = 1}^{\nyN} (\yN_j - \bar{y}_{\nyN})^2.
\end{align*}
Similarly, working with the second term in (\ref{T_scaled}),
\begin{align*}
\mu_x(\mN) &=\mathbb{E}[\Wx | \bxN] = \mN \bar{x}_{\nxN} \\
V_x(\mN) &= \mN \frac{\nxN -\mN}{(\nxN - 1)} \frac{1}{\nxN} \sum_{i = 1}^{\nxN} (\xN_i - \bar{x}_{\nxN})^2.
\end{align*}
Applying Theorem \ref{Wlem}, we have
\begin{align*}
\frac{\Wy - \mu_y(\mN)}{\sqrt{V_y(\mN)}} = \frac{\bar{v}_{\mN} - \mathbb{E}[\bar{v}_{\mN} | \byN]}{\sqrt{\Var(\bar{v}_{\mN} | \byN)}} &\rightarrow N(0,1).
\end{align*}
Similarly, we have
\begin{align*}
\frac{\Wx - \mu_x(\mN)}{\sqrt{V_x(\mN)}} &\rightarrow N(0,1).
\end{align*}
Then by Lemma \ref{Wind}, we have
\begin{align*}
\Var \left( \Wy - \Wx \right | \bx, \by) = V_y(\mN) + V_x(\mN),
\end{align*}
Also, since uncorrelated normal random variables are independent, for $N$ sufficiently large we also have $W_y(\mN) \perp W_x(\mN)$.
Since the sum of independent normal random variables is also normal, for $N$ sufficiently large we have
$$
\W = \Wy -\Wx \sim N\left(\mu_y(\mN) - \mu_x(\mN), V_y(\mN) + V_x(\mN) \right).
$$
Equivalently, we have
\begin{equation*}
\frac{\W - \mu(\mN)}{\sqrt{ V(\mN)}} \rightarrow N(0,1)
\end{equation*}
which proves the corollary.
\end{proof}

In the rest of this appendix, we assume that $N$ is sufficiently large for asymptotic normality to hold for any given partition $m$, and so we drop $N$ from the notation.

In Corollary \ref{Tnormal} below, we apply the delta method to show that for sufficiently large $N$, the permutation distribution of the statistic $R(m)$ is normal within each partition.

\begin{corollary}
Let $R = g(W)$, and suppose that $g'(\mu(m)) > 0$ exists. Also, suppose the conditions in Theorem \ref{Wlem} hold. Then conditional on the observed data $\bx, \by$, and for N sufficiently large, $R(m) \sim N(\nu(m), \sigma^2(m))$, where the mean $\nu(m)$ and variance $\sigma^2(m)$ are functions of the partition $m$.
\label{Tnormal}
\end{corollary}

\begin{proof}
By Corollary \ref{Wlem2}, $W$ is normal for $N$ sufficiently large. Then by the delta method, $g(W)$ also converges to a normal distribution, which proves the corollary.
\end{proof}

The result in Corollary \ref{Tnormal} for the one-sided statistic $R(m)$ leads directly to the following result for its two-sided counterpart $T(m)$, given in Corollary \ref{cor1} below. However, we first define a new function $g^{\text{con}}$, the conjugate of $g$.

\begin{definition}[Conjugate $\gc$]
Let $g(W)$ be a function of $W$, in which the only other terms are the constants $n_x$, $n_y$, $\bar{x}$ and $\bar{y}$. The conjugate $\gc$ is formed by switching the place of $n_x$ with $n_y$, and $\bar{x}$ with $\bar{y}$, and reversing the sign on each occurrence of $W$.
\end{definition}
For example, for $R = \bar{x} / \bar{y}$, we have
\begin{align*}
g = \frac{n_y}{n_x} \left( \frac{n_x \bar{x} + W}{n_y \bar{y} - W} \right) && \gc = \frac{n_x}{n_y} \left( \frac{n_y \bar{y} - W}{n_x \bar{x} + W} \right)
\end{align*}
and for $R = \bar{x} - \bar{y}$, as shown below, we have
\begin{align*}
g = \bar{x} - \bar{y} + \left(\frac{1}{n_x} + \frac{1}{n_y} \right) W && \gc = \bar{y} - \bar{x} - \left(\frac{1}{n_y} + \frac{1}{n_x} \right) W.
\end{align*}
We also note that $(\gc)^{\text{con}} = g$.

\begin{corollary}
Let $T(m) = \max \left(g(W(m)), \gc(W(m)) \right)$. Under the conditions of Theorem \ref{Wlem}, and assuming $g'(\mu(m)) > 0$ and $(\gc)'(\mu(m)) > 0$ exist, then for $N$ sufficiently large,
\begin{equation}
\Pr \left( T(m) \ge t | \bx, \by \right) \approx 2-\Phi \left[\xi \left(\min \left\{ m, 2 m_{\text{max}} - m \right\} \right) \right]-\Phi \left[\xic \left(\min \left\{ m, 2 m_{\text{max}} - m \right\} \right) \right], \label{p_val}
\end{equation}
where $\Phi$ is the standard normal CDF, $m_{\text{max}}=\arg \max f(m)$, and
\begin{align*}
\xi(m) = \frac{t- g \left( \mu(m) \right) }{g'\left( \mu(m) \right) \sqrt{V(m)}}, && \xic(m) = \frac{t- \gc \left( \mu(m) \right) }{(\gc)'\left( \mu(m) \right) \sqrt{V(m)}}.
\end{align*}
\label{cor1}
\end{corollary}

\begin{proof}
For $m = 1,\ldots, m_{\max}$,
\begin{align}
\Pr(T(m) > t | \bx, \by) &= \Pr \left(g(W(m)) > t \right) + \Pr \left(\gc(W(m)) > t \right) \nonumber \\
  &= \Pr \left( Z > \frac{t- g \left( \mu(m) \right) }{g'\left( \mu(m) \right) \sqrt{V(m)}} \right) + \Pr \left( Z > \frac{t- \gc \left( \mu(m) \right) }{(\gc)'\left( \mu(m) \right) \sqrt{V(m)}} \right) \label{delta} \\
  &\approx 1-\Phi \left( \xi(m) \right) + 1 - \Phi \left( \xic(m) \right) \label{asymApprox}
\end{align}
where $Z$ is a standard normal random variable, and $\mu(m)$ and $V(m)$ are given in Corollary \ref{Wlem2}. Line (\ref{delta}) follows from the delta method, and line (\ref{asymApprox}) follows from Corollary \ref{Tnormal} for $N$ sufficiently large. 

Furthermore, since the partition-specific p-values are approximately symmetric about $m_{\max}$ (the p-values are exactly symmetric for equal sample sizes, and the symmetry worsens as the sample sizes become more imbalanced), we can get the asymptotic p-value for any partition $m  = 1,\ldots, \min(n_y, n_x)$ as 
$$
\Pr \left( T(m) \ge t | \bx, \by \right) \approx 2-\Phi \left[\xi \left(\min \left\{ m, 2 m_{\text{max}} - m \right\} \right) \right]-\Phi \left[\xic \left(\min \left\{ m, 2 m_{\text{max}} - m \right\} \right) \right].
$$
This proves the corollary.
\end{proof}

We also note that when $n_x=n_y$, the approximation in (\ref{p_val}) is equally accurate for partitions both smaller and larger than $m_{\text{max}}$. However, for unequal sample size, the approximation is less accurate for partitions larger than $m_{\text{max}}$.

In summary, and to be explicit with all quantities, for the statistic $T=\max(\bar{x}/\bar{y}, \bar{y}/\bar{x})$, we have
$$
\Pr \left( T(m) \ge t | \bx, \by \right) \approx 2-\Phi \left[\xi \left(\min \left\{ m, 2 m_{\text{max}} - m \right\} \right) \right]-\Phi \left[\xic \left(\min \left\{ m, 2 m_{\text{max}} - m \right\} \right) \right]
$$
where $\Phi$ is the standard normal CDF, $m_{\max} = \arg \max_m f(m)$, $f(m) = \binom{N}{n_{\min}}^{-1}\binom{n_x}{m}\binom{n_y}{m}$, $n_{\min} = \min(n_x, n_y)$, and \footnote{Implementation note: In the fastPerm package, we use the same function to compute $\xi$ and $\xic$, just reversing the order of the arguments related to $x$ and $y$.}
\begin{align*}
\xi(m) &= \frac{t- g \left( \mu(m) \right) }{g'\left( \mu(m) \right) \sqrt{V(m)}} && \xic(m) = \frac{t- \gc \left( \mu(m) \right) }{(\gc)'\left( \mu(m) \right) \sqrt{V(m)}} \\
g(\mu(m))&=\frac{n_y}{n_x}\left( \frac{ n_x \bar{x} + \mu(m)}{ n_y \bar{y} - \mu(m)} \right) && \gc(\mu(m))=\frac{n_x}{n_y}\left( \frac{ n_y \bar{y} - \mu(m)}{ n_x \bar{x} + \mu(m)} \right) \\
g' \left( \mu(m)\right) &= \frac{n_y}{n_x} \left( \frac{ n_y \bar{y} + n_x \bar{x} }{\left( n_y \bar{y} - \mu(m) \right)^2} \right) && (\gc)' \left( \mu(m)\right) = -\frac{n_y}{n_x} \left( \frac{n_x \bar{x} + n_y \bar{y} }{\left( n_x \bar{x} + \mu(m) \right)^2} \right)
\end{align*}
where
\begin{align*}
\mu(m) &= m(\bar{y} - \bar{x}) \\
V(m) &= m \left[\frac{n_y - m}{n_y (n_y - 1)} \sum_{j=1}^{n_y}(y_j - \bar{y})^2 + \frac{n_x - m}{n_x (n_x - 1)} \sum_{i=1}^{n_x}(x_i - \bar{x})^2 \right].
\end{align*}

To get the expected trend shown Figure 1 of Section 3, we set $t=\bar{x}/\bar{y}$ (the observed test statistic), and substituted expected values for the sample quantities. For example, if we generated the elements of $\bx$ as iid realizations of a random variable $X$, then we substituted $\mathbb{E}[X]$ for $\bar{x}$, and $\Var(X)$ for $(n_x -1)^{-1} \sum_{i=1}^{n_x}(x_i - \bar{x})^2$.

We note that we get similar results for $T=|\bar{x}-\bar{y}|$. In this case we can write $R(m)$ as
\begin{align*}
R(m) &= \frac{1}{n_x} [ (\bOne-\bd_x)' \bx + \bd_y' \by] - \frac{1}{n_y}[{\bd_x}' \bx + (\bOne-\bd_y)' \by]\\
  &=\bar{x}-\bar{y} + \left( \frac{1}{n_x} + \frac{1}{n_y} \right) W(m)
\end{align*}
Therefore, (\ref{p_val}) still holds, but with $g(\mu(m)) = \bar{x} - \bar{y} +\left(n_x^{-1} + n_y^{-1} \right) \mu(m) $, and $g'(\mu(m)) = \left(n_x^{-1} + n_y^{-1} \right)$, with the corresponding results for $\gc$ and $(\gc)'$. All other formula are the same as those given for the ratio of means. The resulting trend for $T = |\bar{x} - \bar{y}|$ is shown in Figure \ref{pApprox_Diff_100_100} with $n_x = n_y = 100$, $\mu_x = 4, \mu_y = 2$, and $\sigma_x^2 = \sigma_y^2 = 1$.

\begin{figure}[htbp]
\centering
\includegraphics[scale = 0.48]{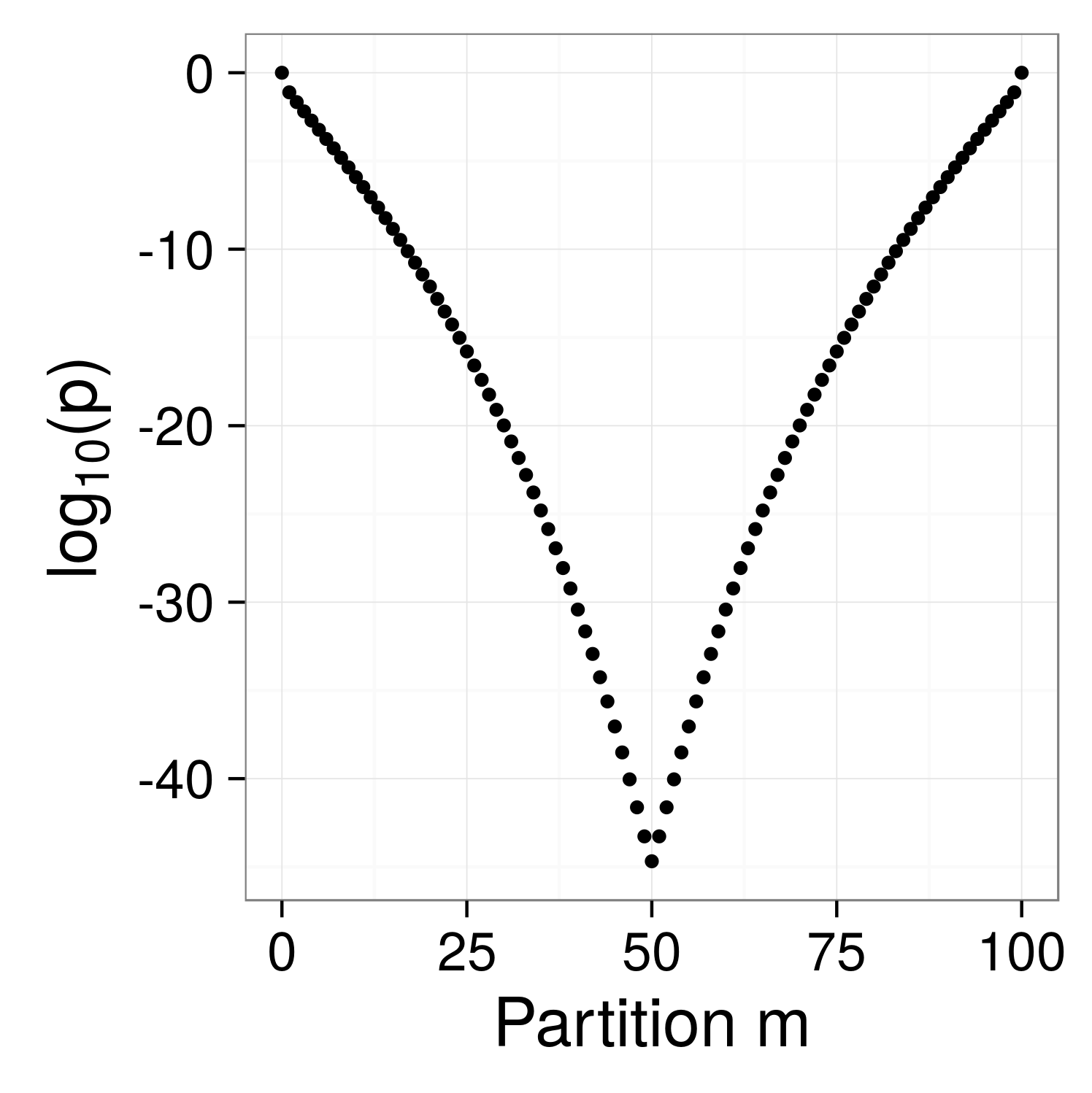}
\caption{Trend in p-values across the partitions for $T = |\bar{x} - \bar{y}|$ with $n_x = n_y = 100$, $\mu_x = 4, \mu_y = 2, \sigma_x^2 = \sigma_y^2 = 1$.}
\label{pApprox_Diff_100_100}
\end{figure}

While this appendix shows that the nearly log concave trend holds for both $T = |\bar{x} - \bar{y}|$ and $T = \max(\bar{x}/ \bar{y}, \bar{y} / \bar{x})$, we speculate that the trend might be similar for other statistics that are smooth functions of the means. The results for $R = \bar{x} / \bar{y}$ and $R = \bar{x} - \bar{y}$ above suggest a general formulation of permutation statistics in terms of $W$, which might help with this effort. This general formulation is presented in Proposition \ref{generalize}, in which $R$ could be any statistic of the sample means, and not necessarily the ratio or difference of means.
\begin{proposition}
Let $R(m) = R(\bar{x}^*(m), \bar{y}^*(m) | \bx, \by)$ be any statistic of the permuted sample means conditional on observed data $\bx, \by$, where $\bar{x}^*(m)$ and $\bar{y}^*(m)$ are the means of a permuted dataset $({\bx^*}', {\by^*}')'$ corresponding to a permutation $\pi \in \Pi (m)$. Then we can always write $R(m) = g(W(m))$ for some function $g$ that is conditional on the observed data $\bm{x}, \bm{y}$.
\label{generalize}
\end{proposition}
\begin{proof}
Noting that $\bar{x}^*(m) = \bar{x} + (1/n_x)W(m)$ and $\bar{y}^*(m) = \bar{y} - (1/n_y)W(m)$, we have
\begin{align*}
R \left( \bar{x}^*(m), \bar{y}^*(m) | \bx, \by \right) &= R \left( \bar{x} + (1/n_x)W(m), \bar{y} - (1/n_y)W(m) | \bx, \by \right) \\
 &=g \left( W(m) \right)
\end{align*}
where the last line follows, because $\bar{x}$, $\bar{y}$, $n_x$, and $n_y$ are constant conditional on $\bx, \by$, and can be absorbed into the functional form of $R$. This proves the proposition.
\end{proof}

Then for any one-sided statistic $R = g(W)$, in order for asymptotic normality to hold within each partition for the corresponding two-sided statistic $T$, we must check the conditions in Theorem \ref{Wlem} and Corollary \ref{cor1}. However, it remains to be shown what additional properties are required to ensure a log concave trend in p-values across the partitions, so we must currently check new statistics on a case-by-base basis.

\section{Parametric p-values for ratios and differences of gamma random variables \label{B}}

The results in this appendix are used in our simulations of exponential and gamma random variables to obtain parametric approximations to the permutation p-value.

\subsection{Ratio of means}

Let $F$ be the beta prime CDF (also called a Pearson type VI distribution \citep[][p. 248]{johnson2002continuous}), and let $f$ be the corresponding pdf. Following the form given by \citet{Becker2016}, for $Z \sim F$,
$$
f_Z(z; \alpha_1, \alpha_2, s, q) = \frac{\left(\frac{z - q}{s}\right)^{\alpha_1 -1} \left(1 + \frac{z - q}{s}\right)^{-\alpha_1 - \alpha_2}}{sB(\alpha_1, \alpha_2)}.
$$
As we show in this section, if $X_i \overset{iid}{\sim} \text{Exp}(\lambda_x)$ and  $Y_j \overset{iid}{\sim} \text{Exp}(\lambda_y)$, then $\bar{X} / \bar{Y}$ and $\bar{Y}/\bar{X}$ follow scaled beta prime distributions. This allows us to approximate the permutation p-value for the ratio statistic with the p-value from a beta prime. We note that the beta prime p-value is not conditional on the data, so is not the same as the permutation p-value, but simulation results suggest it is a reasonable approximation.

As in Section 5.2, let $x_i, i=1,\ldots,n_x$, and $y_j, j=1,\ldots,n_y$, be realizations of the respective random variables $X_i \overset{iid}{\sim} \text{Exp}(\lambda_x)$ and  $Y_j \overset{iid}{\sim} \text{Exp}(\lambda_y)$. We consider the quantity $T=\max \left(\bar{X}/\bar{Y}, \bar{Y}/\bar{X} \right)$, and denote the observed statistic as $t=\max \left(\bar{x}/\bar{y}, \bar{y}/\bar{x} \right)$. Then under the null hypothesis that $\lambda_x=\lambda_y$, the p-value from the beta prime distribution is
\begin{align}
p_\beta &= \Pr(T \ge t) \nonumber \\
  &= \Pr \left( \max(\bar{X}/\bar{Y}, \bar{Y}/\bar{X}) \ge t \right) \nonumber \\
  &= \Pr \left( \left\{ \bar{X}/\bar{Y} \ge t \right\} \cup \left\{\bar{Y}/\bar{X} \ge t \right\} \right) \nonumber \\
  &= \Pr \left(\bar{X}/\bar{Y} \ge t \right) + \Pr \left( \bar{Y}/\bar{X} \ge t \right) && (\text{disjoint})  \label{disjoint}\\
  &= \Pr \left( \frac{n_y}{n_x} \frac{\sum_i X_i}{\sum_j Y_j} \ge t \right) + \Pr \left(\frac{n_x}{n_y} \frac{\sum_j Y_j}{\sum_i X_i} \ge t \right) \label{transform} \\
  &=1-F \left(t; \alpha_1 = n_x, \alpha_2 = n_y, s = n_y / n_x, q = 0 \right) \label{betaResult}\\
  &+1 - F \left(t; \alpha_1 = n_y, \alpha_2 = n_x, s = n_x/n_y, q = 0 \right). \nonumber
\end{align}
The equality in (\ref{disjoint}) follows because $\bar{X}/\bar{Y} \ge t$ if and only if $\bar{Y}/\bar{X} < t$ (assuming $t \ne 1$, which occurs with probability one). Line \ref{betaResult} follows from well known properties, which we outline below.

Let $U_1 \sim \text{Gamma}(\alpha_1,\lambda_1)$ and $U_2 \sim \text{Gamma}(\alpha_2,\lambda_2)$, $U_1 \perp U_2$. Also, let $V_1 = h_1(U_1, U_2) = U_1 / U_2$ and $V_2 = h_2(U_1,U_2) = U_2$, with respective inverse transformations $U_1 = h^{-1}(V_1, V_2)=V_1 V_2$ and $U_2 = h^{-1}(V_1, V_2)=V_2$. Then, noting that the Jacobian of the transformation is
\begin{equation*}
J = \begin{vmatrix}
\partial u_1 / \partial v_1 & \partial u_1 / \partial v_2 \\
\partial u_2 / \partial v_1 & \partial u_2 / \partial v_2
\end{vmatrix} =
\begin{vmatrix}
v_2 & v_1 \\
0 & 1
\end{vmatrix} = 
v_2,
\end{equation*}
we have
\begin{align*}
f_{V_1, V_2}(v_1, v_2) &= f_{U_1, U_2} \left( h^{-1}_1(v_1, v_2), h^{-1}_2(v_1, v_2) \right) |J| \\
&= \frac{ \lambda_1^{\alpha_1}}{\Gamma(\alpha_1)} (v_1 v_2)^{\alpha_1 - 1} e^{-\lambda_1 v_1 v_2}
\frac{ \lambda_2^{\alpha_2}}{\Gamma(\alpha_2)} v_2^{\alpha_2 - 1} e^{-\lambda_2 v_2} v_2\\
&= \frac{ \lambda_1^{\alpha_1} \lambda_2^{\alpha_2} }{\Gamma(\alpha_1) \Gamma(\alpha_2)} v_1^{\alpha_1 - 1} v_2^{\alpha_1 + \alpha_2 -1} e^{-(\lambda_1 v_1 + \lambda_2)v_2}.
\end{align*}
Therefore,
\begin{align*}
f_{V_1}(v_1) &= \int_0^\infty f_{V_1, V_2}(v_1, v_2) d v_2 \\
&= \frac{ \lambda_1^{\alpha_1} \lambda_2^{\alpha_2} }{\Gamma(\alpha_1) \Gamma(\alpha_2)} v_1^{\alpha_1 - 1} \int_0^\infty v_2^{\alpha_1 + \alpha_2 -1} e^{-(\lambda_1 v_1 + \lambda_2)v_2} d v_2 \\
&= \frac{ \lambda_1^{\alpha_1} \lambda_2^{\alpha_2} }{\Gamma(\alpha_1) \Gamma(\alpha_2)} v_1^{\alpha_1 - 1} \frac{\Gamma(\alpha_1 + \alpha_2)}{(\lambda_1 v_1 + \lambda_2)^{\alpha_1 + \alpha_2}} \\
&= \frac{\left( \frac{v_1}{\lambda_2 / \lambda_1} \right)^{\alpha_1 - 1} \left(1 + \frac{v_1}{\lambda_2/\lambda_1} \right)^{-\alpha_1 - \alpha_2}}{(\lambda_2 / \lambda_1) B(\alpha_1, \alpha_2)},
\end{align*}
which is a generalized beta prime distribution with shape parameters $\alpha_1$ and $\alpha_2$, location parameter $q=0$, and scale parameter $s=\lambda_2/\lambda_1$. In the case where $\lambda_1 = \lambda_2$, this simplifies to the standard beta prime distribution with shape parameters $\alpha_1$ and $\alpha_2$. This shows that whenever $U_1 \sim \text{Gamma}(\alpha_1,\lambda)$, $U_2 \sim \text{Gamma}(\alpha_2,\lambda)$, and $U_1 \perp U_2$, we have $U_1/U_2 \sim F(\alpha_1, \alpha_2, 1, 0)$. We note that some sources report that for $U_1 \sim \text{Gamma}(\alpha_1,\lambda_1)$, $U_2 \sim \text{Gamma}(\alpha_2,\lambda_2)$, and $U_1 \perp U_2$, we have $U_1/U_2 \sim F(\alpha_1, \alpha_2, 1, 0)$ if $\lambda_1 = \lambda_2 =1$ \citep[e.g.,][]{leemis2008univariate}. However, as shown above, this also holds when $\lambda_1 = \lambda_2 \ne 1$.

Now let $Z = \left(\sum_{i=1}^{n_x} X_i \right) / \left(\sum_{j=1}^{n_y}Y_i \right)$. Since $X_i \overset{iid}{\sim} \text{Exp}(\lambda_x)$ and  $Y_j \overset{iid}{\sim} \text{Exp}(\lambda_y)$, it follows that $\sum_{i=1}^{n_x} X_i \sim \text{Gamma}(n_x, \lambda_x)$ and $\sum_{j=1}^{n_y} Y_j \sim \text{Gamma}(n_y, \lambda_y)$. Then under the null of $\lambda_x = \lambda_y$, the results above give $Z \sim F(n_x, n_y, 1, 0)$ and $1/Z \sim F(n_y, n_x, 1, 0)$.

Now let $W = s Z$. Then by a change of variable, we have
$$
f_W(w) = \frac{\left(\frac{w}{s}\right)^{n_x -1} \left(1 + \frac{w}{s}\right)^{-n_x - n_y}}{sB(n_x, n_y)}
$$
Applying this result to (\ref{transform}), we have 
\begin{align*}
\frac{n_y}{n_x}\frac{\sum_{i=1}^{n_x} X_i}{\sum_{j=1}^{n_y} Y_j} \sim F(\cdot;n_x, n_y, n_y/n_x, 0)
\end{align*}
and similarly,
\begin{align*}
\frac{n_x}{n_y}\frac{\sum_{j=1}^{n_y} Y_j}{\sum_{i=1}^{n_x} X_i} \sim F(\cdot;n_y, n_x, n_x/n_y, 0)
\end{align*}
Then (\ref{betaResult}) follows directly from (\ref{transform}).

To compute the CDF values for the scaled beta prime, we used the \verb|PearsonDS| package for \textsf{R} \citep{Becker2016}.

Similarly, for $X_i \overset{iid}{\sim} \text{Gamma}(\alpha_x, \lambda_x)$ and $Y_j \overset{iid}{\sim} \text{Gamma}(\alpha_y, \lambda_y)$, $\sum_{i=1}^{n_x} X_i \sim \text{Gamma}(n_x \alpha_x, \lambda_x)$ and $\sum_{j=1}^{n_y} Y_j \sim \text{Gamma}(n_y \alpha_y, \lambda_y)$. Then letting $Z = \left( \sum_{i=1}^{n_x} X_i \right) / \left(\sum_{j=1}^{n_y} Y_j \right)$, under the null of $H_0: \lambda_x = \lambda_y, \alpha_x = \alpha_y = \alpha$, we have $Z \sim F(\cdot ; n_x \alpha, n_y \alpha, 1, 0)$ and $1 / Z \sim F(\cdot ; n_y \alpha, n_x \alpha, 1, 0)$, so $(n_y / n_x) Z \sim F(\cdot; n_x \alpha, n_y \alpha, n_y / n_x, 0)$ and $(n_x / n_y) Z \sim F(\cdot; n_y \alpha, n_x \alpha, n_x / n_y, 0)$. Therefore, 
\begin{align*}
p_\beta = \Pr(T \ge t) &=1-F \left(t; n_x \alpha, n_y \alpha, n_y / n_x, 0 \right) \\
  &+1 - F \left(t; n_y \alpha, n_x \alpha, n_x/n_y, 0 \right). \nonumber
\end{align*}

In our simulations, we generate data under the alternative $H_1: \lambda_x \ne \lambda_y, \alpha_x = \alpha_y = \alpha$ for various values of $\alpha$. While we would ideally also simulate under the alternatives $H_1: \lambda_x \ne \lambda_y, \alpha_x \ne \alpha_y$ and $H_1: \lambda_x = \lambda_y, \alpha_x \ne \alpha_y$, in these scenarios it is not possible to compute $p_\beta$ under $H_0: \alpha_x = \alpha_y, \lambda_x = \lambda_y$, because $\alpha$ does not disappear in the beta prime density. Consequently, we would have to compute $p_\beta$ under $H_0: \alpha_x = \alpha_y = c, \lambda_x = \lambda_y$ for a specified constant $c$. This is more restrictive than the null hypothesis for the permutation test, and consequently, it would not be clear how to compute the parametric p-value to use as an approximation for the true permutation p-value.

\subsection{Difference in means}

Let $M_X(t)$ be the moment generating function (MGF) for random variable $X$. Then for $X_i \overset{iid}{\sim} \text{Gamma}(\alpha, \lambda), i=1,\ldots, n$, $M_{\frac{1}{n}\sum_{i^=1}^n X_i}(t) = M_{\sum_{i=1}^n X_i}(t/n) = \prod_{i=1}^n M_{X_i}(t/n) = \left(1 - \frac{1}{n \lambda} t\right)^{-n \alpha}$, which is the MGF for a Gamma distribution with shape parameter $n \alpha$ and rate parameter $n \lambda$. Therefore, $\bar{X} \sim \text{Gamma}(n \alpha, n\lambda)$.

Then for $X_i \overset{iid}{\sim} \text{Gamma}(\alpha, \lambda), i=1,\ldots, n_x$ and $Y_j \overset{iid}{\sim} \text{Gamma}(\alpha, \lambda), j=1,\ldots, n_y$, the distribution of $\bar{X} - \bar{Y}$, which we denote as $G$, is \citep{klar2015note} 
\begin{equation}
G(z) = \Pr(\bar{X} - \bar{Y} \le z) = C \underset{A(z)}{\underbrace{\int_{\max\{0, -z\}}^\infty v^{n_y \alpha -1} e^{-n_y \lambda v} \gamma \left(n_x\alpha, n_x \lambda(v + z) \right) dv}},
\label{difGammaCDF}
\end{equation}
where $\gamma(a,b) = \int_0^b s^{a-1}e^{-s} ds$ is the lower incomplete gamma function, and\\
$C = (n_y\lambda)^{n_y \alpha} / \left(\Gamma(n_x\alpha) \Gamma(n_y\alpha)\right)$ is the normalizing constant. \citet{klar2015note} also gives the density for $\bar{X} - \bar{Y}$, which was derived by \citet{mathai1993noncentral}.

However, we found that in our simulations, several scenarios led to numerical problems in computing (\ref{difGammaCDF}) due to large gamma and incomplete gamma function values. These were not solved by computing $G(z) = \exp\{n_y \alpha \log(n_y \lambda) - \log \Gamma(n_x \alpha) - \log \Gamma(n_y \alpha)\ + \log(A(z))\}$ where $\log \Gamma$ is the log gamma function. As an alternative, we used a saddlepoint approximation for (\ref{difGammaCDF}). As described below, the saddlepoint approximation is accurate, and did not pose computational difficulties.

To compute the saddlepoint approximation, note that under $H_0: \lambda_x = \lambda_y = \lambda, \alpha_x = \alpha_y = \alpha$, the MGF of $\bar{X} - \bar{Y}$ is
\begin{align*}
M_{\bar{X} - \bar{Y}}(t) = \left(1-\frac{1}{n_x\lambda} t \right)^{-n_x \alpha} \left(1 + \frac{1}{n_y \lambda} t \right)^{-n_y \alpha} \quad t \in (-n_y \lambda, n_x \lambda),
\end{align*}
and the cumulant generating function is 
$$
K(t) = \log \left( M_{\bar{X} - \bar{Y}}(t) \right) = -n_x \alpha \log \left(1 - \frac{t}{n_x \lambda} \right) - n_y \alpha \log \left(1 + \frac{t}{n_y \lambda} \right).
$$
After some algebra, we get the derivatives
\begin{align*}
K'(t) &= \frac{\alpha (n_x + n_y)t}{(n_x \lambda - t) (n_y \lambda + t)} \\
K''(t) &= \alpha (n_x + n_y) \frac{t^2 + n_x n_y \lambda^2}{\left[(n_x \lambda - t) (n_y \lambda + t)\right]^2}.
\end{align*}
Let $\hat{t} = \hat{t}(z) \in (-n_y \lambda, n_x \lambda)$ be the solution to $K'(\hat{t}) = z$. Then as \citet{butler2007saddlepoint} describes, the saddlepoint approximation of the cumulative distribution for $z \ne \mathbb{E}[\bar{X} - \bar{Y}] = 0$ is \citep{lugannani1980saddle}
\begin{equation}
\hat{G}(z) = \Phi(\hat{w}) + \phi(\hat{w}) \left(\frac{1}{\hat{w}} - \frac{1}{\hat{u}} \right),
\label{FHat}
\end{equation}
where $\hat{w} = \text{sgn}(\hat{t}) \sqrt{2 \left[\hat{t}z - K(\hat{t}) \right]}$, $\hat{u} = \hat{t}\sqrt{K''(\hat{t})}$, and $\Phi$ and $\phi$ are the standard normal distribution and density, respectively. The two-sided p-value is then $p_{\text{saddle}} = \Pr(T \ge t) =  1 - \hat{G}(t; n_x, n_y, \lambda, \alpha) + \hat{G}(-t; n_x, n_y, \lambda, \alpha).$

Figure \ref{FHat_Ftrue_comparison} compares the true distribution (\ref{difGammaCDF}) and saddlepoint approximation (\ref{FHat}) for $n_x = n_y = 100$, $\alpha = 1$, and $\lambda = 4$. Figure \ref{FHat_Ftrue_comparison} shows agreement between the true distribution and saddlepoint approximation far into the tail. The trend is similar for other parameter values (not shown), and appears to be reliable up to quantile values of around $10^{-200}$. We also note that through simulations, we found that both the true distribution and the saddlepoint approximation agreed with empirical distribution for a variety of parameter values (not shown).

\begin{figure}[htbp]
\centering
\includegraphics[scale = 0.5]{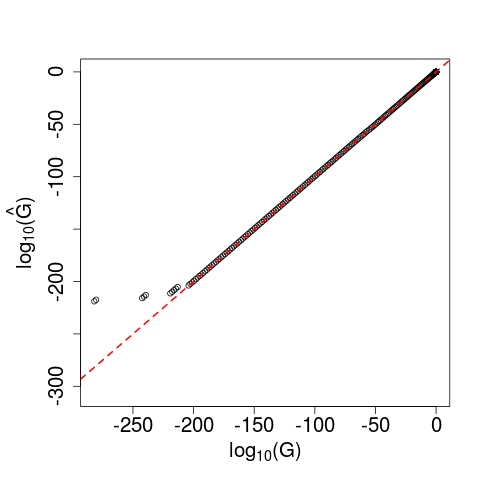}
\caption{Comparison of true ($G$) and saddlepoint approximation ($\hat{G}$) distributions of the difference of gamma random variables}
\label{FHat_Ftrue_comparison}
\end{figure}

Both the true distribution (\ref{difGammaCDF}) and saddlepoint approximation (\ref{FHat}) are functions of $\alpha$ and $\lambda$. Neither parameter disappears under the null of $H_0: \alpha_x = \alpha_y = \alpha, \lambda_x = \lambda_y = \lambda$, so we must set $\alpha$ and $\lambda$ to fixed values to compute p-values. To do this, in the simulations, we pooled the generated data, computed the maximum likelihood estimates (MLEs), and plugged the MLEs into (\ref{FHat}). In the simulations, we found that allowing both $\alpha$ and $\lambda$ to vary led to less reliable results than allowing just one parameter to vary. To be consistent with our simulations for the ratio of gamma means, we fixed $\alpha$ and used the MLE estimate for $\lambda$ in the simulations.

We note that this procedure for obtaining a parametric approximation to the permutation p-value involves three approximations: 1) approximating the permutation p-value (conditional on the data) with a parametric distribution (not conditional on the data), 2) approximating the parametric distribution with a saddlepoint approximation, and 3) approximating the general null $H_0: \lambda_x = \lambda_y$ with the more restrictive null $H_0: \lambda_x = \lambda_y = \hat{\lambda}$, where $\hat{\lambda}$ is the MLE.

To obtain the MLE estimates, let $\bm{z} = (\bx', \by')'$ be the pooled data, $N = n_x + n_y$ be the total sample size, and $\bar{z} = N^{-1} \sum_{i=1}^N z_i, s^2 = (N-1)^{-1} \sum_i (z_i - \bar{z})^2$ be the sample mean and variance, respectively. Then assuming iid observations, the joint log likelihood is
$$
\ell = N \alpha \log(\lambda) - N \log \left(\Gamma(\alpha) \right) + (\alpha - 1) \sum_i \log(z_i) - N \lambda \bar{z}.
$$
Taking the derivative with respect to $\lambda$ and setting to zero, we get $\lambda = \alpha / \bar{z}$. Then taking $\partial \ell / \partial \alpha$ and substituting in $\lambda = \alpha / \bar{z}$, we get
\begin{align*}
\ell' (\alpha) &= N \log \left( \frac{\alpha}{\bar{z}} \right) - N \Psi(\alpha) + \sum_i \log(x_i) \\
\ell'' (\alpha) &= \frac{N}{\alpha} - N \Psi'(\alpha),
\end{align*}
where $\Psi(\alpha) = d \log(\Gamma(\alpha)) / d \alpha$ is the digamma function, and $\Psi'(\alpha) = d \Psi(\alpha) / d \alpha$ is the trigamma function. We used Newton-Raphson until convergence of $\ell(\alpha)$ to get the MLE $\hat{\alpha}$, where each update is given by $\alpha^{k+1} = \alpha^k - \ell' \left(\alpha^k \right) / \ell'' \left(\alpha^k \right)$, and then set $\hat{\lambda} = \hat{\alpha} \bar{z}$. To get initial values for $\alpha$, we used the method of moments and set $\alpha^0 = \bar{z}^2 / s^2$.

\section{Additional Simulations \label{C}}

In this section, we present simulation results under additional scenarios.

\subsection{Difference in means with normal data \label{diffMean_normal}}

In this subsection, we use the statistic $T = |\bar{x} - \bar{y}|$ with data generated as normal random variables.

\subsubsection{Small sample sizes}

We generated data $x_{i}, i=1,\ldots,n_x$ and $y_j, j=1,\ldots,n_y$ as realizations of the respective random variables $X_i\overset{\text{iid}}{\sim} N(\mu_x, 1)$ and $Y_{j} \overset{\text{iid}}{\sim} N(\mu_y, 1)$. For equal sample sizes, we set $n=n_x = n_y = 20, 40, 60$, and for unequal sample sizes we set $n_x = 20, 40, 60$ and $n_y = 100$. For both equal and unequal sample sizes, and for each each $n$ or $n_x$, we set $\mu_x =2$ or 3, and $\mu_y = 0$, and simulated 100 datasets for each combination of parameters. We use the p-value from a t-distribution, denoted as $p_t$ as an approximation for the true permutation p-value. 

Results for equal and unequal sample size are shown in Figures \ref{simDiff_sym_smallN} and \ref{simDiff_nonSym_smallN}, respectively. \textit{Alg 1} is our resampling algorithm with $B_{\text{pred}}=10^3$ iterations in each partition, \textit{Asym} is our asymptotic approximation, \textit{SAMC} is the SAMC algorithm, and $p_t$ is a two-sided t-test with equal variance. The number of iterations used by our resampling algorithm is shown in Figures \ref{simDiff_sym_iter_smallN} and \ref{simDiff_nonSym_iter_smallN}. We note that the bias shown in Figures \ref{simDiff_sym_smallN_pvals} and \ref{simDiff_nonSym_smallN_pvals} are similar to that obtained with moment-corrected correlation (MCC) \citep{zhou2015hypothesis}, shown in Figure \ref{mcc_smallN} of Appendix \ref{D}.

\begin{figure}[htbp]
\centering
  \begin{subfigure}{0.58\textwidth}
  \centering
  \includegraphics[width=1\linewidth]{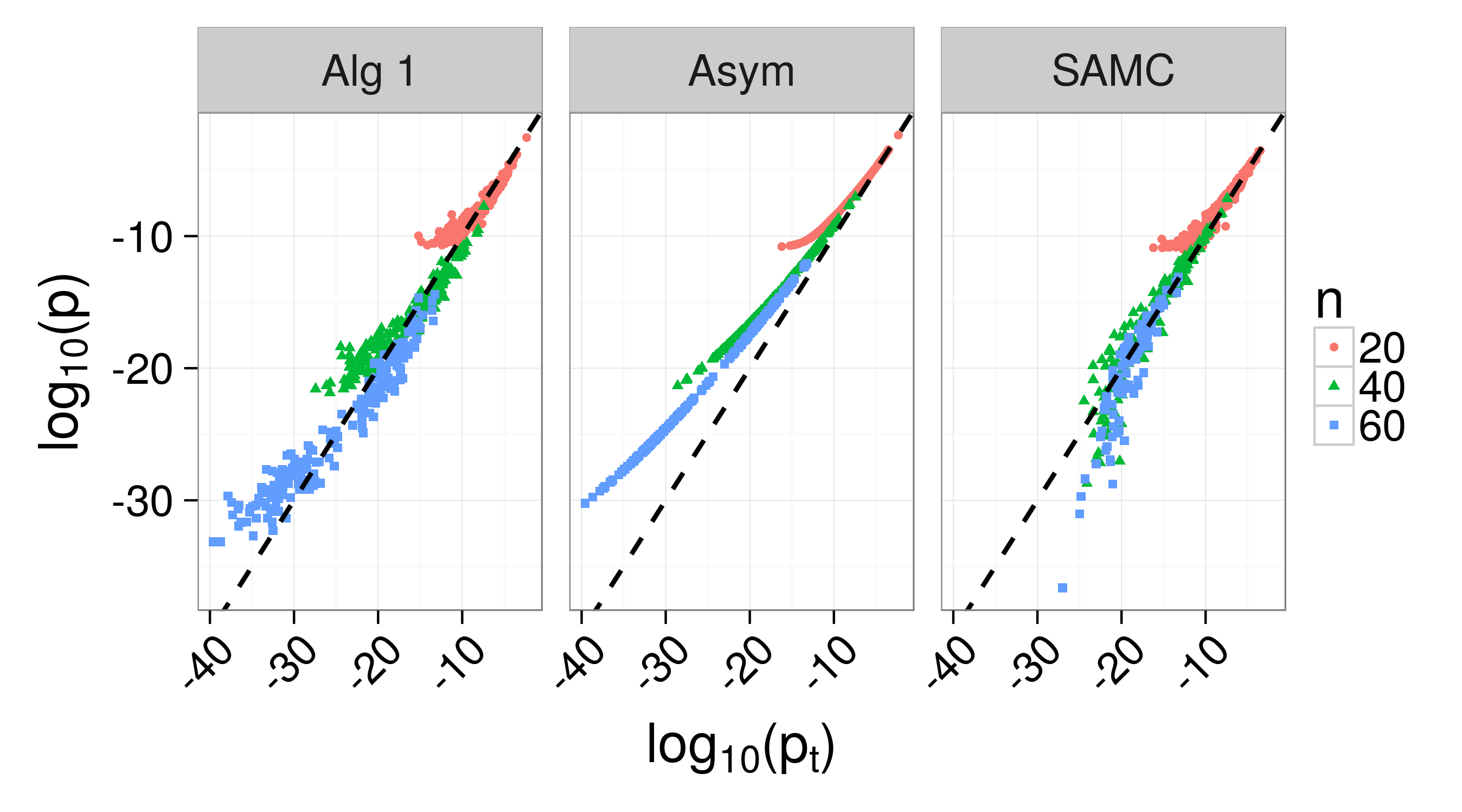}
  \caption{p-values}
  \label{simDiff_sym_smallN_pvals}
  \end{subfigure}
  \begin{subfigure}{0.38\textwidth}
  \centering
  \includegraphics[width=1\linewidth]{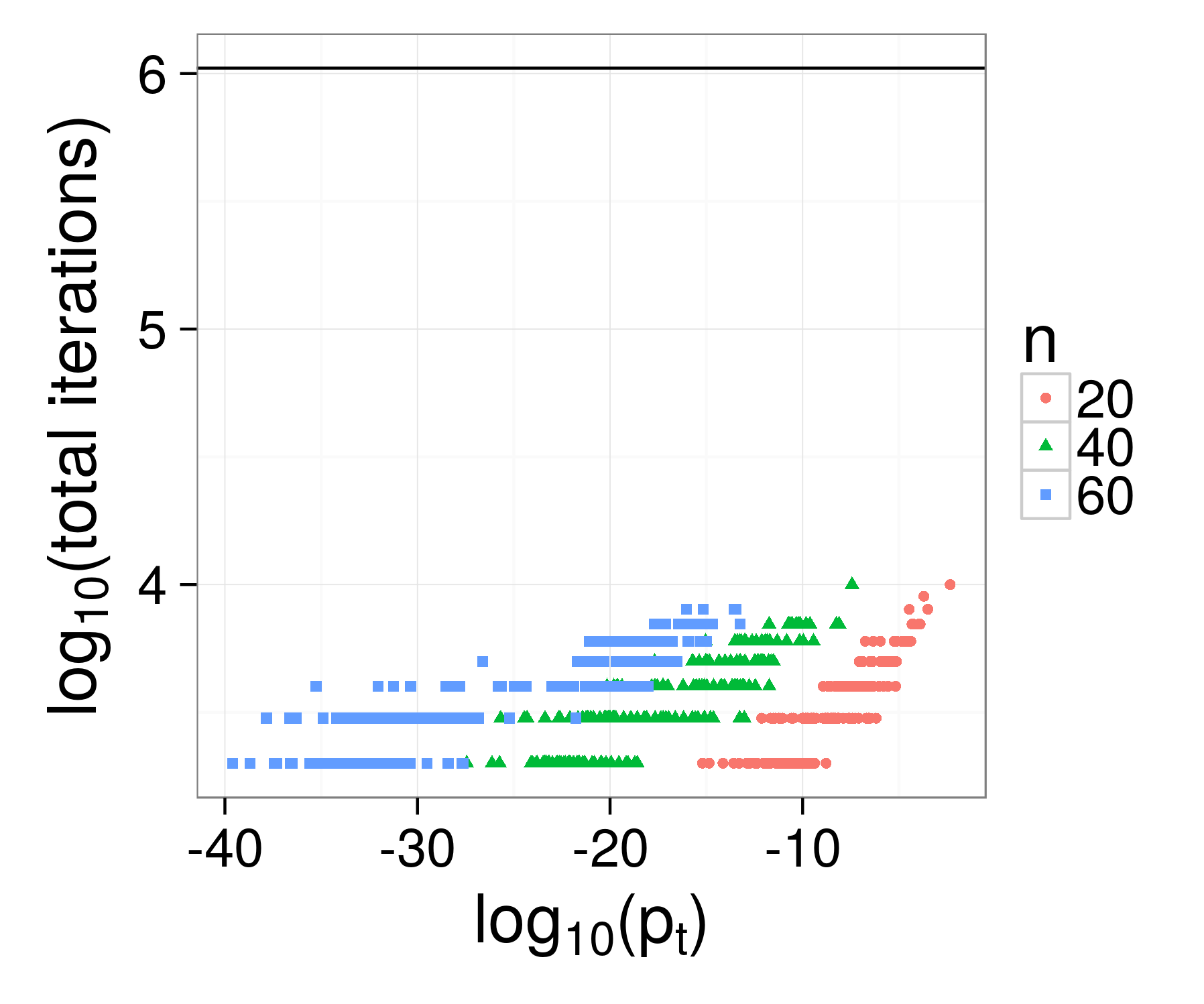}
  \caption{Iterations in resampling algorithm}
  \label{simDiff_sym_iter_smallN}
  \end{subfigure}
\caption{Simulation results using the statistic $T=|\bar{x}-\bar{y}|$ with normal data and $\mu_x = 2$ or 3, and $\mu_y = 0$, with equal sample sizes of $n=n_x=n_y=20$, 40, 60. \textit{Alg 1} is our resampling algorithm with $B_{\text{pred}}=10^3$ iterations in each partition, \textit{Asym} is our asymptotic approximation, \textit{SAMC} is the SAMC algorithm, and $p_t$ is a two-sided t-test with equal variance. The diagonal dashed line has slope of 1 and intercept of 0, and indicates agreement between methods.}
\label{simDiff_sym_smallN}
\end{figure}

\begin{figure}[htbp]
\centering
  \begin{subfigure}{0.58\textwidth}
  \centering
  \includegraphics[width=1\linewidth]{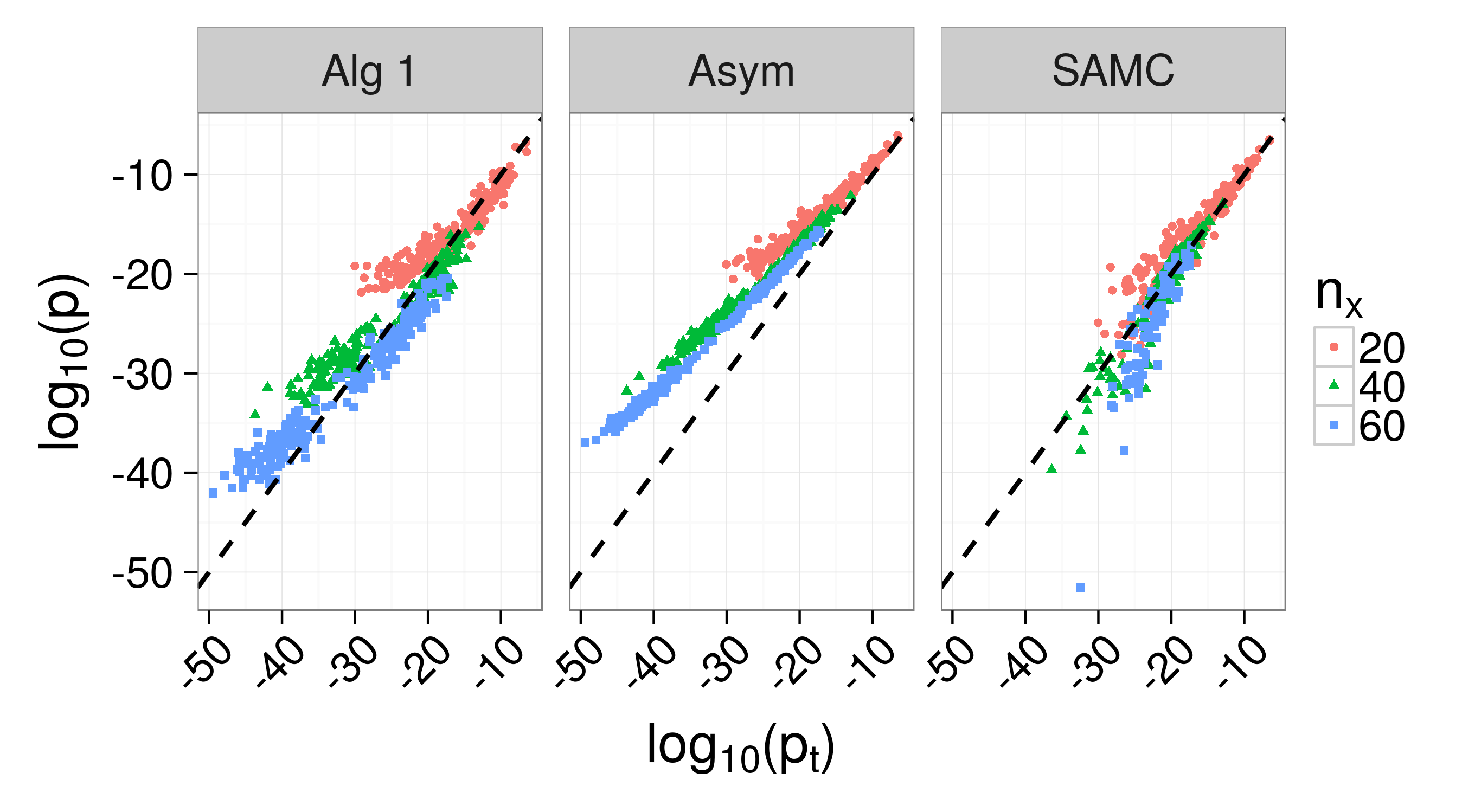}
  \caption{p-values}
  \label{simDiff_nonSym_smallN_pvals}
  \end{subfigure}
  \begin{subfigure}{0.38\textwidth}
  \centering
  \includegraphics[width=1\linewidth]{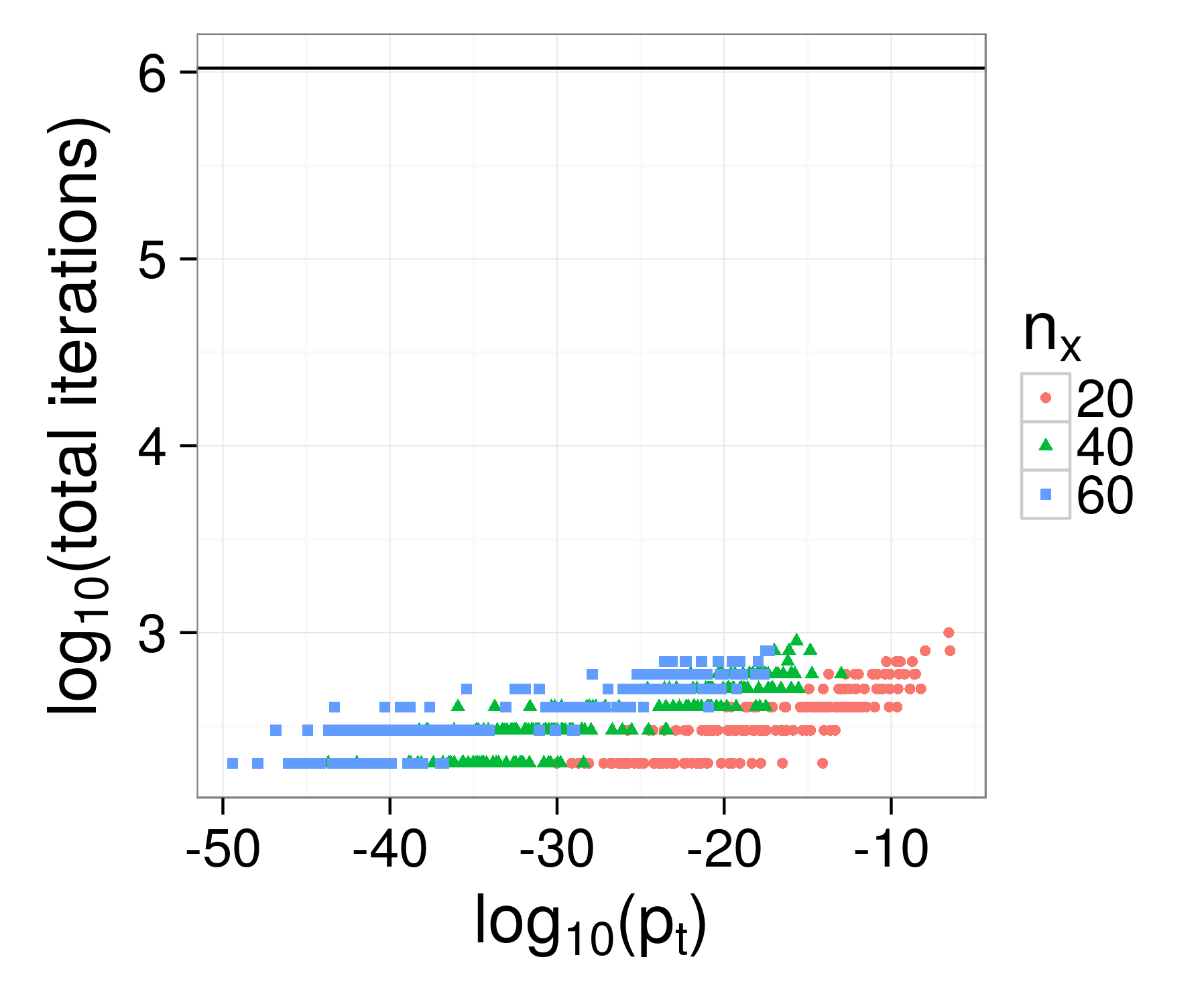}
  \caption{Iterations in resampling algorithm}
   \label{simDiff_nonSym_iter_smallN}
  \end{subfigure}
\caption{Simulation results using the statistic $T=|\bar{x}-\bar{y}|$ with normal data and $\mu_x = 2$ or $3$, $\mu_y = 0$, with unequal sample sizes, where $n_y=100$ and $n_x=20, 40, 60$. \textit{Alg 1} is our resampling algorithm with $B_{\text{pred}}=10^3$ iterations in each partition, \textit{Asym} is our asymptotic approximation, \textit{SAMC} is the SAMC algorithm, and $p_t$ is a two-sided t-test with equal variance. The diagonal dashed line has slope of 1 and intercept of 0, and indicates agreement between methods.}
\label{simDiff_nonSym_smallN}
\end{figure}

\subsubsection{Under the null hypothesis $P_x = P_y$}
 
We generated data $x_{i}, i=1,\ldots,n_x$ and $y_j, j=1,\ldots,n_y$ as realizations of the respective random variables $X_i\overset{\text{iid}}{\sim} N(0, 1)$ and $Y_{j} \overset{\text{iid}}{\sim} N(0, 1)$. For equal sample sizes, we set $n = n_x = n_y = 20, 40, 60$, and for unequal sample sizes we set $n_x = 20, 40, 60$ and $n_y = 100$. For both equal and unequal sample sizes, and for each each $n$ or $n_x$, we simulated 1,000 datasets (we used 1,000 datasets instead of 100 to better investigate the type I error rate). We used the p-value from simple Monte Carlo resampling with $10^5$ iterations, denoted as $\tilde{p}$, as an approximation for the true permutation p-value.

Results for equal and unequal sample size are shown in Figures \ref{simDiff_sym_null} and \ref{simDiff_nonSym_null}, respectively. \textit{Alg 1} is our resampling algorithm with $B_{\text{pred}}=10^3$ iterations in each partition, \textit{Asym} is our asymptotic approximation, \textit{t-test} shows the p-value from a two-sided t-test with equal variance, and $\tilde{p}$ is from simple Monte Carlo resampling with $10^5$ iterations. We compare p-values from the t-test against $\tilde{p}$, which shows close agreement. We do not show results from the SAMC algorithm, because the \verb|EXPERT| package \citep{yu2011} does not provide results for p-values $>10^{-3}$.

\begin{figure}[htbp]
\centering
\includegraphics[scale = 0.5]{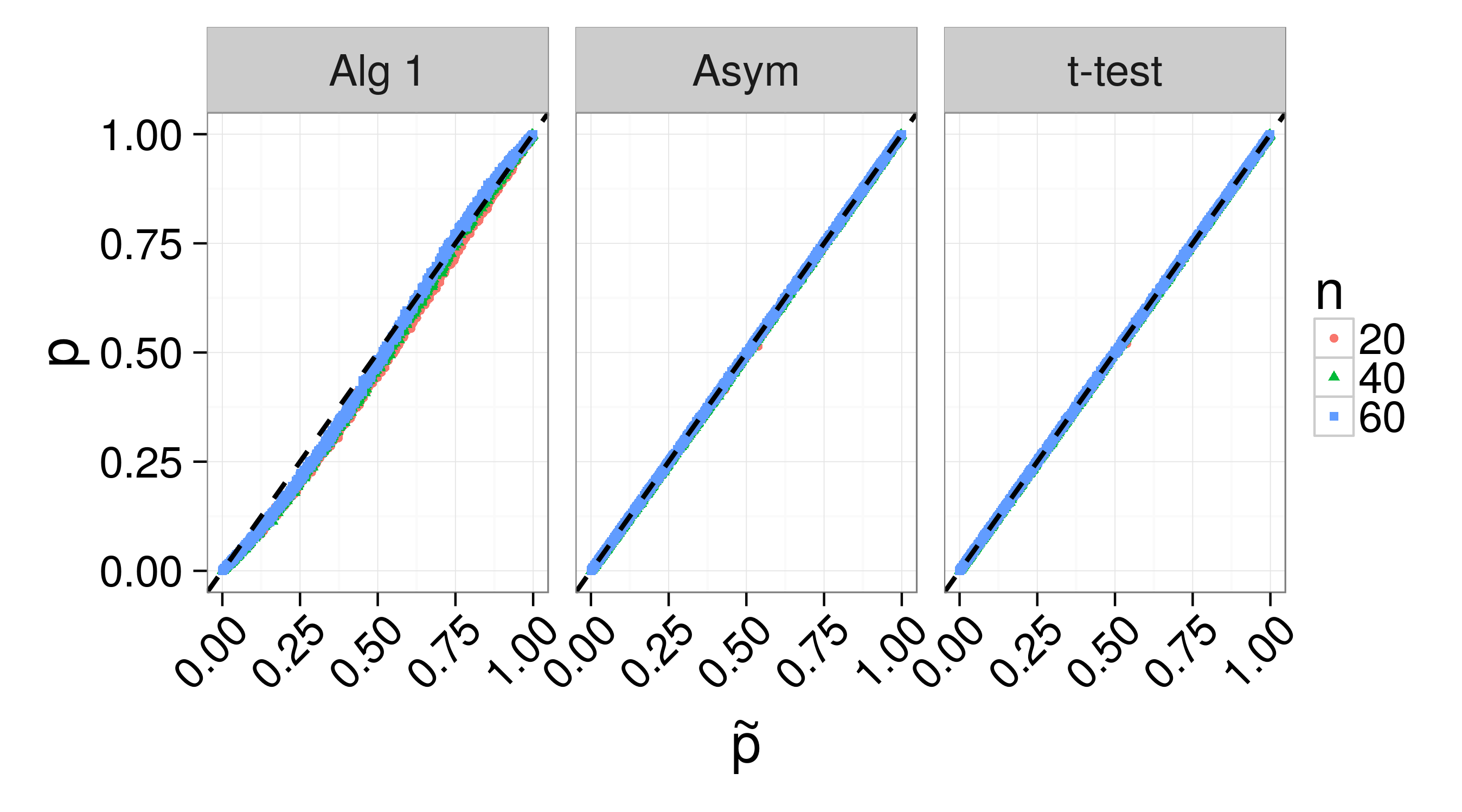}
\caption{Simulation results using the statistic $T=|\bar{x}-\bar{y}|$ with normal data under the null $P_x = P_y$ (means $\mu_x = \mu_y = 0$), with equal sample sizes of $n=n_x=n_y=20$, 40, 60. \textit{Alg 1} is our resampling algorithm with $B_{\text{pred}}=10^3$ iterations in each partition, \textit{Asym} is our asymptotic approximation, \textit{t-test} shows the p-value from a two-sided t-test with equal variance, and $\tilde{p}$ is from simple Monte Carlo resampling with $10^5$ iterations. The diagonal dashed line has slope of 1 and intercept of 0, and indicates agreement between methods.}
\label{simDiff_sym_null}
\end{figure}

\begin{figure}[htbp]
\centering
\includegraphics[scale = 0.5]{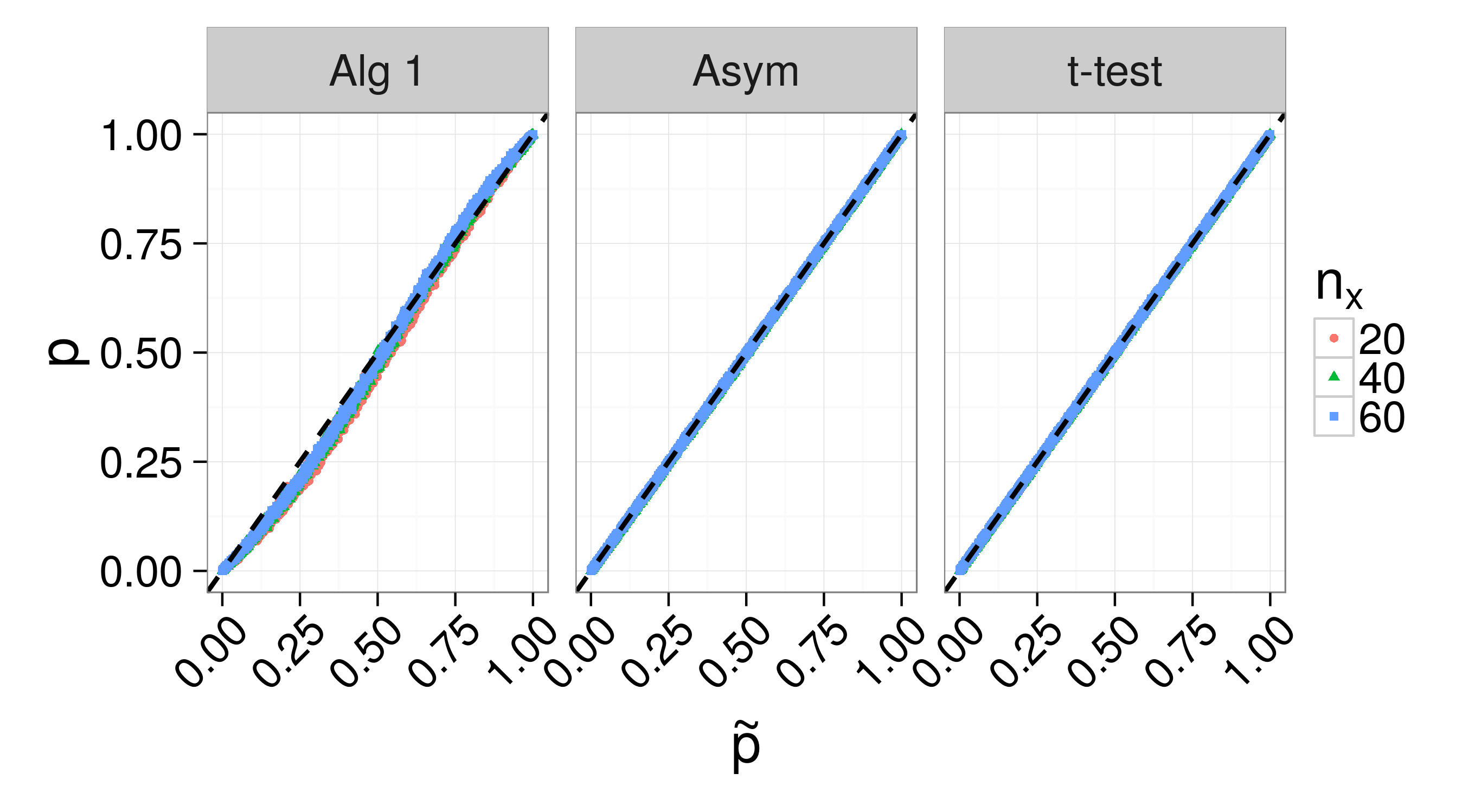}
\caption{Simulation results using the statistic $T=|\bar{x}-\bar{y}|$ with normal data under the null $P_x = P_y$ (means $\mu_x = \mu_y = 0$), with unequal sample sizes of $n_x=20, 40, 60$ and $n_y = 100$. \textit{Alg 1} is our resampling algorithm with $B_{\text{pred}}=10^3$ iterations in each partition, \textit{Asym} is our asymptotic approximation, \textit{t-test} shows the p-value from a two-sided t-test with equal variance, and $\tilde{p}$ is from simple Monte Carlo resampling with $10^5$ iterations. The diagonal dashed line has slope of 1 and intercept of 0, and indicates agreement between methods.}
\label{simDiff_nonSym_null}
\end{figure}

Tables \ref{type1ErrorDiffSym} and \ref{type1ErrorDiffNonSym} show the Type I error rates under the null $H_0: P_x = P_y$ for the equal and unequal sample size simulations, respectively. \textit{MC} is the unadjusted p-value from simple Monte Carlo resampling and $10^5$ iterations, \textit{t-test} is a two-sided t-test with equal variance, \textit{Alg 1} is our resampling algorithm, and \textit{Asymptotic} is our asymptotic approximation.

\begin{table}[htbp]
\centering
\caption{Type I error rates $\Pr(\text{p-value} \le \text{signif level} | H_0)$ for $T = |\bar{x} - \bar{y}|$ with normal data and equal sample sizes $n=n_x = n_y$. \textit{MC} is the unadjusted p-value from simple Monte Carlo resampling and $10^5$ iterations, \textit{t-test} is a two-sided t-test with equal variance, \textit{Alg 1} is our resampling algorithm, and \textit{Asymptotic} is our asymptotic approximation.}
\begin{tabular}{cccccc}
\hline \hline
signif level & $n$ & MC & t-test & Alg 1 & Asymptotic \\
\hline
\multirow{3}{*}{0.01} & 20 & 0.010 & 0.010 & 0.015 & 0.010 \\
 & 40 & 0.013 & 0.013 & 0.015 & 0.013 \\
 & 60 & 0.010 & 0.010 & 0.011 & 0.010 \\
\hline
\multirow{3}{*}{0.05}  & 20 & 0.048 & 0.050 & 0.064 & 0.050 \\
 & 40 & 0.055 & 0.055 & 0.075 & 0.056 \\
 & 60 & 0.049 & 0.050 & 0.061 & 0.050 \\
\hline
\multirow{3}{*}{0.1}  & 20 & 0.098 & 0.098 & 0.14 & 0.11 \\
 & 40 & 0.11 & 0.11 & 0.14 & 0.11 \\
 & 60 & 0.10 & 0.10 & 0.12 & 0.10
\end{tabular}
\label{type1ErrorDiffSym}
\end{table}

\begin{table}[htbp]
\centering
\caption{Type I error rates $\Pr(\text{p-value} \le \text{signif level} | H_0)$ for $T = |\bar{x} - \bar{y}|$ with normal data and unequal sample sizes $n_x \ne n_y$ ($n_x$ shown, and $n_y = 100$). \textit{MC} is the unadjusted p-value from simple Monte Carlo resampling and $10^5$ iterations, \textit{t-test} is a two-sided t-test with equal variance, \textit{Alg 1} is our resampling algorithm, and \textit{Asymptotic} is our asymptotic approximation.}
\begin{tabular}{cccccc}
\hline \hline
signif level & $n_x$ & MC & t-test & Alg 1 & Asymptotic \\
\hline
\multirow{3}{*}{0.01} & 20 & 0.013 & 0.013 & 0.018 & 0.013 \\
 & 40 & 0.016 & 0.016 & 0.018 & 0.016 \\
 & 60 & 0.010 & 0.010 & 0.013 & 0.010 \\
\hline
\multirow{3}{*}{0.05} & 20 & 0.049 & 0.049 & 0.075 & 0.049 \\
 & 40 & 0.047 & 0.047 & 0.066 & 0.047 \\
 & 60 & 0.044 & 0.044 & 0.057 & 0.044 \\
\hline
\multirow{3}{*}{0.1} & 20 & 0.090 & 0.090 & 0.14 & 0.092 \\
 & 40 & 0.10 & 0.10 & 0.14 & 0.11 \\
 & 60 & 0.090 & 0.090 & 0.13 & 0.090
\end{tabular}
\label{type1ErrorDiffNonSym}
\end{table}

\subsection{Ratio of means with exponential data}

In this subsection, we use the statistic $T = \max(\bar{x}/ \bar{y}, \bar{y}/ \bar{x})$ with data generated as exponential random variables.

\subsubsection{Small sample sizes}

We generated data $x_{i}, i=1,\ldots,n_x$ and $y_j, j=1,\ldots,n_y$ as realizations of the respective random variables $X_i\overset{\text{iid}}{\sim} \text{Exp}(\lambda_x)$ and $Y_{j} \overset{\text{iid}}{\sim} \text{Exp}(\lambda_y)$. For equal sample sizes, we set $n = n_x = n_y = 20, 40, 60$, and for unequal sample sizes, we set $n_x = 20, 40, 60$ and $n_y = 100$. For each $n$ or $n_x$, we set $\lambda_y=5$ or 10, and $\lambda_x = 1$. For both equal and unequal sample sizes, we simulated 100 datasets for each combination of parameters. We use the p-value from the beta prime distribution, denoted as $p_\beta$ (see Appendix \ref{B}) as an approximation to the true permutation p-value.

Results for equal and unequal sample size are shown in Figures \ref{sim_sym_smallN} and \ref{sim_nonSym_smallN}, respectively. \textit{Alg 1} is our resampling algorithm with $B_{\text{pred}}=10^3$ iterations in each partition, \textit{Asym} is our asymptotic approximation, \textit{Delta} is the delta method, \textit{SAMC} is the SAMC algorithm, and $p_{\beta}$ is the two-sided p-value from the beta prime distribution. The number of iterations used by our resampling algorithm is shown in Figures \ref{simExp_sym_iter_smallN} and \ref{simExp_nonSym_iter_smallN}.

\begin{figure}[htbp]
\centering
  \begin{subfigure}{0.58\textwidth}
  \centering
  \includegraphics[width=1\linewidth]{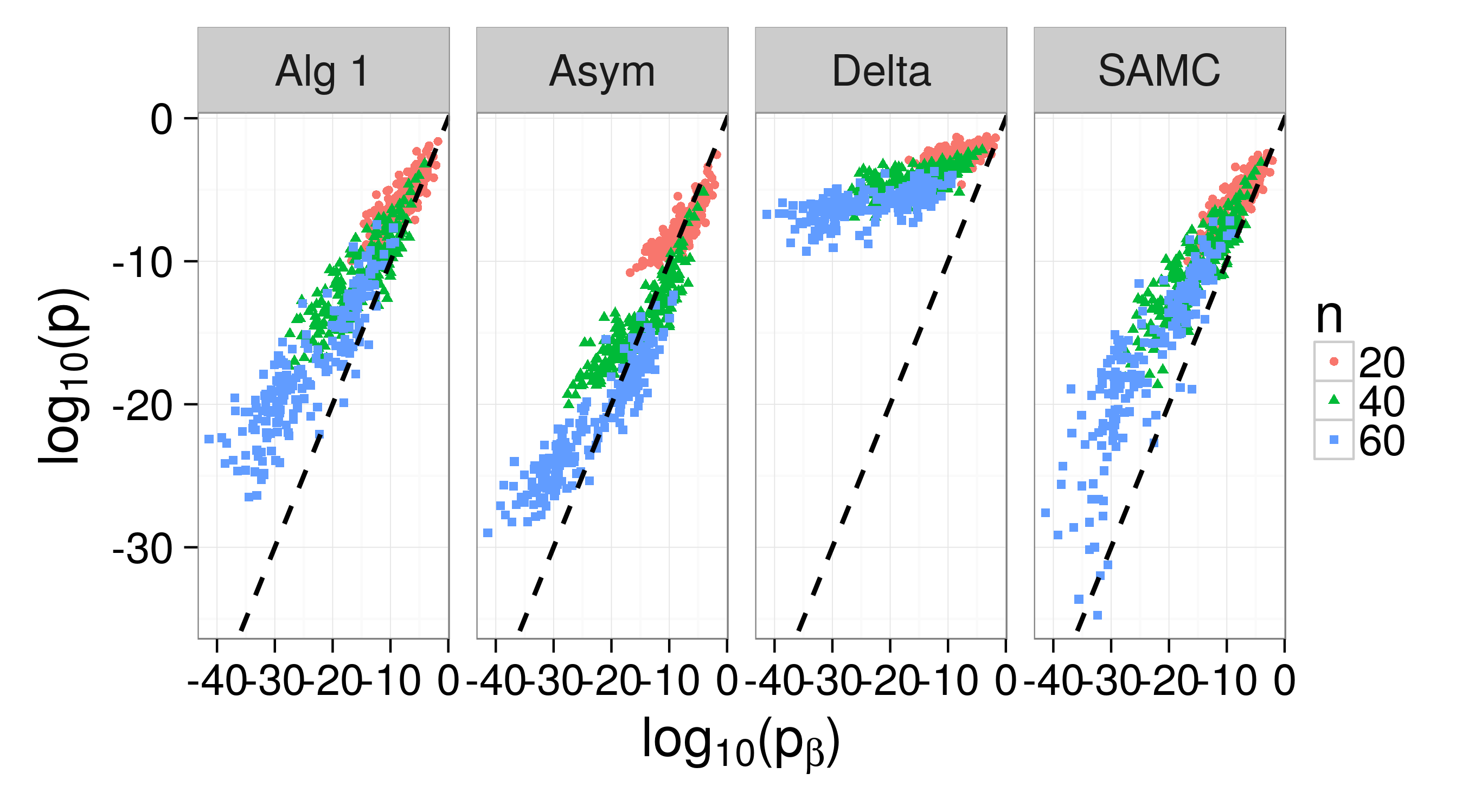}
  \caption{p-values}
  \end{subfigure}
  \begin{subfigure}{0.38\textwidth}
  \centering
  \includegraphics[width=1\linewidth]{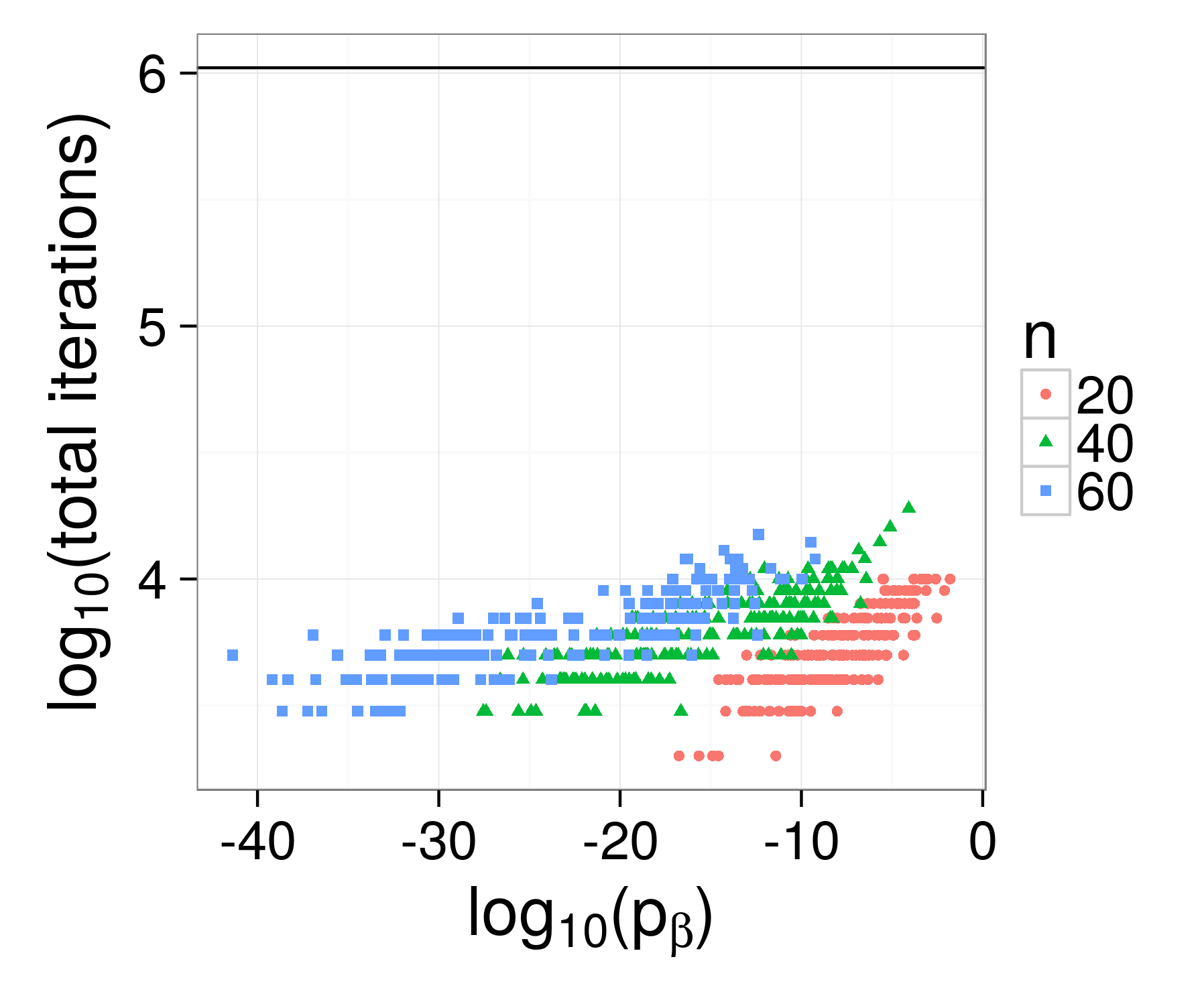}
  \caption{Iterations in resampling algorithm}
  \label{simExp_sym_iter_smallN}
  \end{subfigure}
\caption{Simulation results using the statistic $T= \max(\bar{x} / \bar{y}, \bar{y} / \bar{x})$, with exponential data, $n = n_x=n_y=20$, 40, 60, and rates $\lambda_y = 5, 10$, and $\lambda_x = 1$. \textit{Alg 1} is our resampling algorithm with $B_{\text{pred}}=10^3$ iterations in each partition, \textit{Asym} is our asymptotic approximation, \textit{Delta} is the delta method, \textit{SAMC} is the SAMC algorithm, and $p_{\beta}$ is the two-sided p-value from the beta prime distribution. The diagonal dashed line has slope of 1 and intercept of 0, and indicates agreement between methods. The horizontal line in \ref{simExp_sym_iter_smallN} shows the number of permutations used in the SAMC algorithm (set in advance, and independent of p-value). The SAMC algorithm did not produce values for 15 tests (points missing).}
\label{sim_sym_smallN}
\end{figure}

\begin{figure}[htbp]
\centering
  \begin{subfigure}{0.58\textwidth}
  \centering
  \includegraphics[width=1\linewidth]{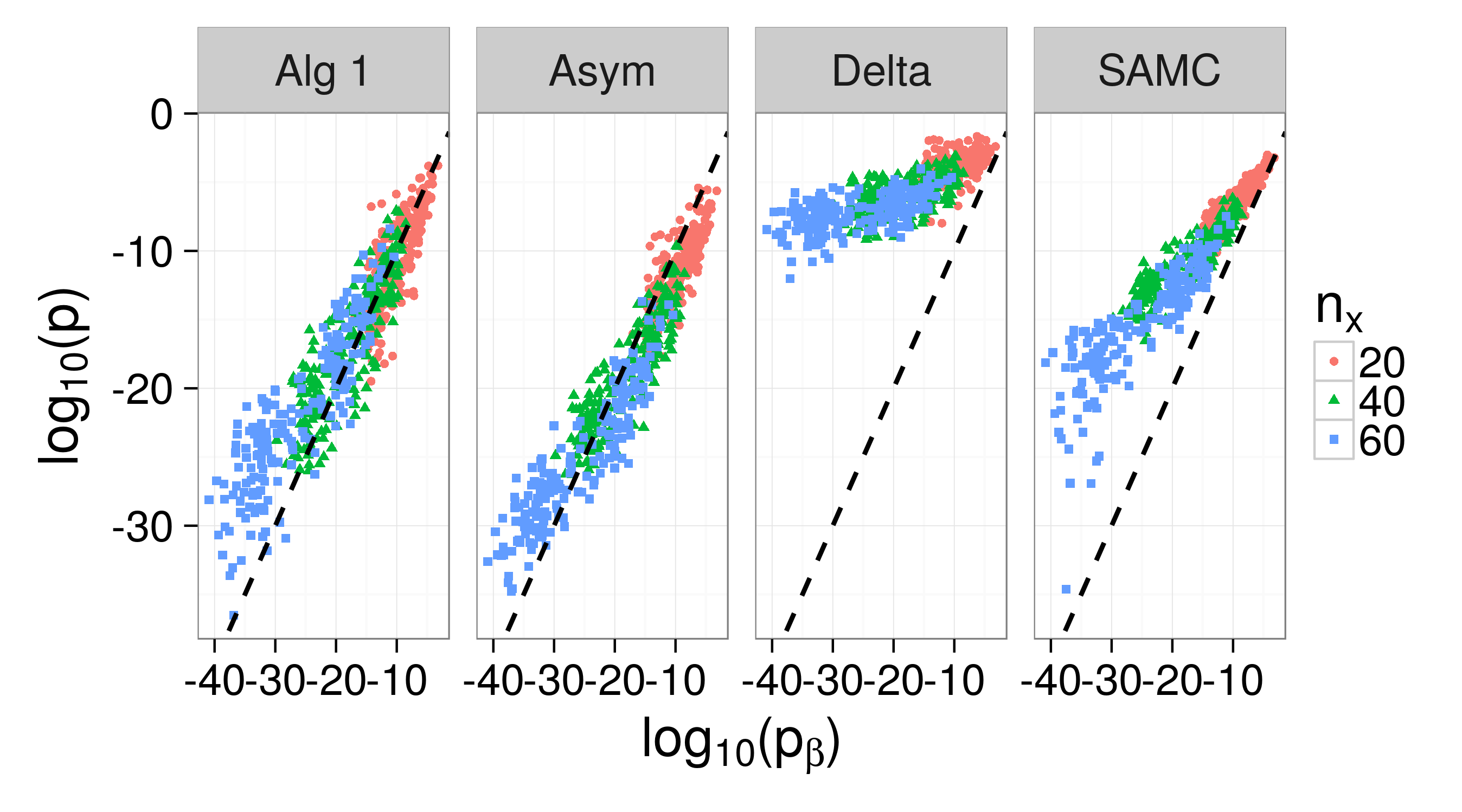}
  \caption{p-values}
  \end{subfigure}
  \begin{subfigure}{0.38\textwidth}
  \centering
  \includegraphics[width=1\linewidth]{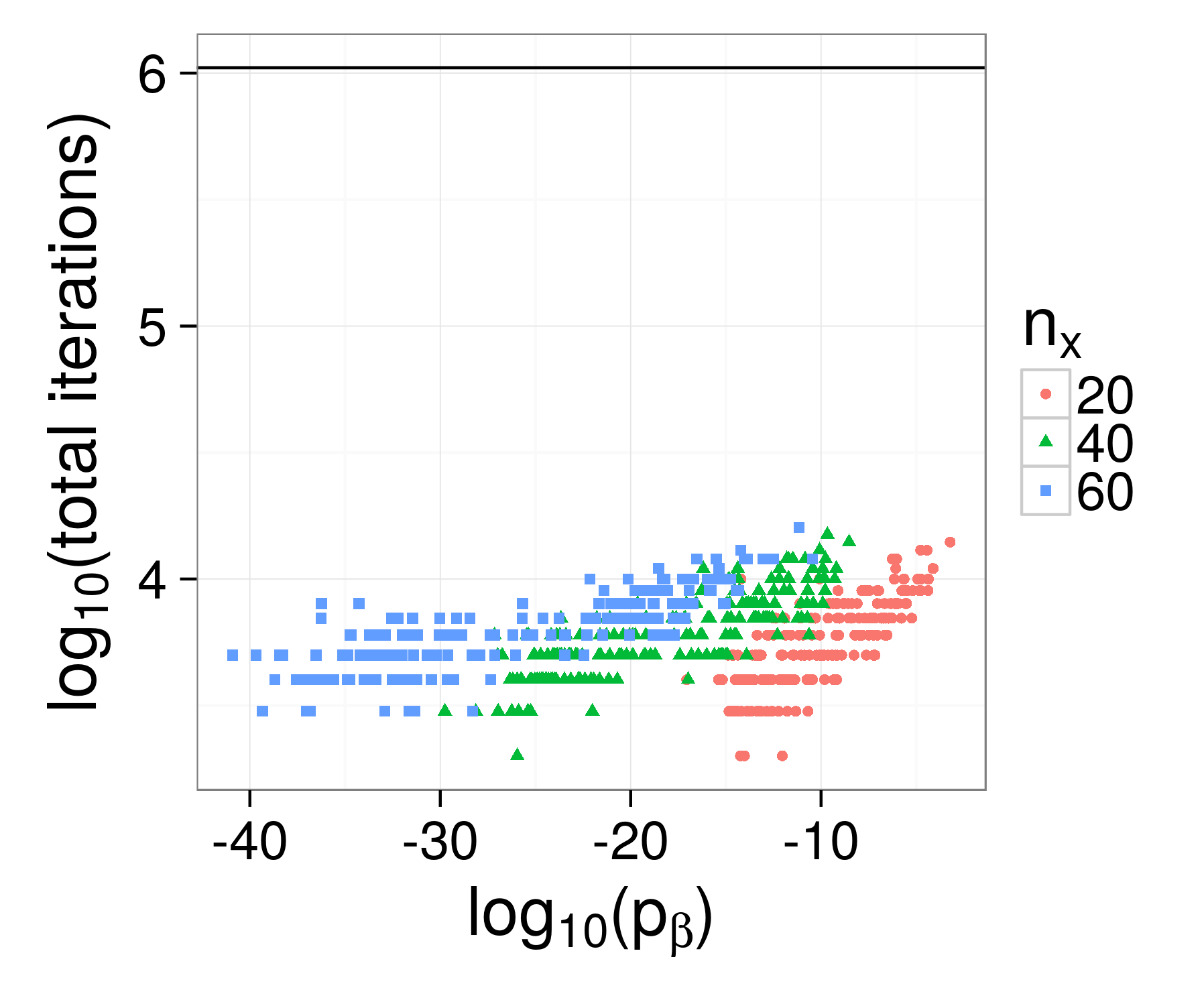}
  \caption{Iterations in resampling algorithm}
  \label{simExp_nonSym_iter_smallN}
  \end{subfigure}
\caption{Simulation results using the statistic $T= \max(\bar{x} / \bar{y}, \bar{y} / \bar{x})$, with exponential data, $n_x=20$, 40, 60, $n_y = 100$, and rates $\lambda_y = 5, 10$, and $\lambda_x = 1$. \textit{Alg 1} is our resampling algorithm with $B_{\text{pred}}=10^3$ iterations in each partition, \textit{Asym} is our asymptotic approximation, \textit{Delta} is the delta method, \textit{SAMC} is the SAMC algorithm, and $p_{\beta}$ is the two-sided p-value from the beta prime distribution. The diagonal dashed line has slope of 1 and intercept of 0, and indicates agreement between methods. The horizontal line in \ref{simExp_sym_iter_smallN} shows the number of permutations used in the SAMC algorithm (set in advance, and independent of p-value).}
\label{sim_nonSym_smallN}
\end{figure}

\subsubsection{Under the null $P_x = P_y$}

We generated data $x_{i}, i=1,\ldots,n_x$ and $y_j, j=1,\ldots,n_y$ as realizations of the respective random variables $X_i\overset{\text{iid}}{\sim} \text{Exp}(1)$ and $Y_{j} \overset{\text{iid}}{\sim} \text{Exp}(1)$. For equal sample sizes, we set $n = n_x = n_y = 20, 40, 60$. For unequal sample sizes, we set $n_x = 20, 40, 60$ and $n_y = 100$. For both equal and unequal sample sizes, we simulated 1,000 datasets for each combination of parameters (we used 1,000 datasets, as opposed to 100, to better investigate the type I error rate). We used the p-value from simple Monte Carlo resampling with $10^5$ iterations, denoted as $\tilde{p}$, as an approximation for the true permutation p-value.

Results for equal and unequal sample size are shown in Figures \ref{simExp_sym_null} and \ref{simExp_nonSym_null}, respectively. \textit{Alg 1} is our resampling algorithm with $B_{\text{pred}}=10^3$ iterations in each partition, \textit{Asym} is our asymptotic approximation, \textit{Delta} is the delta method, \textit{Beta prime} gives the p-value from the beta prime distribution, and $\tilde{p}$ is from simple Monte Carlo resampling with $10^5$ iterations. Given the large p-values, using $10^5$ Monte Carlo resamples should be sufficient to obtain reliable estimates of the true permutation p-value. Therefore, this comparison demonstrates that the permutation p-value is not exactly the same as the p-value from the beta prime distribution. However, it appears reasonably close, and so we use it as an approximation to the truth in other simulations, in which the p-values are much smaller and simple Monte Carlo methods are not feasible.

We do not show results from the SAMC algorithm, because as noted above, the \verb|EXPERT| package \citep{yu2011} does not provide results for p-values $>10^{-3}$.

\begin{figure}[htbp]
\centering
\includegraphics[scale = 0.5]{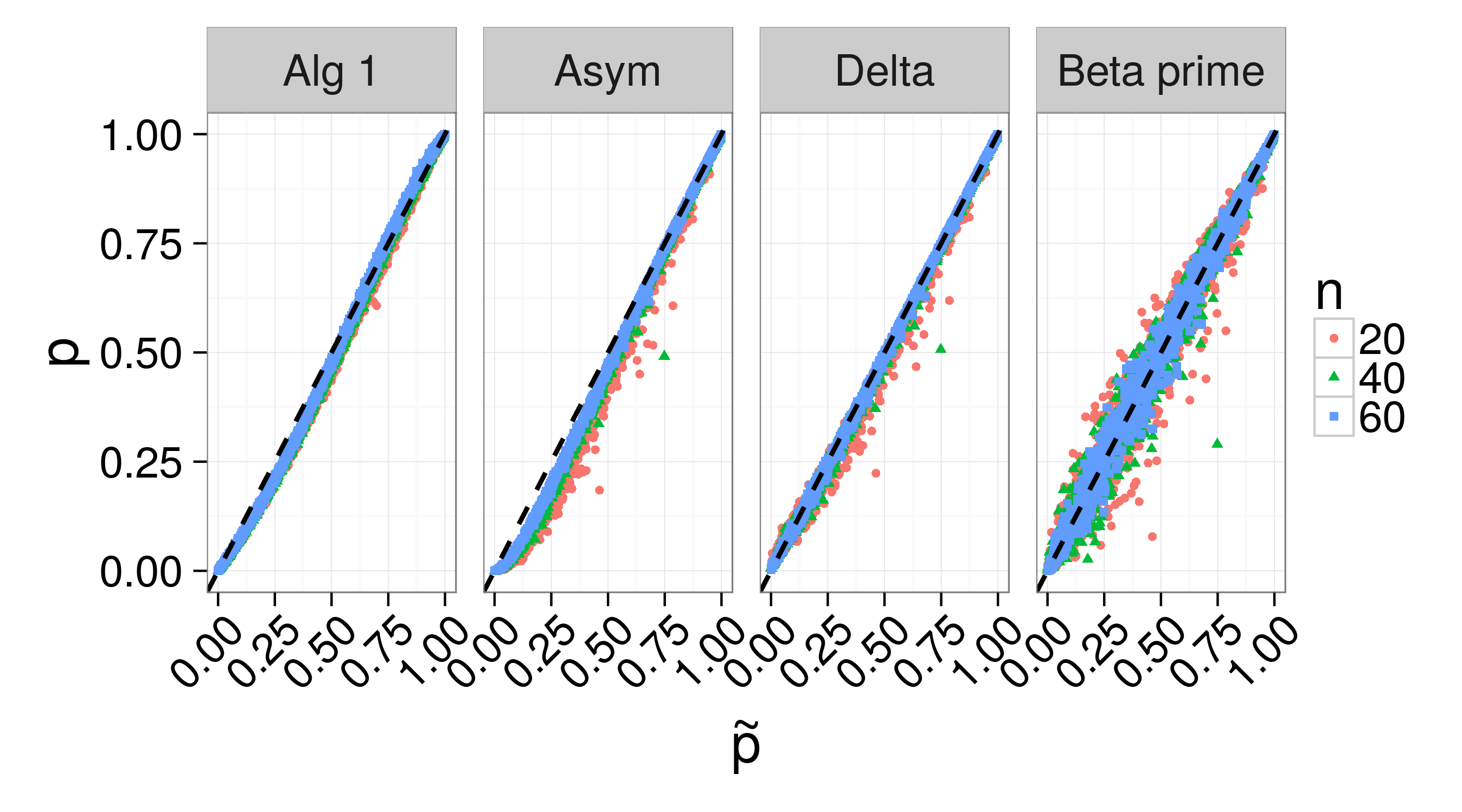}
\caption{Simulation results using the statistic $T= \max(\bar{x} / \bar{y}, \bar{y} / \bar{x})$, with exponential data under the null of $P_x = P_y$ (rates $\lambda_x = \lambda_y = 1$), with equal sample sizes of $n=n_x=n_y=20$, 40, 60. \textit{Alg 1} is our resampling algorithm with $B_{\text{pred}}=10^3$ iterations in each partition, \textit{Asym} is our asymptotic approximation, \textit{Delta} is the delta method, \textit{Beta prime} gives the p-value from the beta prime distribution, and $\tilde{p}$ is from simple Monte Carlo resampling with $10^5$ iterations. The diagonal dashed line has slope of 1 and intercept of 0, and indicates agreement between methods.}
\label{simExp_sym_null}
\end{figure}

\begin{figure}[htbp]
\centering
\includegraphics[scale = 0.5]{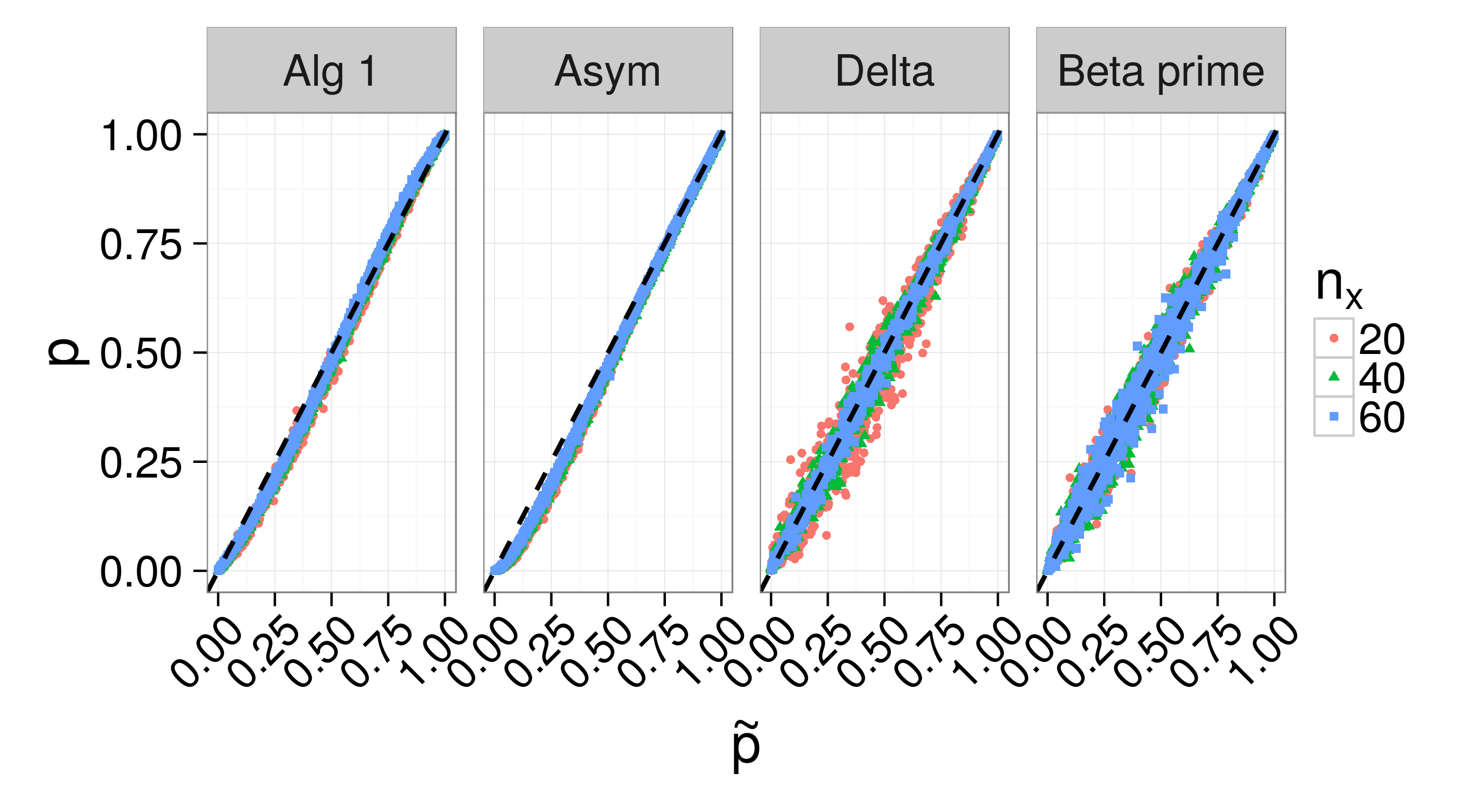}
\caption{Simulation results using the statistic $T= \max(\bar{x} / \bar{y}, \bar{y} / \bar{x})$ with exponential data under the null of $P_x = P_y$ (rates $\lambda_x = \lambda_y = 1$), with unequal sample sizes of $n_x = 20$, 40, 60 and $n_y = 100$. \textit{Alg 1} is our resampling algorithm with $B_{\text{pred}}=10^3$ iterations in each partition, \textit{Asym} is our asymptotic approximation, \textit{Delta} is the delta method, \textit{Beta prime} gives the p-value from the beta prime distribution, and $\tilde{p}$ is from simple Monte Carlo resampling with $10^5$ iterations. The diagonal dashed line has slope of 1 and intercept of 0, and indicates agreement between methods.}
\label{simExp_nonSym_null}
\end{figure}

Tables \ref{type1ErrorRatioSym} and \ref{type1ErrorRatioNonSym} show the Type I error rates under the null $H_0: P_x = P_y$ for the equal and unequal sample size simulations, respectively. \textit{MC} is the unadjusted p-value from simple Monte Carlo resampling and $10^5$ iterations, \textit{Beta prime} is the p-value from the beta prime distribution, \textit{Alg 1} is our resampling algorithm, and \textit{Asymptotic} is our asymptotic approximation.

\begin{table}[htbp]
\centering
\caption{Type I error rates $\Pr(\text{p-value} \le \text{signif level} | H_0)$ for $T = \max(\bar{x} / \bar{y}, \bar{y} / \bar{x})$ with exponential data and equal sample sizes $n=n_x = n_y$. \textit{MC} is simple Monte Carlo resampling with $10^5$ iterations, \textit{Alg 1} is our resampling algorithm, and \textit{Asymptotic} is our asymptotic approximation, \textit{Delta} is the delta method, and \textit{Beta prime} is the the beta prime distribution.}
\begin{tabular}{ccccccc}
\hline \hline
signif level & $n$ & MC & Alg 1 & Asymptotic & Delta & Beta prime \\
\hline
 \multirow{3}{*}{0.01} & 20 & 0.010 & 0.016 & 0.066 &  0.003 & 0.009 \\
 & 40 & 0.010 & 0.018 & 0.050 &  0.002 & 0.008 \\
 & 60 & 0.013 & 0.013 & 0.031 &  0.006 & 0.015 \\
 \hline
 \multirow{3}{*}{0.05} & 20 & 0.064 & 0.084 & 0.14 &  0.045 & 0.058 \\
 & 40 & 0.061 & 0.079 & 0.11 &  0.054 & 0.061 \\
 & 60 & 0.051 & 0.063 & 0.091 &  0.050 & 0.047 \\
 \hline
 \multirow{3}{*}{0.10} & 20 & 0.11 & 0.15 & 0.21 &  0.12 & 0.11 \\
 & 40 & 0.11 & 0.14 & 0.17 & 0.11 & 0.11 \\
 & 60 & 0.093 & 0.11 & 0.14 & 0.095 & 0.092
\end{tabular}
\label{type1ErrorRatioSym}
\end{table}

\begin{table}[htbp]
\centering
\caption{Type I error rates $\Pr(\text{p-value} \le \text{signif level} | H_0)$ for $T = \max(\bar{x} / \bar{y}, \bar{y} / \bar{x})$ with exponential data and unequal sample sizes $n_x \ne n_y$ ($n_x$ shown, and $n_y = 100$). \textit{MC} is simple Monte Carlo resampling with $10^5$ iterations, \textit{Alg 1} is our resampling algorithm, and \textit{Asymptotic} is our asymptotic approximation, \textit{Delta} is the delta method with, and \textit{Beta prime} is the beta prime distribution.}
\begin{tabular}{ccccccc}
\hline \hline
signif level & $n$ & MC & Alg 1 & Asymptotic & Delta & Beta prime \\
\hline
\multirow{3}{*}{0.01} & 20 & 0.011 & 0.016 & 0.054 & 0.008 & 0.012 \\
 & 40 & 0.008 & 0.012 & 0.033 & 0.004 & 0.006 \\
 & 60 & 0.012 & 0.016 & 0.035 & 0.007 & 0.014 \\
 \hline
\multirow{3}{*}{0.05} & 20 & 0.061 & 0.082 & 0.127 & 0.065 & 0.056 \\
 & 40 & 0.048 & 0.062 & 0.097 & 0.047 & 0.050 \\
 & 60 & 0.047 & 0.065 & 0.083 & 0.044 & 0.051 \\
 \hline
\multirow{3}{*}{0.10} & 20 & 0.12 & 0.16 & 0.19 & 0.14 & 0.12 \\
 & 40 & 0.10 & 0.14 & 0.17 & 0.11 & 0.10 \\
 & 60 & 0.091 & 0.12 & 0.14 & 0.093 & 0.088
\end{tabular}
\label{type1ErrorRatioNonSym}
\end{table}

\subsection{Difference in means with gamma data \label{diffMean_gamma}}

In this subsection, we use the statistic $T = |\bar{x} - \bar{y}|$ with data generated as gamma random variables.

\subsubsection{Small sample sizes}

We generated data $x_{i}, i=1,\ldots,n_x$ and $y_j, j=1,\ldots,n_y$ as realizations of the respective random variables $X_i\overset{\text{iid}}{\sim} \text{Gamma}(\alpha, \lambda_x)$ and $Y_{j} \overset{\text{iid}}{\sim} \text{Gamma}(\alpha, \lambda_y)$, where $\alpha = 0.5, 3, 5$, $\lambda_x = 1$, and $\lambda$ is the rate parameter. For equal sample sizes, we set $n=n_x = n_y = 20, 40, 60$, and for unequal sample sizes we set $n_x = 20, 40, 60$ and $n_y = 100$. For $\alpha = 0.5$, we set $\lambda_y = 2.5, 3$ for all $n$ or $n_x$. For $\alpha = 3$, we set $\lambda_y = 1.5, 1.75$ for all $n$ or $n_x$. For $\alpha = 5$, we set $\lambda_y = 1.25, 1.5$ for all $n$ or $n_x$. For both equal and unequal sample sizes, we simulated 100 datasets for each combination of parameters.

Results for equal and unequal sample size are shown in Figures \ref{simGammaDiff_sym_smallN} and \ref{simGammaDiff_nonSym_smallN}, respectively. \textit{Alg 1} is our resampling algorithm with $B_{\text{pred}}=10^3$ iterations in each partition, \textit{Asym} is our asymptotic approximation, \textit{t-test} is a t-test with unequal variance, and \textit{Saddle} is the saddlepoint approximation (see Appendix \ref{B}). SAMC results are not shown, as the \verb|EXPERT| package does not provide p-values larger than $10^{-3}$. We use the p-values from simple Monte Carlo resampling, denoted as $\tilde{p}$, with $10^5$ iterations as a basis of comparison, and only show values for which $\tilde{p} > 10^{-3}$ to ensure that the $\tilde{p}$ are reliable (1,023 values shown in Figure \ref{simGammaDiff_sym_smallN}, and 573 values shown in Figure \ref{simGammaDiff_nonSym_smallN}).

\begin{figure}[htbp]
\centering
\includegraphics[scale = 0.6]{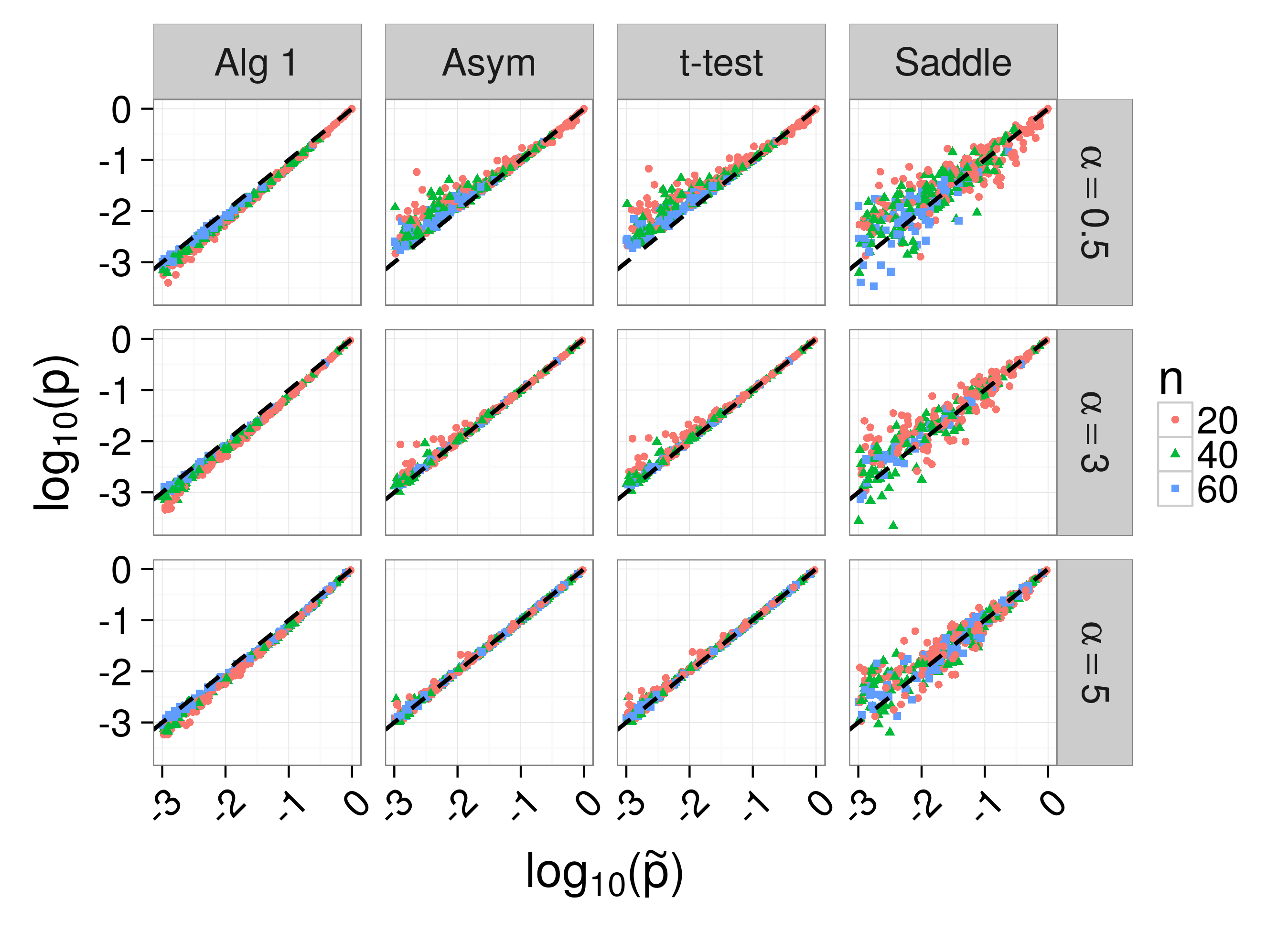}
\caption{Simulation results using the statistic $T=|\bar{x}-\bar{y}|$ with gamma data and equal sample sizes of $n=n_x=n_y=20$, 40, 60. \textit{Alg 1} is our resampling algorithm with $B_{\text{pred}}=10^3$ iterations in each partition, \textit{Asym} is our asymptotic approximation, \textit{t-test} is a t-test with unequal variance, and \textit{Saddle} is the saddlepoint approximation (see Appendix \ref{B}). $\tilde{p}$ is the p-values from simple Monte Carlo resampling with $10^5$ iterations. SAMC results not shown, as the EXPERT package does not produce p-values larger than $10^{-3}$. Only simulations with $\tilde{p} > 10^{-3}$ shown (1,023 values shown). The diagonal dashed line has slope of 1 and intercept of 0, and indicates agreement between methods.}
\label{simGammaDiff_sym_smallN}
\end{figure}

\begin{figure}[htbp]
\centering
 \includegraphics[scale = 0.6]{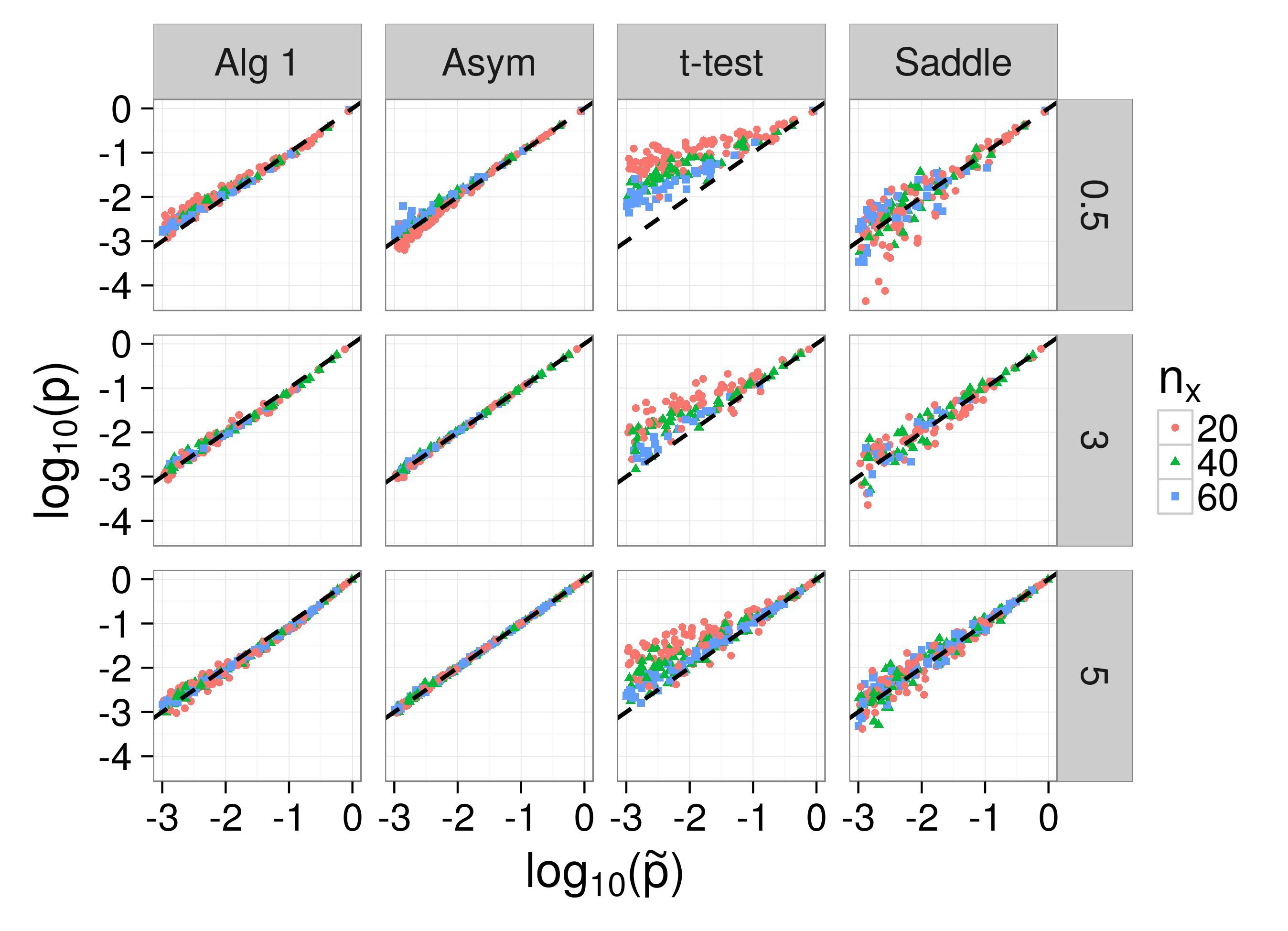}
\caption{Simulation results using the statistic $T=|\bar{x}-\bar{y}|$ with gamma data and unequal sample sizes of $n_x=20, 40, 60$ and $n_y=100$. \textit{Alg 1} is our resampling algorithm with $B_{\text{pred}}=10^3$ iterations in each partition, \textit{Asym} is our asymptotic approximation, \textit{SAMC} is the SAMC algorithm, and $p_t$ is a two-sided t-test with equal variance. SAMC results not shown, as the EXPERT package does not produce p-values larger than $10^{-3}$. Only simulations with $\tilde{p} > 10^{-3}$ shown (573 values shown). The diagonal dashed line has slope of 1 and intercept of 0, and indicates agreement between methods.}
\label{simGammaDiff_nonSym_smallN}
\end{figure}

Overall, Figures \ref{simGammaDiff_sym_smallN} and \ref{simGammaDiff_nonSym_smallN} suggest that our methods work well in this setting, though our resampling algorithm might be liberal for equal sample sizes and $\alpha = 0.5$. The t-test performs well in some scenarios, but tends to be too conservative, particularly for unequal sample sizes. Overall, the Saddlepoint approximation with fixed $\alpha$ and the MLE $\hat{\lambda}$ from the pooled data appears to have more variance than the other methods.

\subsubsection{Under the null hypothesis $P_x = P_y$}
 
We generated data $x_{i}, i=1,\ldots,n_x$ and $y_j, j=1,\ldots,n_y$ as realizations of the respective random variables $X_i\overset{\text{iid}}{\sim} \text{Gamma}(\alpha, \lambda)$ and $Y_{j} \overset{\text{iid}}{\sim} \text{Gamma}(\alpha, \lambda)$ for $\alpha = 0.5, 3, 5$ and $\lambda = 1,5$, where $\lambda$ is the rate parameter. For equal sample sizes, we set $n = n_x = n_y = 20, 40, 60$, and for unequal sample sizes we set $n_x = 20, 40, 60$ and $n_y = 100$. For both equal and unequal sample sizes, and for each each $n$ or $n_x$, and combination of $\alpha$ and $\lambda$, we simulated 1,000 datasets (we used 1,000 datasets instead of 100 to better investigate the type I error rate). We used the p-value from simple Monte Carlo resampling with $10^5$ iterations, denoted as $\tilde{p}$, as an approximation for the true permutation p-value.

Results for equal and unequal sample size are shown in Figures \ref{simGammaDiff_sym_null} and \ref{simGammaDiff_nonSym_null}, respectively. textit{Alg 1} is our resampling algorithm with $B_{\text{pred}}=10^3$ iterations in each partition, \textit{Asym} is our asymptotic approximation, \textit{Saddle} is the saddlepoint approximation described in Appendix \ref{B}, \textit{t-test} shows the p-value from a two-sided t-test with unequal variance, and $\tilde{p}$ is from simple Monte Carlo resampling with $10^5$ iterations. We do not show results from the SAMC algorithm, because the \verb|EXPERT| package \citep{yu2011} does not provide results for p-values $>10^{-3}$.

\begin{figure}[htbp]
\centering
\includegraphics[scale = 0.6]{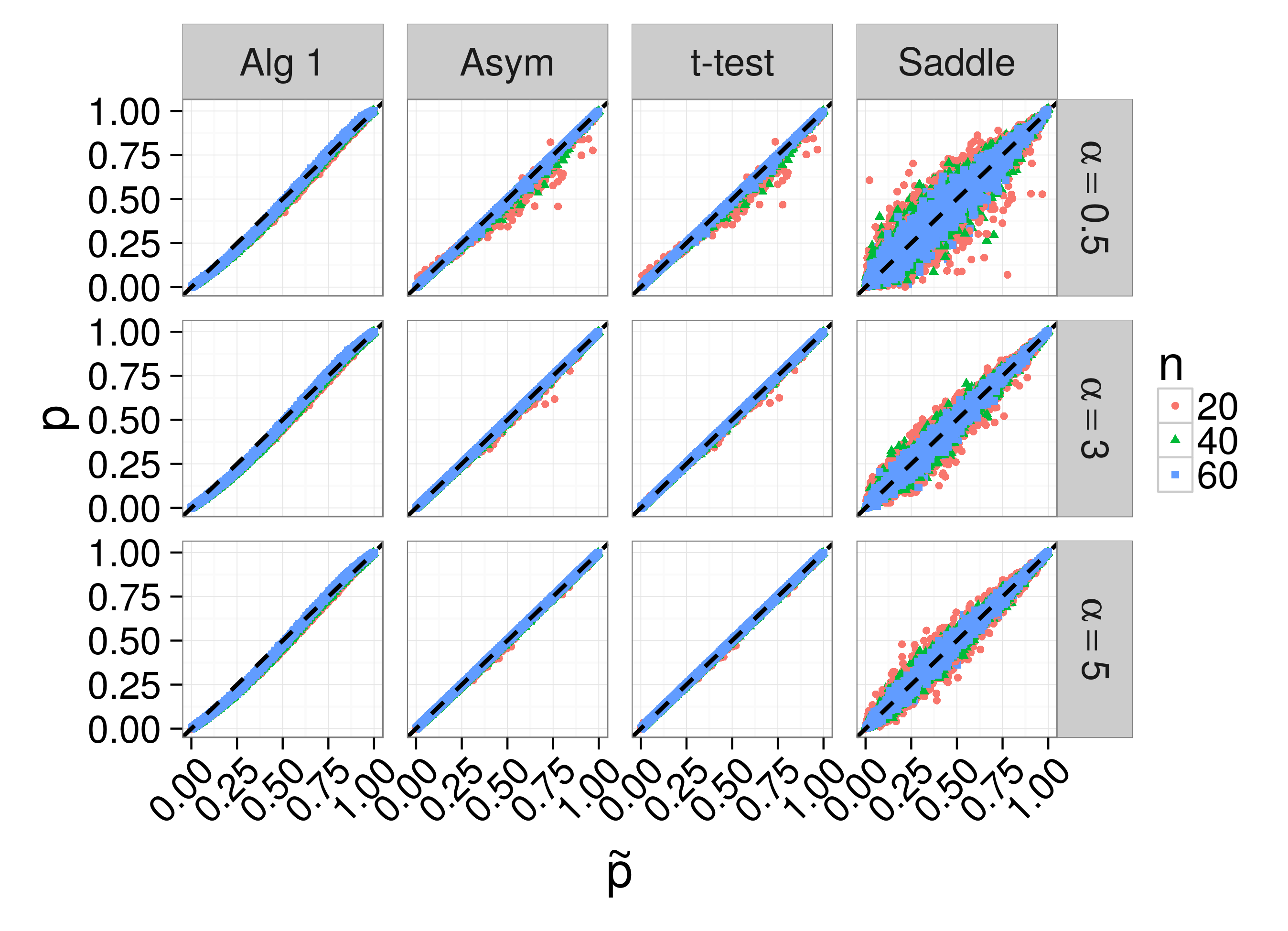}
\caption{Simulation results using the statistic $T=|\bar{x}-\bar{y}|$ with gamma data under the null $P_x = P_y$, with equal sample sizes of $n=n_x=n_y=20$, 40, 60. \textit{Alg 1} is our resampling algorithm with $B_{\text{pred}}=10^3$ iterations in each partition, \textit{Asym} is our asymptotic approximation, \textit{Saddle} is the saddlepoint approximation described in Appendix \ref{B}, \textit{t-test} shows the p-value from a two-sided t-test with unequal variance, and $\tilde{p}$ is from simple Monte Carlo resampling with $10^5$ iterations. The diagonal dashed line has slope of 1 and intercept of 0, and indicates agreement between methods.}
\label{simGammaDiff_sym_null}
\end{figure}

\begin{figure}[htbp]
\centering
\includegraphics[scale = 0.6]{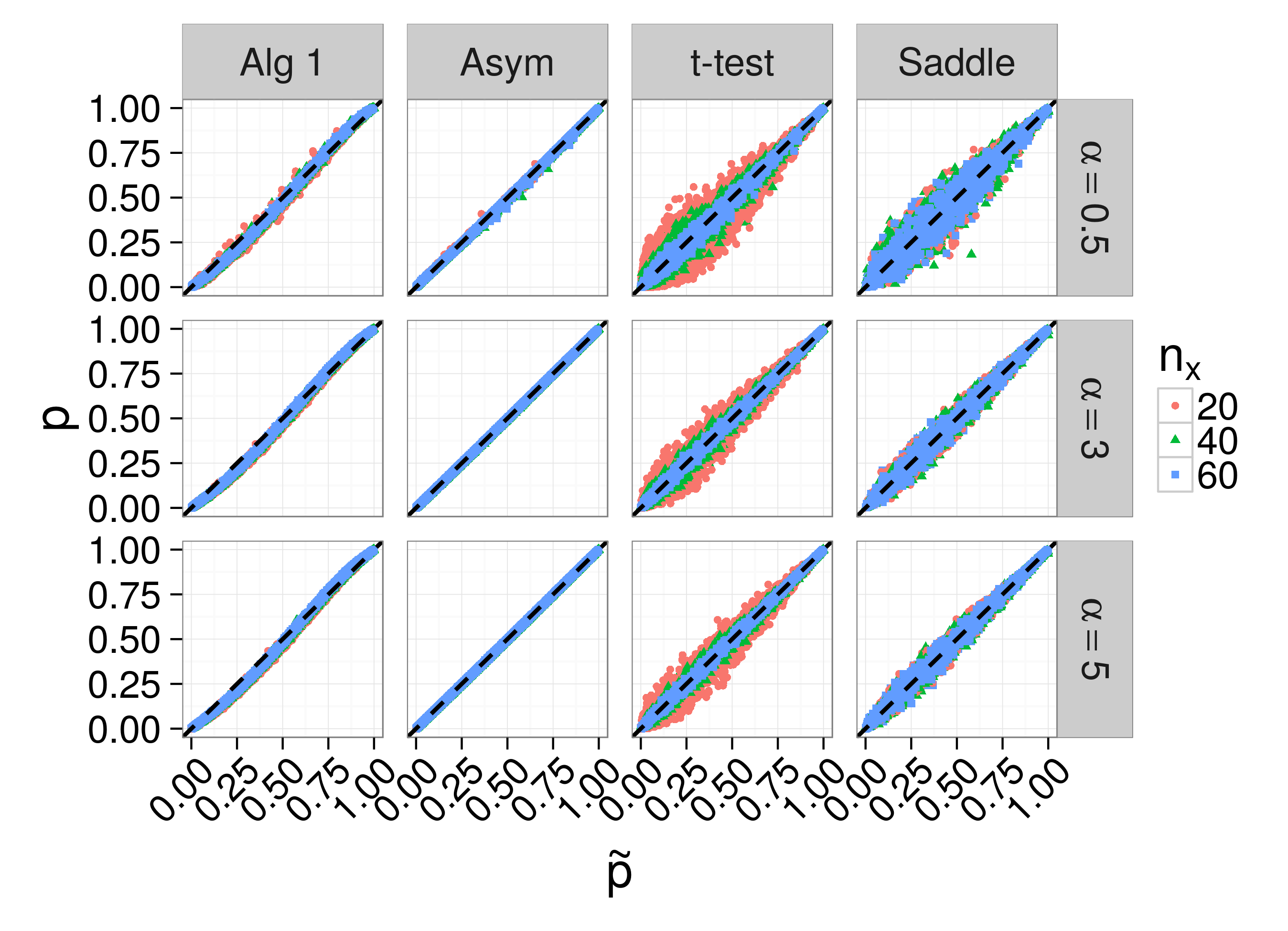}
\caption{Simulation results using the statistic $T=|\bar{x}-\bar{y}|$ with gamma data under the null $P_x = P_y$, with unequal sample sizes of $n_x=20, 40, 60$ and $n_y = 100$. \textit{Alg 1} is our resampling algorithm with $B_{\text{pred}}=10^3$ iterations in each partition, \textit{Asym} is our asymptotic approximation, \textit{Saddle} is the saddlepoint approximation described in Appendix \ref{B}, \textit{t-test} shows the p-value from a two-sided t-test with unequal variance, and $\tilde{p}$ is from simple Monte Carlo resampling with $10^5$ iterations. The diagonal dashed line has slope of 1 and intercept of 0, and indicates agreement between methods.}
\label{simGammaDiff_nonSym_null}
\end{figure}

Figures \ref{simGammaDiff_sym_null} and \ref{simGammaDiff_nonSym_null} suggest that our methods work well in this setting, and have less variability than both the t-test and saddlepoint approximation (using fixed $\alpha$ fixed and the MLE $\hat{\lambda}$ from the pooled data).

Tables \ref{type1ErrorGammaDiffSym} and \ref{type1ErrorGammaDiffNonSym} show the Type I error rates under the null $H_0: P_x = P_y$ for the equal and unequal sample size simulations, respectively. \textit{MC} is the unadjusted p-value from simple Monte Carlo resampling and $10^5$ iterations, \textit{Saddle} is the saddlepoint approximation described in Appendix \ref{B}, \textit{Alg 1} is our resampling algorithm with $B_{\text{pred}}=10^3$ iterations in each partition, \textit{Asym} is our asymptotic approximation, and \textit{t-test} shows the p-value from a two-sided t-test with unequal variance.

\begin{table}[htbp]
\centering
\caption{Type I error rates $\Pr(\text{p-value} \le \text{signif level} | H_0)$ for $T = |\bar{x} - \bar{y}|$ with gamma data and equal sample sizes $n=n_x = n_y$. \textit{MC} is the unadjusted p-value from simple Monte Carlo resampling and $10^5$ iterations, \textit{Saddle} is the saddlepoint approximation described in Appendix \ref{B}, \textit{Alg 1} is our resampling algorithm with $B_{\text{pred}}=10^3$ iterations in each partition, \textit{Asym} is our asymptotic approximation, and \textit{t-test} shows the p-value from a two-sided t-test with unequal variance.}
\begin{tabular}{cccccccc}
\hline \hline
$\alpha$ & signif level & $n_x$ & MC & Saddle & Alg 1 & Asym & t-test \\
\hline
0.5 & 0.01 & 20 & 0.0110 & 0.0100 & 0.0165 & 0.0060 & 0.0045 \\
0.5 & 0.01 & 40 & 0.0125 & 0.0110 & 0.0150 & 0.0090 & 0.0085 \\
0.5 & 0.01 & 60 & 0.0115 & 0.0085 & 0.0140 & 0.0105 & 0.0105 \\
\cline{2-8}
0.5 & 0.05 & 20 & 0.0495 & 0.0560 & 0.0665 & 0.0460 & 0.0410 \\
0.5 & 0.05 & 40 & 0.0515 & 0.0490 & 0.0660 & 0.0520 & 0.0485 \\
0.5 & 0.05 & 60 & 0.0455 & 0.0450 & 0.0595 & 0.0435 & 0.0425 \\
\cline{2-8}
0.5 & 0.10 & 20 & 0.1000 & 0.1020 & 0.1280 & 0.1020 & 0.0945 \\
0.5 & 0.10 & 40 & 0.0995 & 0.0950 & 0.1260 & 0.1020 & 0.0975 \\
0.5 & 0.10 & 60 & 0.0980 & 0.0950 & 0.1230 & 0.0990 & 0.0965 \\
\hline
3.0 & 0.01 & 20 & 0.0115 & 0.0070 & 0.0165 & 0.0095 & 0.0095 \\
3.0 & 0.01 & 40 & 0.0120 & 0.0115 & 0.0150 & 0.0120 & 0.0120 \\
3.0 & 0.01 & 60 & 0.0075 & 0.0075 & 0.0080 & 0.0070 & 0.0070 \\
\cline{2-8}
3.0 & 0.05 & 20 & 0.0510 & 0.0465 & 0.0715 & 0.0515 & 0.0495 \\
3.0 & 0.05 & 40 & 0.0545 & 0.0575 & 0.0680 & 0.0560 & 0.0525 \\
3.0 & 0.05 & 60 & 0.0470 & 0.0475 & 0.0665 & 0.0480 & 0.0475 \\
\cline{2-8}
3.0 & 0.10 & 20 & 0.0940 & 0.0990 & 0.1280 & 0.0980 & 0.0940 \\
3.0 & 0.10 & 40 & 0.0990 & 0.1000 & 0.1320 & 0.0990 & 0.0980 \\
3.0 & 0.10 & 60 & 0.0980 & 0.0985 & 0.1230 & 0.0980 & 0.0980 \\
\hline
5.0 & 0.01 & 20 & 0.0115 & 0.0095 & 0.0175 & 0.0115 & 0.0115 \\
5.0 & 0.01 & 40 & 0.0090 & 0.0065 & 0.0130 & 0.0080 & 0.0080 \\
5.0 & 0.01 & 60 & 0.0045 & 0.0055 & 0.0085 & 0.0040 & 0.0040 \\
\cline{2-8}
5.0 & 0.05 & 20 & 0.0525 & 0.0525 & 0.0675 & 0.0525 & 0.0505 \\
5.0 & 0.05 & 40 & 0.0525 & 0.0545 & 0.0715 & 0.0535 & 0.0520 \\
5.0 & 0.05 & 60 & 0.0460 & 0.0445 & 0.0580 & 0.0470 & 0.0470 \\
\cline{2-8}
5.0 & 0.10 & 20 & 0.0965 & 0.0960 & 0.1220 & 0.0980 & 0.0955 \\
5.0 & 0.10 & 40 & 0.1070 & 0.1060 & 0.1370 & 0.1080 & 0.1080 \\
5.0 & 0.10 & 60 & 0.0925 & 0.0905 & 0.1300 & 0.0940 & 0.0915
\end{tabular}
\label{type1ErrorGammaDiffSym}
\end{table}

\begin{table}[htbp]
\centering
\caption{Type I error rates $\Pr(\text{p-value} \le \text{signif level} | H_0)$ for $T = |\bar{x} - \bar{y}|$ with gamma data and unequal sample sizes $n_x \ne n_y$ ($n_x$ shown, and $n_y = 100$). $\alpha$ is the shape parameter in the gamma distribution, \textit{MC} is the unadjusted p-value from simple Monte Carlo resampling and $10^5$ iterations, \textit{Saddle} is the saddlepoint approximation described in Appendix \ref{B}, \textit{Alg 1} is our resampling algorithm with $B_{\text{pred}}=10^3$ iterations in each partition, \textit{Asym} is our asymptotic approximation, and \textit{t-test} shows the p-value from a two-sided t-test with unequal variance.}
\begin{tabular}{cccccccc}
\hline \hline
$\alpha$ & signif level & $n_x$ & MC & Saddle & Alg 1 & Asym & t-test \\
\hline
0.5 & 0.01 & 20 & 0.0095 & 0.0095 & 0.0105 & 0.0085 & 0.0245 \\
0.5 & 0.01 & 40 & 0.0090 & 0.0060 & 0.0105 & 0.0070 & 0.0140 \\
0.5 & 0.01 & 60 & 0.0130 & 0.0160 & 0.0170 & 0.0105 & 0.0135 \\
\cline{2-8}
0.5 & 0.05 & 20 & 0.0460 & 0.0465 & 0.0675 & 0.0440 & 0.0740 \\
0.5 & 0.05 & 40 & 0.0455 & 0.0470 & 0.0620 & 0.0445 & 0.0540 \\
0.5 & 0.05 & 60 & 0.0505 & 0.0500 & 0.0670 & 0.0495 & 0.0530 \\
\cline{2-8}
0.5 & 0.1 & 20 & 0.0915 & 0.0930 & 0.1260 & 0.0845 & 0.1220 \\
0.5 & 0.1 & 40 & 0.0980 & 0.0945 & 0.1280 & 0.0960 & 0.1040 \\
0.5 & 0.1 & 60 & 0.1100 & 0.1080 & 0.1410 & 0.1100 & 0.1080 \\
\hline
3.0 & 0.01 & 20 & 0.0085 & 0.0095 & 0.0155 & 0.0085 & 0.0135 \\
3.0 & 0.01 & 40 & 0.0135 & 0.0120 & 0.0185 & 0.0135 & 0.0140 \\
3.0 & 0.01 & 60 & 0.0070 & 0.0055 & 0.0090 & 0.0070 & 0.0070 \\
\cline{2-8}
3.0 & 0.05 & 20 & 0.0440 & 0.0440 & 0.0665 & 0.0435 & 0.0480 \\
3.0 & 0.05 & 40 & 0.0480 & 0.0555 & 0.0695 & 0.0485 & 0.0530 \\
3.0 & 0.05 & 60 & 0.0470 & 0.0495 & 0.0635 & 0.0485 & 0.0460 \\
\cline{2-8}
3.0 & 0.1 & 20 & 0.0875 & 0.0885 & 0.1260 & 0.0885 & 0.1000 \\
3.0 & 0.1 & 40 & 0.1050 & 0.1040 & 0.1350 & 0.1060 & 0.0975 \\
3.0 & 0.1 & 60 & 0.1040 & 0.1080 & 0.1370 & 0.1040 & 0.1040 \\
\hline
5.0 & 0.01 & 20 & 0.0140 & 0.0110 & 0.0200 & 0.0140 & 0.0145 \\
5.0 & 0.01 & 40 & 0.0090 & 0.0100 & 0.0155 & 0.0090 & 0.0100 \\
5.0 & 0.01 & 60 & 0.0105 & 0.0090 & 0.0120 & 0.0110 & 0.0075 \\
\cline{2-8}
5.0 & 0.05 & 20 & 0.0540 & 0.0535 & 0.0845 & 0.0540 & 0.0620 \\
5.0 & 0.05 & 40 & 0.0530 & 0.0525 & 0.0730 & 0.0525 & 0.0555 \\
5.0 & 0.05 & 60 & 0.0520 & 0.0510 & 0.0635 & 0.0520 & 0.0500 \\
\cline{2-8}
5.0 & 0.1 & 20 & 0.1140 & 0.1160 & 0.1520 & 0.1140 & 0.1130 \\
5.0 & 0.1 & 40 & 0.0995 & 0.1000 & 0.1300 & 0.0995 & 0.1040 \\
5.0 & 0.1 & 60 & 0.1040 & 0.0985 & 0.1320 & 0.1050 & 0.1060
\end{tabular}
\label{type1ErrorGammaDiffNonSym}
\end{table}

\subsection{Ratio of means with gamma data}

In this subsection, we use the statistic $T = \max(\bar{x}/ \bar{y}, \bar{y}/ \bar{x})$ with data generated as gamma random variables.

\subsubsection{Small sample sizes}

We generated data $x_{i}, i=1,\ldots,n_x$ and $y_j, j=1,\ldots,n_y$ as realizations of the respective random variables $X_i\overset{\text{iid}}{\sim} \text{Gamma}(\alpha, \lambda_x)$ and $Y_{j} \overset{\text{iid}}{\sim} \text{Gamma}(\alpha, \lambda_y)$, where $\lambda$ is the rate parameter, and $\alpha = 0.5, 3, 5$. For equal sample sizes, we set $n = n_x = n_y = 20, 40, 60$, and for unequal sample sizes, we set $n_x = 20, 40, 60$ and $n_y = 100$. For all simulations, we set $\lambda_x = 1$. For equal samples sizes, we set $\lambda_y = 7, 12.5$ for each $n$. For unequal sample sizes, we set $\lambda_y = 2.25, 2.75$ for all $n_x$ for $\alpha = 0.5$, $\lambda_y = 2, 2.5$ for all $n_x$ for $\alpha = 3$, and $\lambda_y = 1.75, 2.25$ for all $n_x$ for $\alpha = 5$. We simulated 100 datasets for each combination of parameters.

Results for equal and unequal sample size are shown in Figures \ref{simGamma_sym_smallN} and \ref{simGamma_nonSym_smallN}, respectively. \textit{Alg 1} is our resampling algorithm with $B_{\text{pred}}=10^3$ iterations in each partition, \textit{Asym} is our asymptotic approximation, \textit{Delta} is the delta method, \textit{SAMC} is the SAMC algorithm, and $p_{\beta}$ is the two-sided p-value from the beta prime distribution. Figures \ref{simGamma_sym_iter_smallN} and \ref{simGamma_nonSym_iter_smallN} show the number of iterations used by our resampling algorithm.

\begin{figure}[htbp]
\centering
  \begin{subfigure}{0.62\textwidth}
  \centering
  \includegraphics[width=1\linewidth]{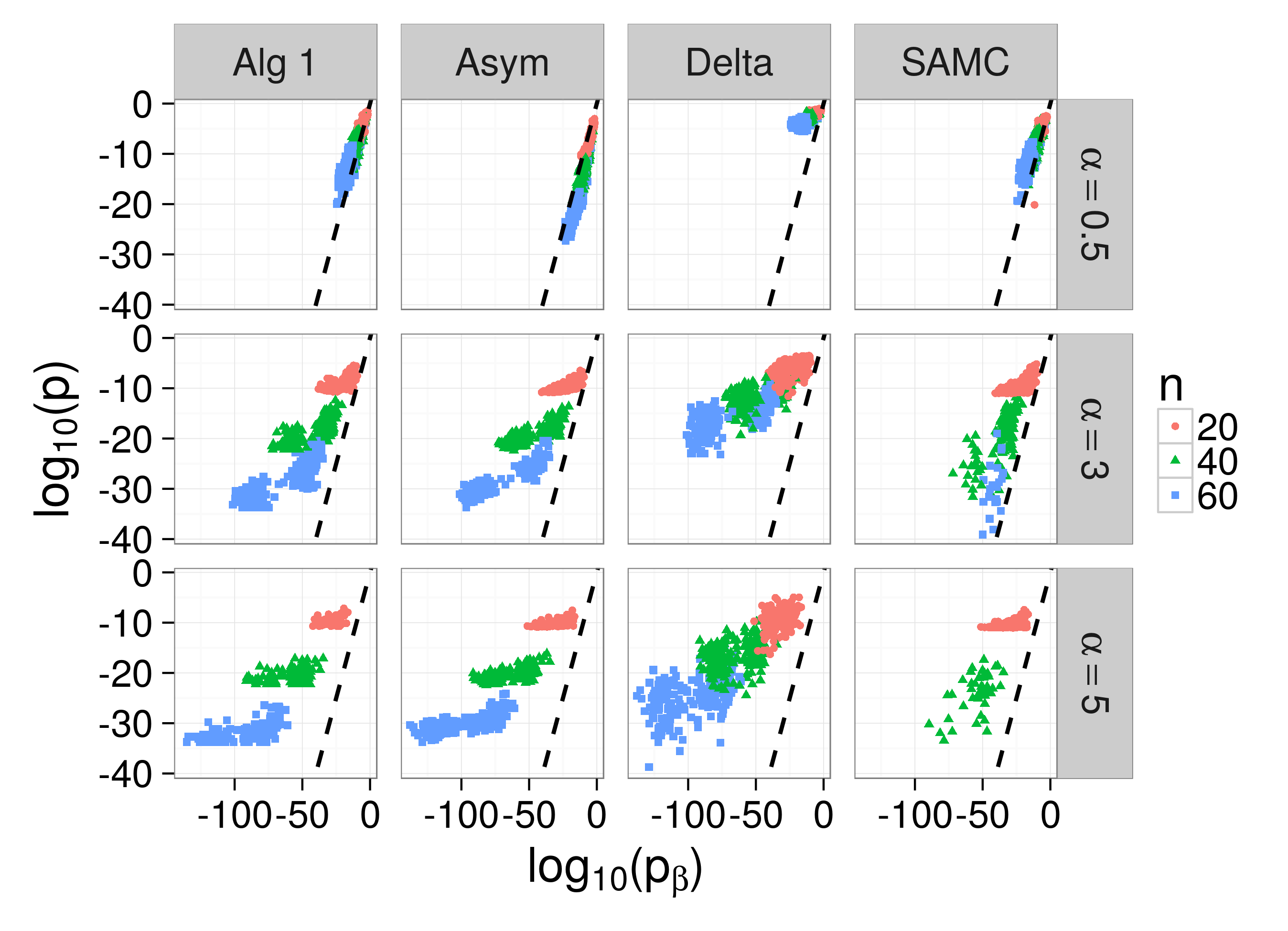}
  \caption{p-values}
  \end{subfigure}
  \begin{subfigure}{0.34\textwidth}
  \centering
  \includegraphics[width=1\linewidth]{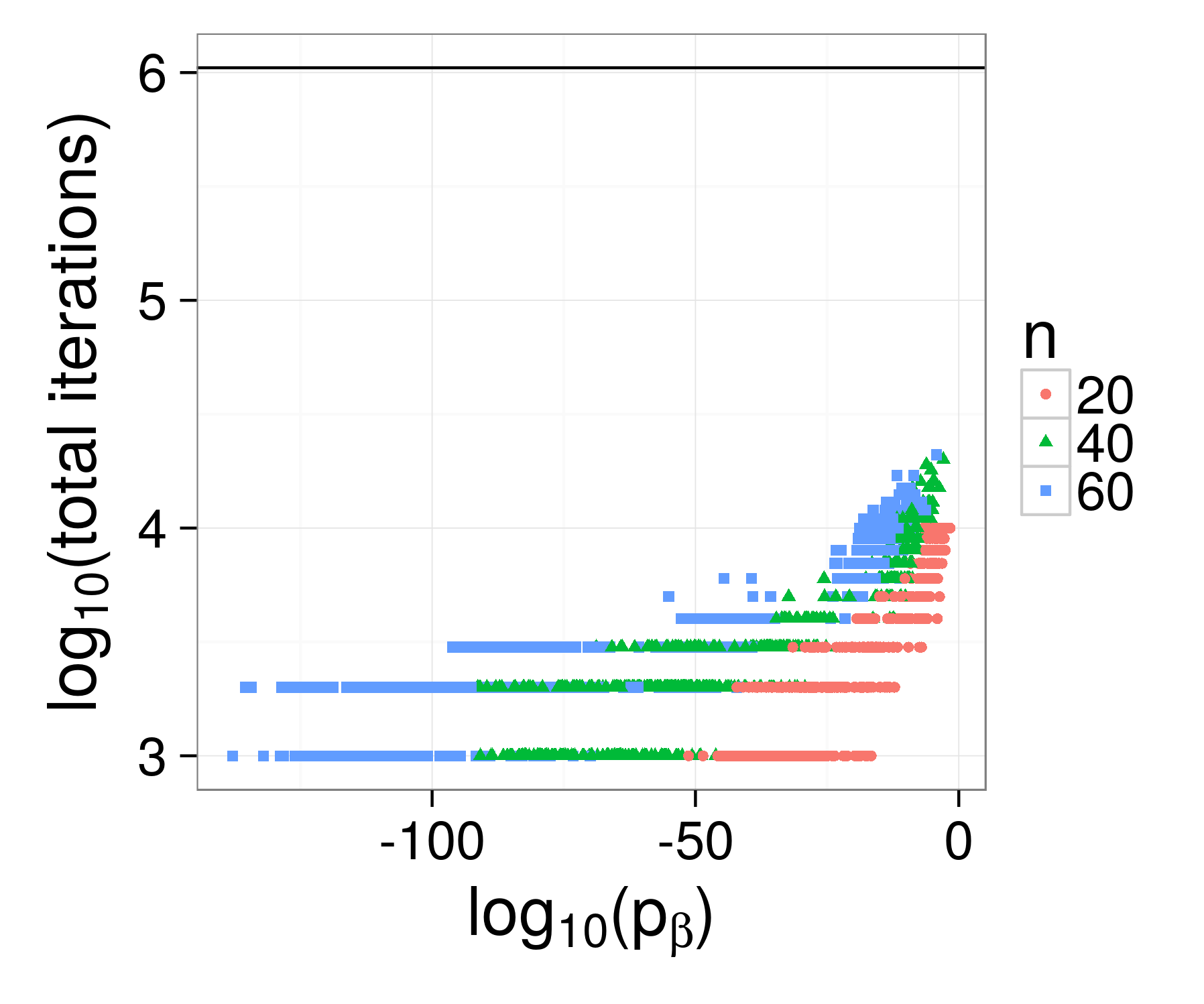}
  \caption{Iterations in resampling algorithm}
  \label{simGamma_sym_iter_smallN}
  \end{subfigure}
\caption{Simulation results using the statistic $T= \max(\bar{x} / \bar{y}, \bar{y} / \bar{x})$ with gamma data and equal sample sizes of $n=n_x=n_y=20$, 40, 60. \textit{Alg 1} is our resampling algorithm with $B_{\text{pred}}=10^3$ iterations in each partition, \textit{Asym} is our asymptotic approximation, \textit{Delta} is the delta method, \textit{SAMC} is the SAMC algorithm, and $p_{\beta}$ is the two-sided p-value from the beta prime distribution. The diagonal dashed line has slope of 1 and intercept of 0, and indicates agreement between methods. The horizontal line in \ref{simGamma_sym_iter_smallN} shows the number of permutations used in the SAMC algorithm (set in advance, and independent of p-value). The SAMC algorithm did not produce values for 652 tests (points missing).}
\label{simGamma_sym_smallN}
\end{figure}

\begin{figure}[htbp]
\centering
  \begin{subfigure}{0.62\textwidth}
  \centering
  \includegraphics[width=1\linewidth]{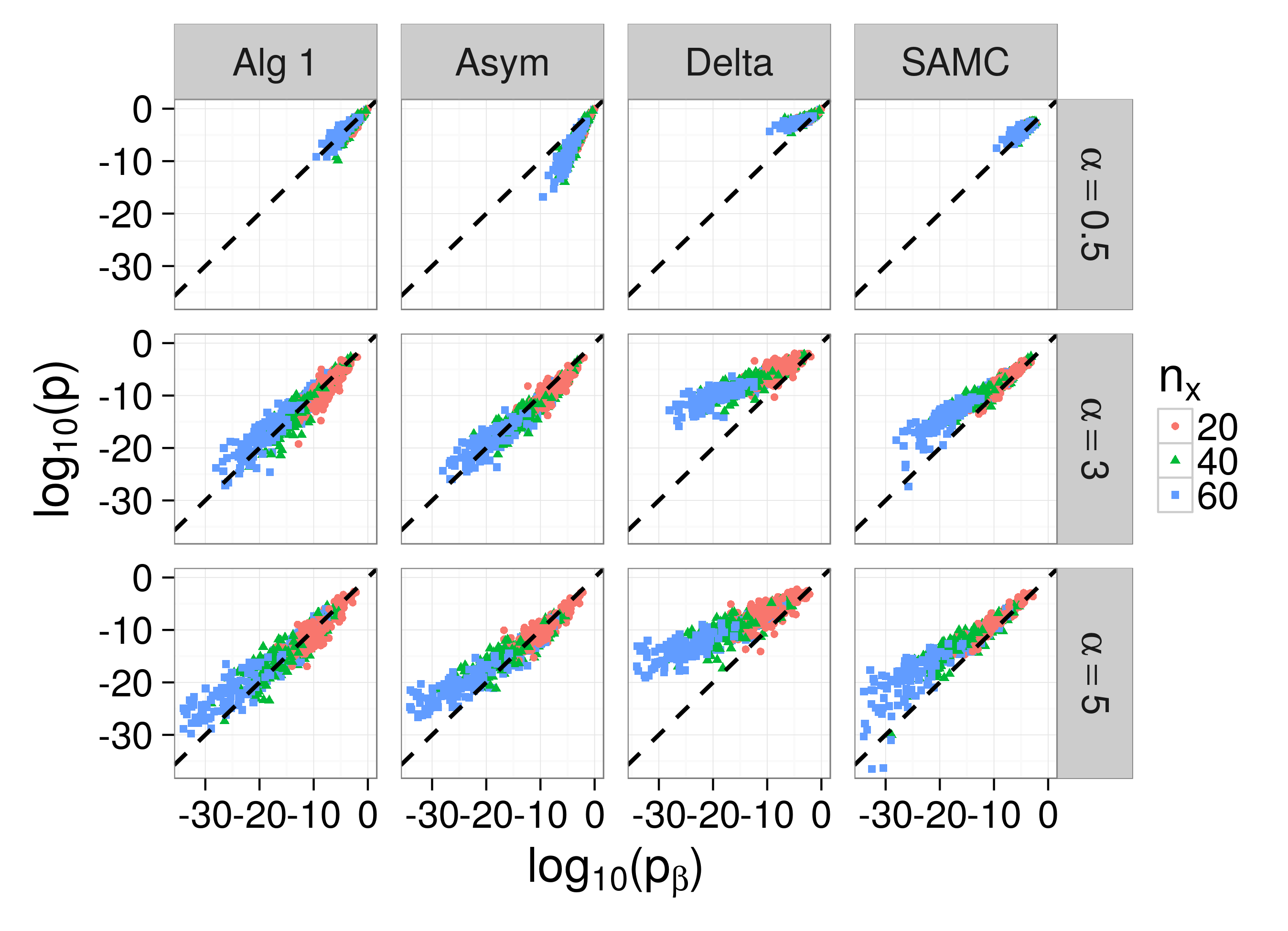}
  \caption{p-values}
  \end{subfigure}
  \begin{subfigure}{0.34\textwidth}
  \centering
  \includegraphics[width=1\linewidth]{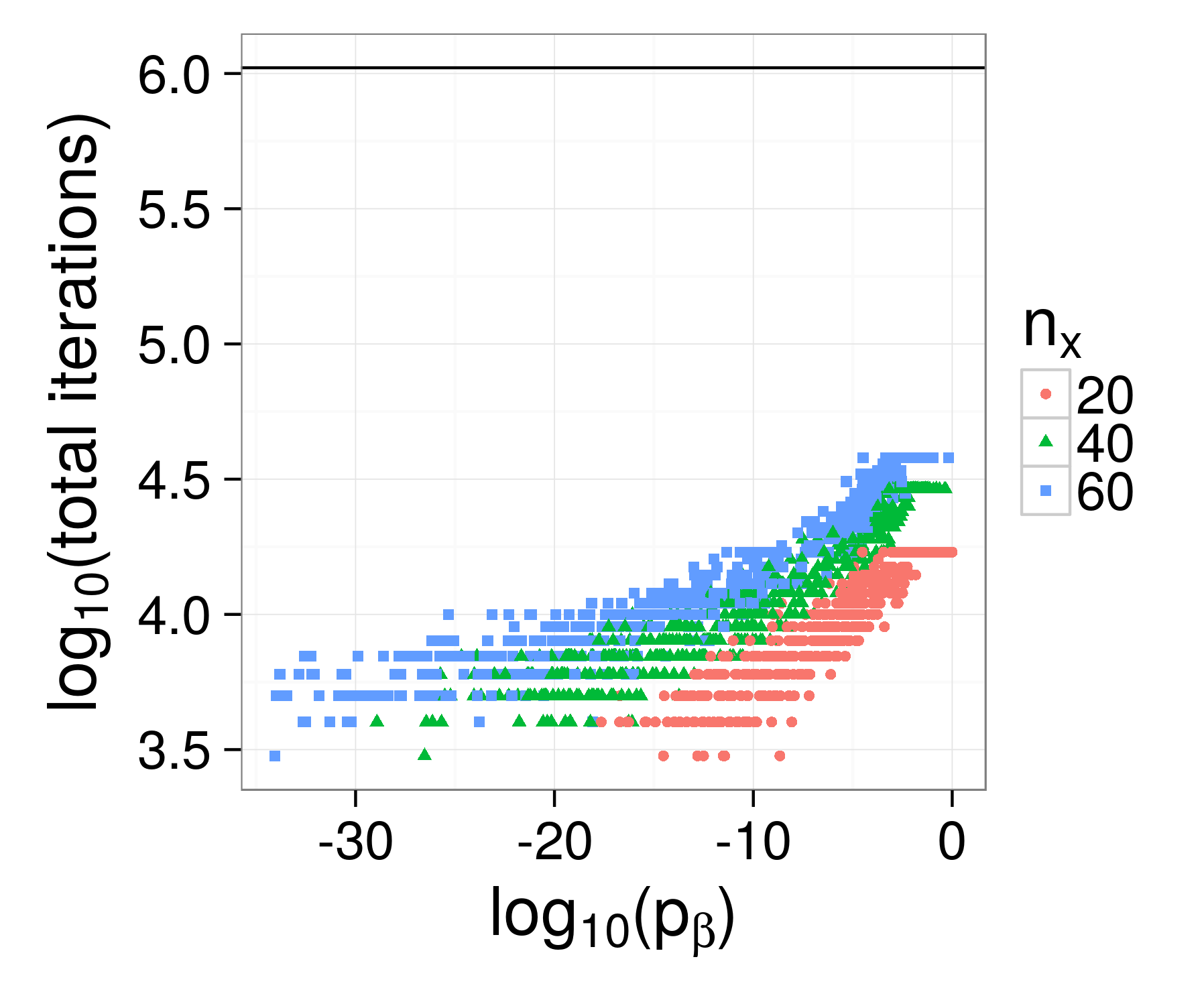}
  \caption{Iterations in resampling algorithm}
  \label{simGamma_nonSym_iter_smallN}
  \end{subfigure}
\caption{Simulation results using the statistic $T= \max(\bar{x} / \bar{y}, \bar{y} / \bar{x})$, with gamma data and unequal sample sizes of $n_x=20$, 40, 60, $n_y = 100$, and rates $\lambda_y = 5, 10$, and $\lambda_x = 1$. \textit{Alg 1} is our resampling algorithm with $B_{\text{pred}}=10^3$ iterations in each partition, \textit{Asym} is our asymptotic approximation, \textit{Delta} is the delta method, \textit{SAMC} is the SAMC algorithm, and $p_{\beta}$ is the two-sided p-value from the beta prime distribution. The diagonal dashed line has slope of 1 and intercept of 0, and indicates agreement between methods. The horizontal line in \ref{simGamma_nonSym_iter_smallN} shows the number of permutations used in the SAMC algorithm (set in advance, and independent of p-value). The SAMC algorithm did not produce values for 304 tests (points missing)}
\label{simGamma_nonSym_smallN}
\end{figure}

\subsubsection{Under the null $P_x = P_y$}

We generated data $x_{i}, i=1,\ldots,n_x$ and $y_j, j=1,\ldots,n_y$ as realizations of the respective random variables $X_i\overset{\text{iid}}{\sim} \text{Gamma}(\alpha, 1)$ and $Y_{j} \overset{\text{iid}}{\sim} \text{Gamma}(\alpha, 1)$ for $\alpha = 0.5, 3, 5$. For equal sample sizes, we set $n = n_x = n_y = 20, 40, 60$. For unequal sample sizes, we set $n_x = 20, 40, 60$ and $n_y = 100$. For both equal and unequal sample sizes, we simulated 1,000 datasets for each combination of parameters (we used 1,000 datasets, as opposed to 100, to better investigate the type I error rate). We used the p-value from simple Monte Carlo resampling with $10^5$ iterations, denoted as $\tilde{p}$, as an approximation for the true permutation p-value.

Results for equal and unequal sample size are shown in Figures \ref{simGamma_sym_null} and \ref{simGamma_nonSym_null}, respectively. \textit{Alg 1} is our resampling algorithm with $B_{\text{pred}}=10^3$ iterations in each partition, \textit{Asym} is our asymptotic approximation, \textit{Delta} is the delta method, \textit{Beta prime} gives the p-value from the beta prime distribution, and $\tilde{p}$ is from simple Monte Carlo resampling with $10^5$ iterations. Given the large p-values, using $10^5$ Monte Carlo resamples should be sufficient to obtain reliable estimates of the true permutation p-value. Therefore, this comparison demonstrates that the permutation p-value is not exactly the same as the p-value from the beta prime distribution. However, it appears reasonably close, and so we use it as an approximation to the truth in other simulations, in which the p-values are much smaller and simple Monte Carlo methods are not feasible.

We do not show results from the SAMC algorithm, because as noted above, the \verb|EXPERT| package \citep{yu2011} does not provide results for p-values $>10^{-3}$.

\begin{figure}[htbp]
\centering
\includegraphics[scale = 0.5]{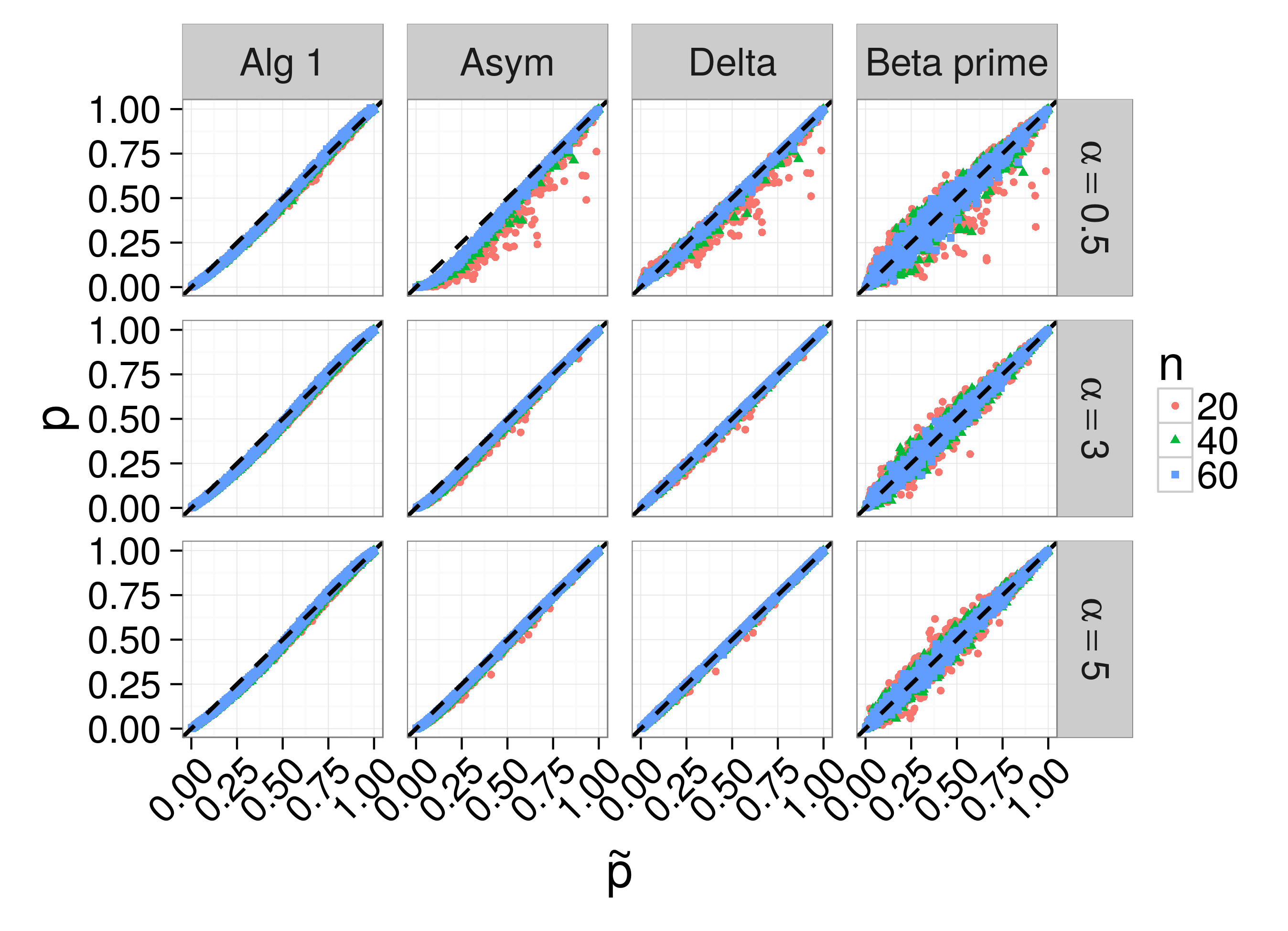}
\caption{Simulation results using the statistic $T= \max(\bar{x} / \bar{y}, \bar{y} / \bar{x})$ with gamma data under the null of $P_x = P_y$ (rates $\lambda_x = \lambda_y = 1$), with equal sample sizes of $n=n_x=n_y=20$, 40, 60. \textit{Alg 1} is our resampling algorithm with $B_{\text{pred}}=10^3$ iterations in each partition, \textit{Asym} is our asymptotic approximation, \textit{Delta} is the delta method, \textit{Beta prime} gives the p-value from the beta prime distribution, and $\tilde{p}$ is from simple Monte Carlo resampling with $10^5$ iterations. The diagonal dashed line has slope of 1 and intercept of 0, and indicates agreement between methods.}
\label{simGamma_sym_null}
\end{figure}

\begin{figure}[htbp]
\centering
\includegraphics[scale = 0.5]{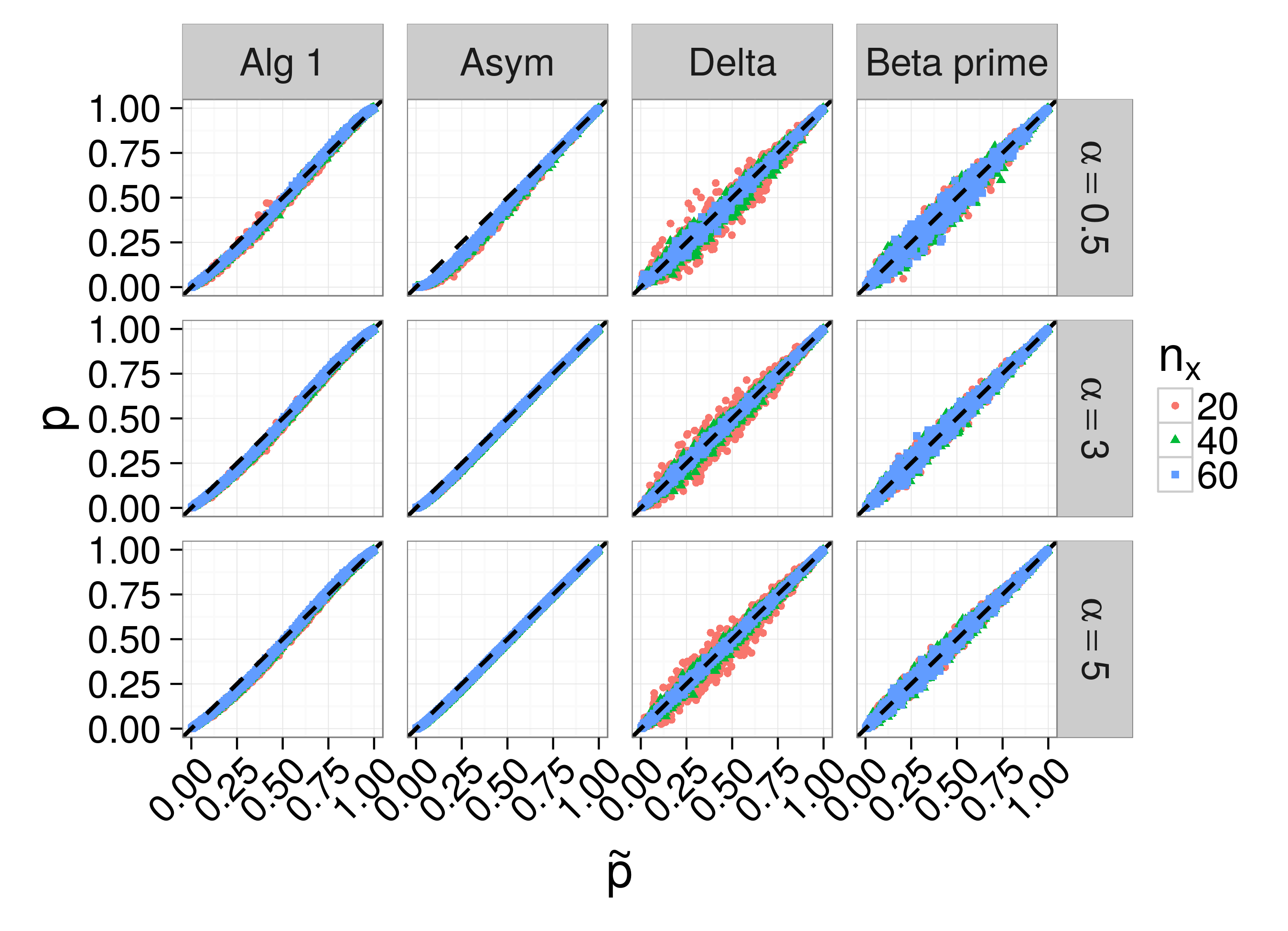}
\caption{Simulation results using the statistic $T= \max(\bar{x} / \bar{y}, \bar{y} / \bar{x})$ with gamma data under the null of $P_x = P_y$ (rates $\lambda_x = \lambda_y = 1$), with unequal sample sizes of $n_x = 20$, 40, 60 and $n_y = 100$. \textit{Alg 1} is our resampling algorithm with $B_{\text{pred}}=10^3$ iterations in each partition, \textit{Asym} is our asymptotic approximation, \textit{Delta} is the delta method, \textit{Beta prime} gives the p-value from the beta prime distribution, and $\tilde{p}$ is from simple Monte Carlo resampling with $10^5$ iterations. The diagonal dashed line has slope of 1 and intercept of 0, and indicates agreement between methods.}
\label{simGamma_nonSym_null}
\end{figure}

Tables \ref{type1ErrorGammaRatioSym} and \ref{type1ErrorGammaRatioNonSym} show the Type I error rates under the null $H_0: P_x = P_y$ for the equal and unequal sample size simulations, respectively. \textit{MC} is the unadjusted p-value from simple Monte Carlo resampling and $10^5$ iterations, \textit{Beta prime} is the p-value from the beta prime distribution, \textit{Alg 1} is our resampling algorithm, and \textit{Asymptotic} is our asymptotic approximation.

\begin{table}[htbp]
\centering
\caption{Type I error rates $\Pr(\text{p-value} \le \text{signif level} | H_0)$ for $T = \max(\bar{x} / \bar{y}, \bar{y} / \bar{x})$ with gamma data and equal sample sizes $n=n_x = n_y$. $\alpha$ is the shape parameter in the gamma distribution, \textit{MC} is simple Monte Carlo resampling with $10^5$ iterations, \textit{Alg 1} is our resampling algorithm, and \textit{Asymptotic} is our asymptotic approximation, \textit{Delta} is the delta method, and \textit{Beta prime} is the the beta prime distribution.}
\begin{tabular}{cccccccc}
\hline \hline
$\alpha$ & signif level & $n$ & MC & Alg 1 & Asym & Delta & Beta prime \\
\hline
0.5 & 0.01 & 20 & 0.013 & 0.018 & 0.093 & 0.002 & 0.015 \\
0.5 & 0.01 & 40 & 0.007 & 0.014 & 0.055 & 0.001 & 0.007 \\
0.5 & 0.01 & 60 & 0.007 & 0.010 & 0.047 & 0.002 & 0.011 \\
\cline{2-8}
0.5 & 0.05 & 20 & 0.050 & 0.076 & 0.182 & 0.026 & 0.053 \\
0.5 & 0.05 & 40 & 0.050 & 0.072 & 0.135 & 0.037 & 0.055 \\
0.5 & 0.05 & 60 & 0.048 & 0.068 & 0.114 & 0.043 & 0.050 \\
\cline{2-8}
0.5 & 0.10 & 20 & 0.110 & 0.136 & 0.243 & 0.106 & 0.108 \\
0.5 & 0.10 & 40 & 0.106 & 0.135 & 0.196 & 0.114 & 0.104 \\
0.5 & 0.10 & 60 & 0.096 & 0.127 & 0.178 & 0.101 & 0.097 \\
\hline
3.0 & 0.01 & 20 & 0.007 & 0.012 & 0.027 & 0.003 & 0.006 \\
3.0 & 0.01 & 40 & 0.012 & 0.016 & 0.025 & 0.010 & 0.010 \\
3.0 & 0.01 & 60 & 0.012 & 0.015 & 0.025 & 0.012 & 0.008 \\
\cline{2-8}
3.0 & 0.05 & 20 & 0.043 & 0.067 & 0.088 & 0.046 & 0.044 \\
3.0 & 0.05 & 40 & 0.053 & 0.062 & 0.073 & 0.052 & 0.051 \\
3.0 & 0.05 & 60 & 0.059 & 0.075 & 0.080 & 0.061 & 0.049 \\
\cline{2-8}
3.0 & 0.10 & 20 & 0.095 & 0.126 & 0.143 & 0.103 & 0.090 \\
3.0 & 0.10 & 40 & 0.098 & 0.133 & 0.147 & 0.104 & 0.103 \\
3.0 & 0.10 & 60 & 0.095 & 0.115 & 0.116 & 0.097 & 0.093 \\
\hline
5.0 & 0.01 & 20 & 0.009 & 0.015 & 0.023 & 0.009 & 0.009 \\
5.0 & 0.01 & 40 & 0.008 & 0.013 & 0.025 & 0.008 & 0.011 \\
5.0 & 0.01 & 60 & 0.012 & 0.012 & 0.019 & 0.012 & 0.013 \\
\cline{2-8}
5.0 & 0.05 & 20 & 0.046 & 0.063 & 0.082 & 0.054 & 0.052 \\
5.0 & 0.05 & 40 & 0.048 & 0.063 & 0.066 & 0.050 & 0.043 \\
5.0 & 0.05 & 60 & 0.055 & 0.078 & 0.079 & 0.057 & 0.057 \\
\cline{2-8}
5.0 & 0.10 & 20 & 0.093 & 0.130 & 0.139 & 0.106 & 0.099 \\
5.0 & 0.10 & 40 & 0.091 & 0.134 & 0.138 & 0.094 & 0.093 \\
5.0 & 0.10 & 60 & 0.115 & 0.138 & 0.136 & 0.116 & 0.112
\end{tabular}
\label{type1ErrorGammaRatioSym}
\end{table}

\begin{table}[htbp]
\centering
\caption{Type I error rates $\Pr(\text{p-value} \le \text{signif level} | H_0)$ for $T = \max(\bar{x} / \bar{y}, \bar{y} / \bar{x})$ with gamma data and unequal sample sizes $n_x \ne n_y$ ($n_x$ shown, and $n_y = 100$). $\alpha$ is the shape parameter in the gamma distribution, \textit{MC} is simple Monte Carlo resampling with $10^5$ iterations, \textit{Alg 1} is our resampling algorithm, and \textit{Asymptotic} is our asymptotic approximation, \textit{Delta} is the delta method with, and \textit{Beta prime} is the beta prime distribution.}
\begin{tabular}{cccccccc}
\hline \hline
$\alpha$ & signif level & $n_x$ & MC & Alg 1 & Asym & Delta & Beta prime \\
\hline
0.5 & 0.01 & 20 & 0.011 & 0.015 & 0.065 & 0.006 & 0.011 \\
0.5 & 0.01 & 40 & 0.015 & 0.018 & 0.053 & 0.003 & 0.013 \\
0.5 & 0.01 & 60 & 0.008 & 0.011 & 0.042 & 0.003 & 0.012 \\
\cline{2-8}
0.5 & 0.05 & 20 & 0.043 & 0.069 & 0.128 & 0.047 & 0.053 \\
0.5 & 0.05 & 40 & 0.057 & 0.072 & 0.133 & 0.048 & 0.056 \\
0.5 & 0.05 & 60 & 0.052 & 0.071 & 0.112 & 0.045 & 0.050 \\
\cline{2-8}
0.5 & 0.10 & 20 & 0.098 & 0.121 & 0.179 & 0.109 & 0.091 \\
0.5 & 0.10 & 40 & 0.113 & 0.141 & 0.195 & 0.119 & 0.108 \\
0.5 & 0.10 & 60 & 0.106 & 0.126 & 0.172 & 0.109 & 0.098 \\
\hline
3.0 & 0.01 & 20 & 0.011 & 0.016 & 0.023 & 0.012 & 0.011 \\
3.0 & 0.01 & 40 & 0.005 & 0.011 & 0.027 & 0.005 & 0.009 \\
3.0 & 0.01 & 60 & 0.011 & 0.013 & 0.017 & 0.011 & 0.011 \\
\cline{2-8}
3.0 & 0.05 & 20 & 0.047 & 0.070 & 0.073 & 0.059 & 0.039 \\
3.0 & 0.05 & 40 & 0.058 & 0.065 & 0.069 & 0.057 & 0.054 \\
3.0 & 0.05 & 60 & 0.053 & 0.066 & 0.070 & 0.050 & 0.052 \\
\cline{2-8}
3.0 & 0.10 & 20 & 0.088 & 0.128 & 0.135 & 0.104 & 0.087 \\
3.0 & 0.10 & 40 & 0.094 & 0.124 & 0.124 & 0.101 & 0.089 \\
3.0 & 0.10 & 60 & 0.094 & 0.119 & 0.117 & 0.097 & 0.097 \\
\hline
5.0 & 0.01 & 20 & 0.010 & 0.014 & 0.022 & 0.007 & 0.009 \\
5.0 & 0.01 & 40 & 0.011 & 0.011 & 0.017 & 0.011 & 0.009 \\
5.0 & 0.01 & 60 & 0.015 & 0.020 & 0.025 & 0.015 & 0.018 \\
\cline{2-8}
5.0 & 0.05 & 20 & 0.058 & 0.074 & 0.085 & 0.066 & 0.054 \\
5.0 & 0.05 & 40 & 0.046 & 0.057 & 0.059 & 0.048 & 0.052 \\
5.0 & 0.05 & 60 & 0.059 & 0.081 & 0.085 & 0.061 & 0.062 \\
\cline{2-8}
5.0 & 0.10 & 20 & 0.110 & 0.145 & 0.143 & 0.121 & 0.114 \\
5.0 & 0.10 & 40 & 0.081 & 0.114 & 0.108 & 0.085 & 0.088 \\
5.0 & 0.10 & 60 & 0.113 & 0.145 & 0.138 & 0.118 & 0.115
\end{tabular}
\label{type1ErrorGammaRatioNonSym}
\end{table}

\section{Comparison with additional methods \label{D}}

\subsection{Moment-corrected correlation}

Moment-corrected correlation (MCC) \citep{zhou2015hypothesis} is an analytical approximation to the permutation p-value, which is applicable in multiple testing situations in which the test statistic is permutationally equivalent to a single inner product. Where applicable, this approach is fast, as it does not involve resampling. However, if the test statistic of interest is not permutationally equivalent to an inner product, the MCC approach cannot be used.

The statistic $T = \bar{x} - \bar{y}$ fits into this setting, whereas, to the best of our knowledge, $T=\bar{x} / \bar{y}$ does not. To see this, let $\bm{z} = (\bx', \by')'$ and $\bm{w}= (\underset{n_x}{\underbrace{1/n_x, \ldots, 1/n_x}}, \underset{n_y}{\underbrace{-1/n_y, \ldots, -1/n_y}})'$. Then $\bar{x} - \bar{y} = \bm{z}'\bm{w}$. In contrast, $\bar{x} / \bar{y}$ cannot be written in this form, and we conjecture that it is not permutationally equivalent to any statistic that can be written in this form.

Figures \ref{mcc_largeN} through \ref{mcc_null} show simulation results for two-sided and doubled p-values, as described by \citet{zhou2015hypothesis}, using the \verb|mcc| package \citep{mcc} under the same normal data settings as in Section \ref{diffMean_normal}. While MCC is more reliable for large sample sizes (Figure \ref{mcc_largeN}), MCC appears to suffer from the same bias as our methods for small sample sizes (Figure \ref{mcc_smallN}). Furthermore, we do not think that MCC can be used to obtain p-values for the statistic $T=\max(\bar{x} / \bar{y}, \bar{y} / \bar{x})$.

\begin{figure}[htbp]
\centering
  \begin{subfigure}{0.48\textwidth}
  \centering
  \includegraphics[width=1\linewidth]{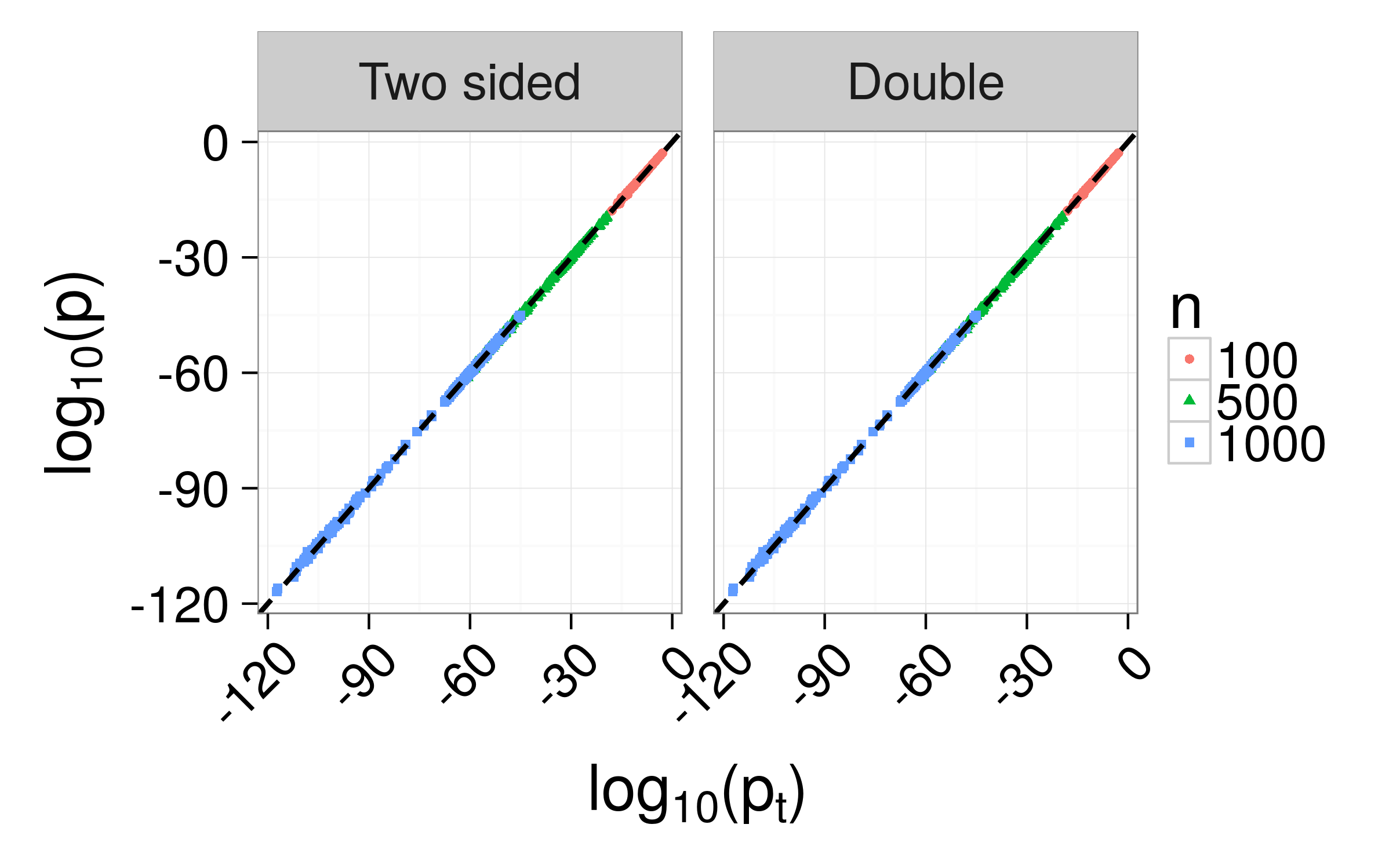}
  \caption{$n_x = n_y$}
  \end{subfigure}
  \begin{subfigure}{0.48\textwidth}
  \centering
  \includegraphics[width=1\linewidth]{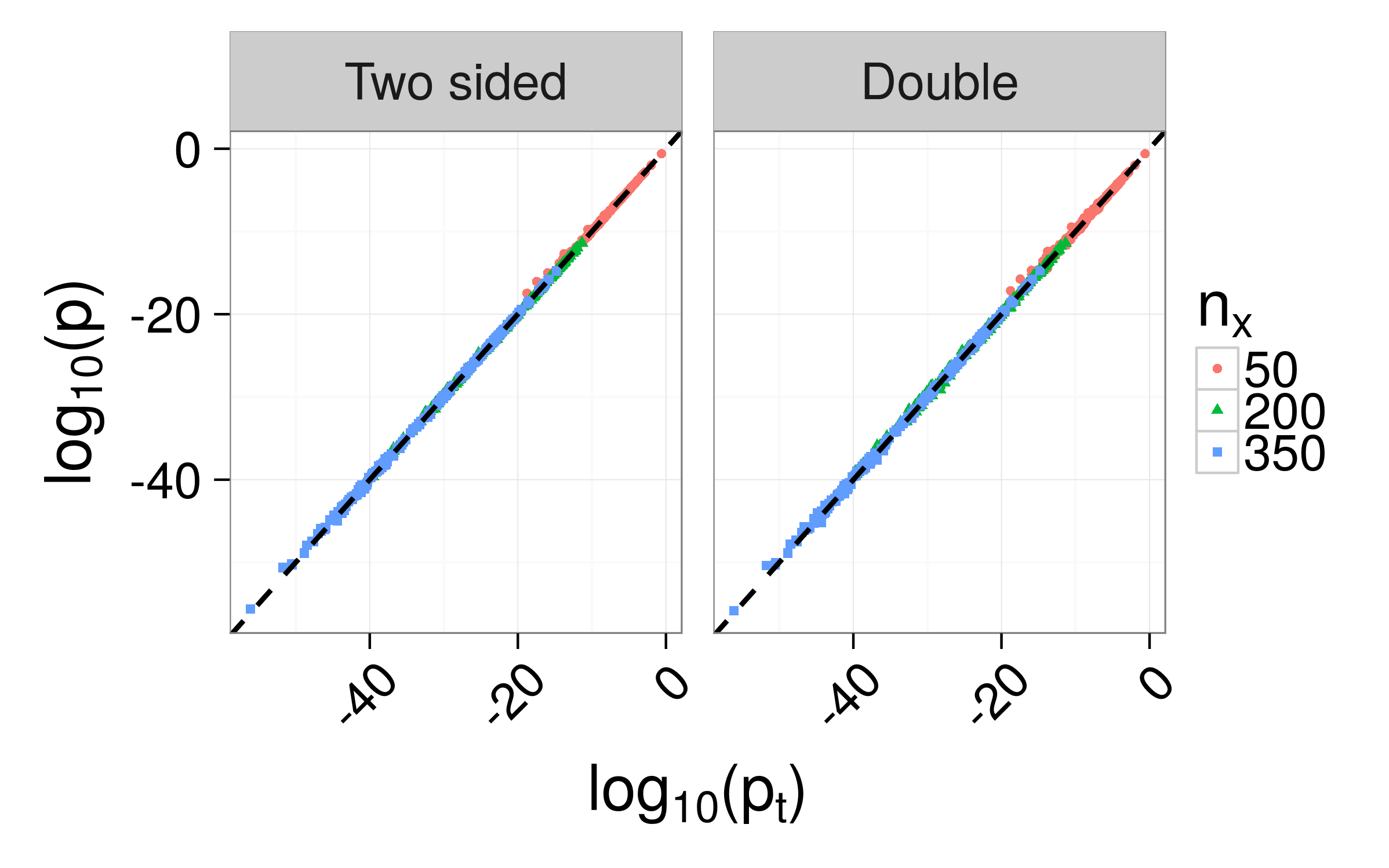}
  \caption{$n_x \ne n_y$}
  \end{subfigure}
\caption{MCC with large sample size for $T = |\bar{x} - \bar{y}|$ with normal data and equal sample sizes of $n = n_x = n_y = 100, 500, 1,000$, and unequal sample sizes of $n_x = 50, 200, 350$ with $n_y = 500$. In both cases, data were simulated as normal random variables with $\mu_y = 0$, $\mu_x = 0.75, 1$ and $\sigma_x^2 = \sigma_x^2 = 1$.}
\label{mcc_largeN}
\end{figure}

\begin{figure}[htbp]
\centering
  \begin{subfigure}{0.48\textwidth}
  \centering
  \includegraphics[width=1\linewidth]{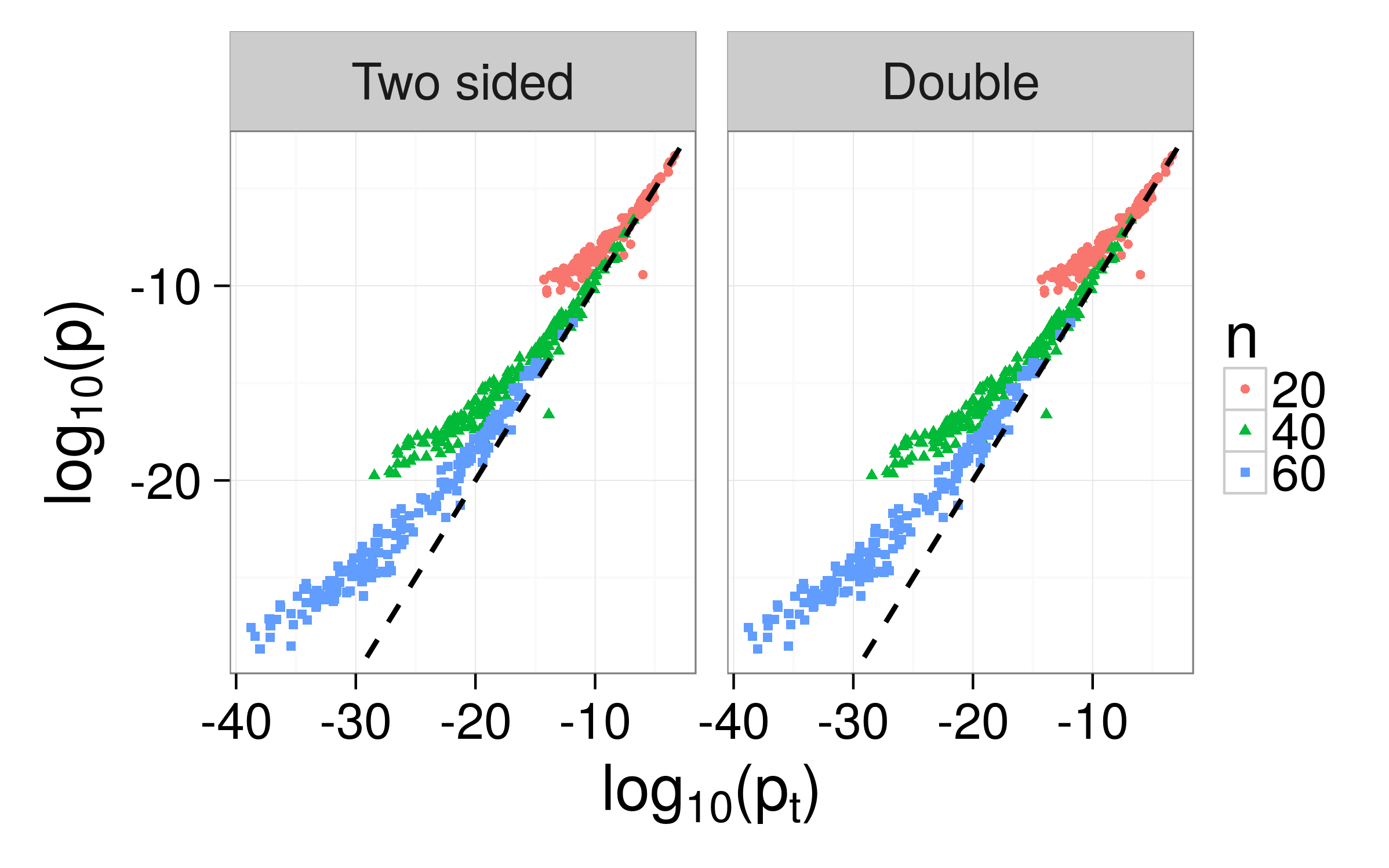}
  \caption{$n_x = n_y$}
  \end{subfigure}
  \begin{subfigure}{0.48\textwidth}
  \centering
  \includegraphics[width=1\linewidth]{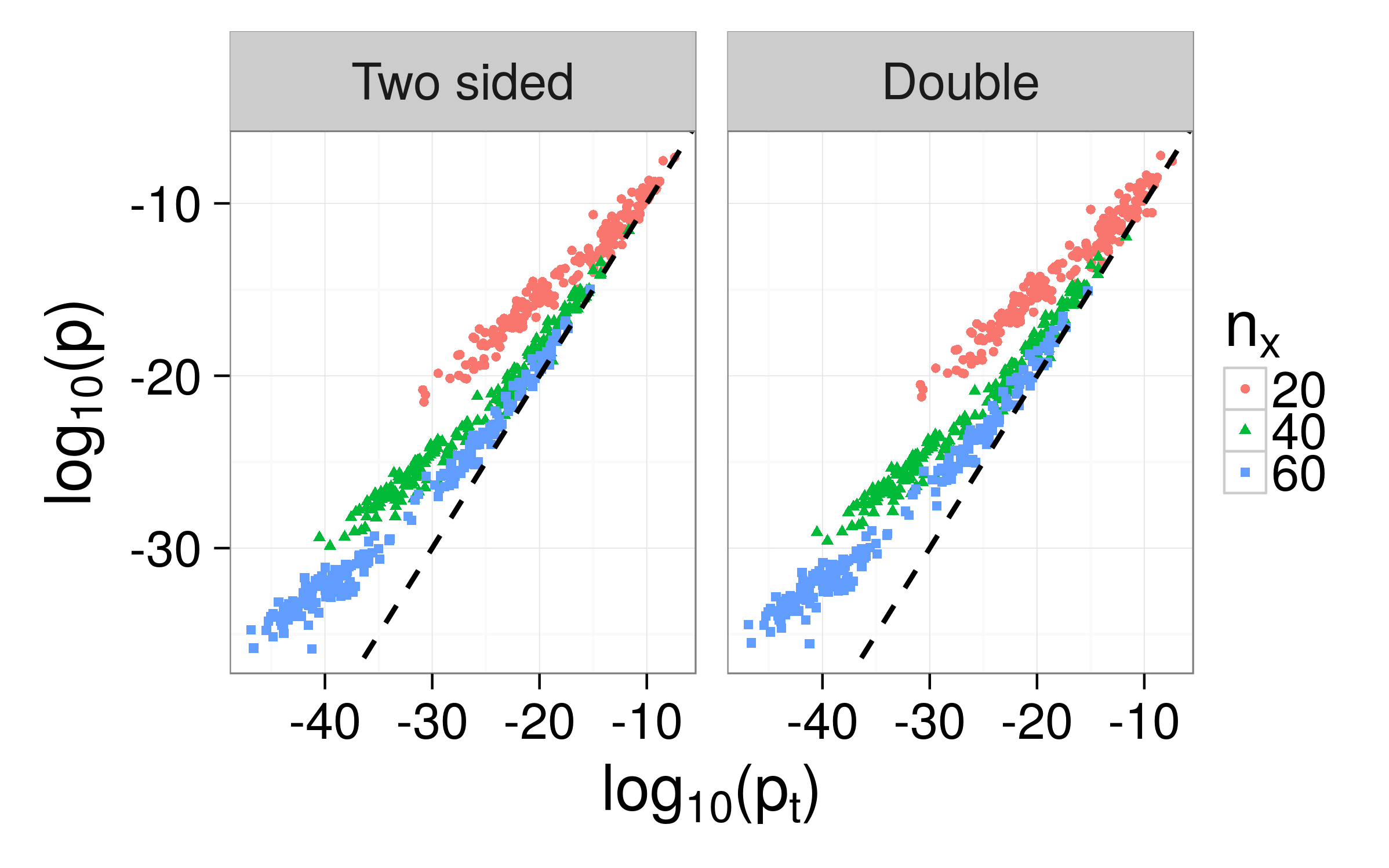}
  \caption{$n_x \ne n_y$}
  \end{subfigure}
\caption{MCC with small sample size for $T = |\bar{x} - \bar{y}|$ with normal data and equal sample sizes of $n = n_x = n_y = 20, 40, 60$, and unequal sample sizes of $n_x = 20, 40, 60$ with $n_y = 100$. In both cases, data were simulated as normal random variables with $\mu_y = 0$, $\mu_x = 2, 3$ and $\sigma_x^2 = \sigma_x^2 = 1$}
\label{mcc_smallN}
\end{figure}

\begin{figure}[htbp]
\centering
  \begin{subfigure}{0.48\textwidth}
  \centering
  \includegraphics[width=1\linewidth]{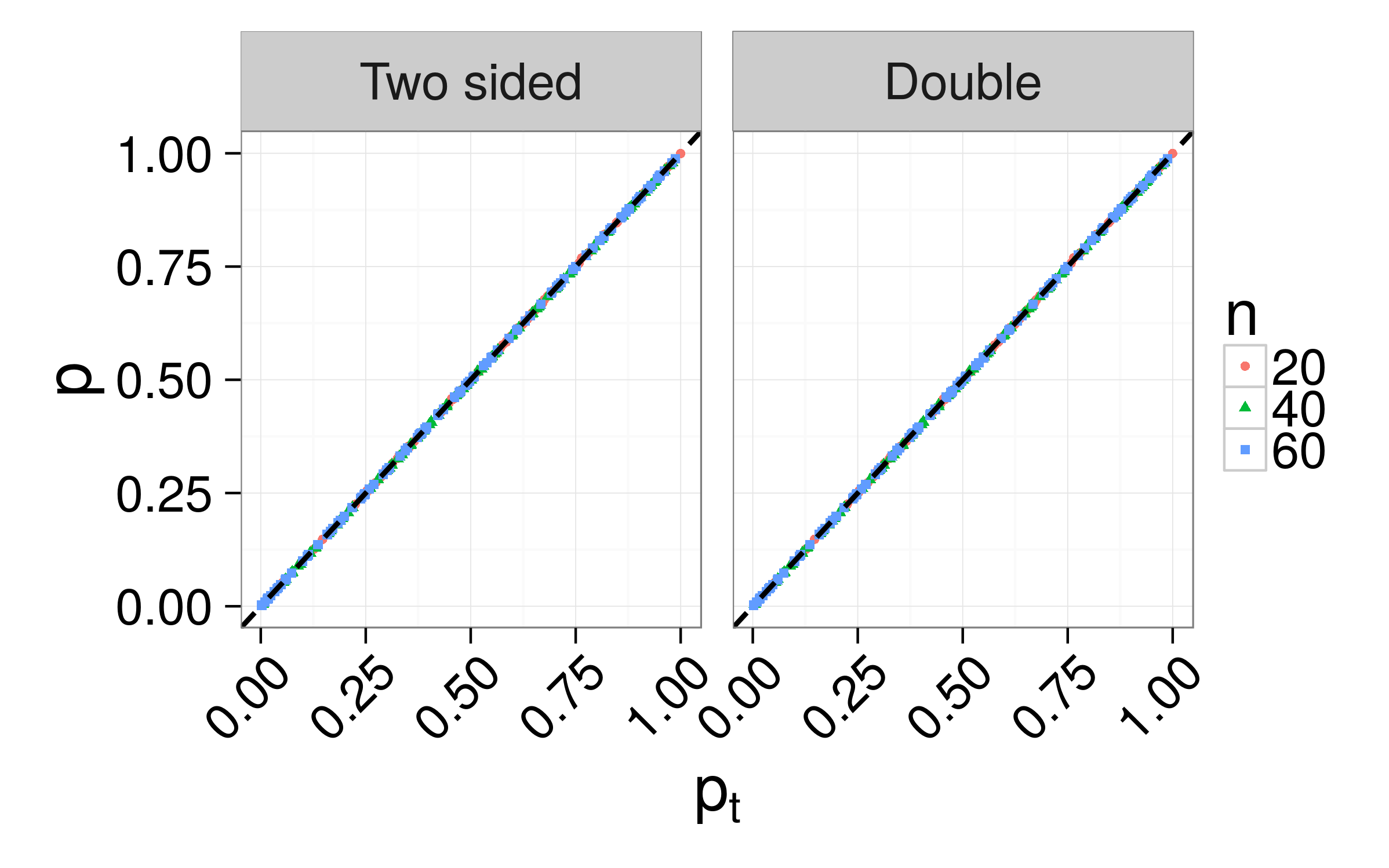}
  \caption{$n_x = n_y$}
  \end{subfigure}
  \begin{subfigure}{0.48\textwidth}
  \centering
  \includegraphics[width=1\linewidth]{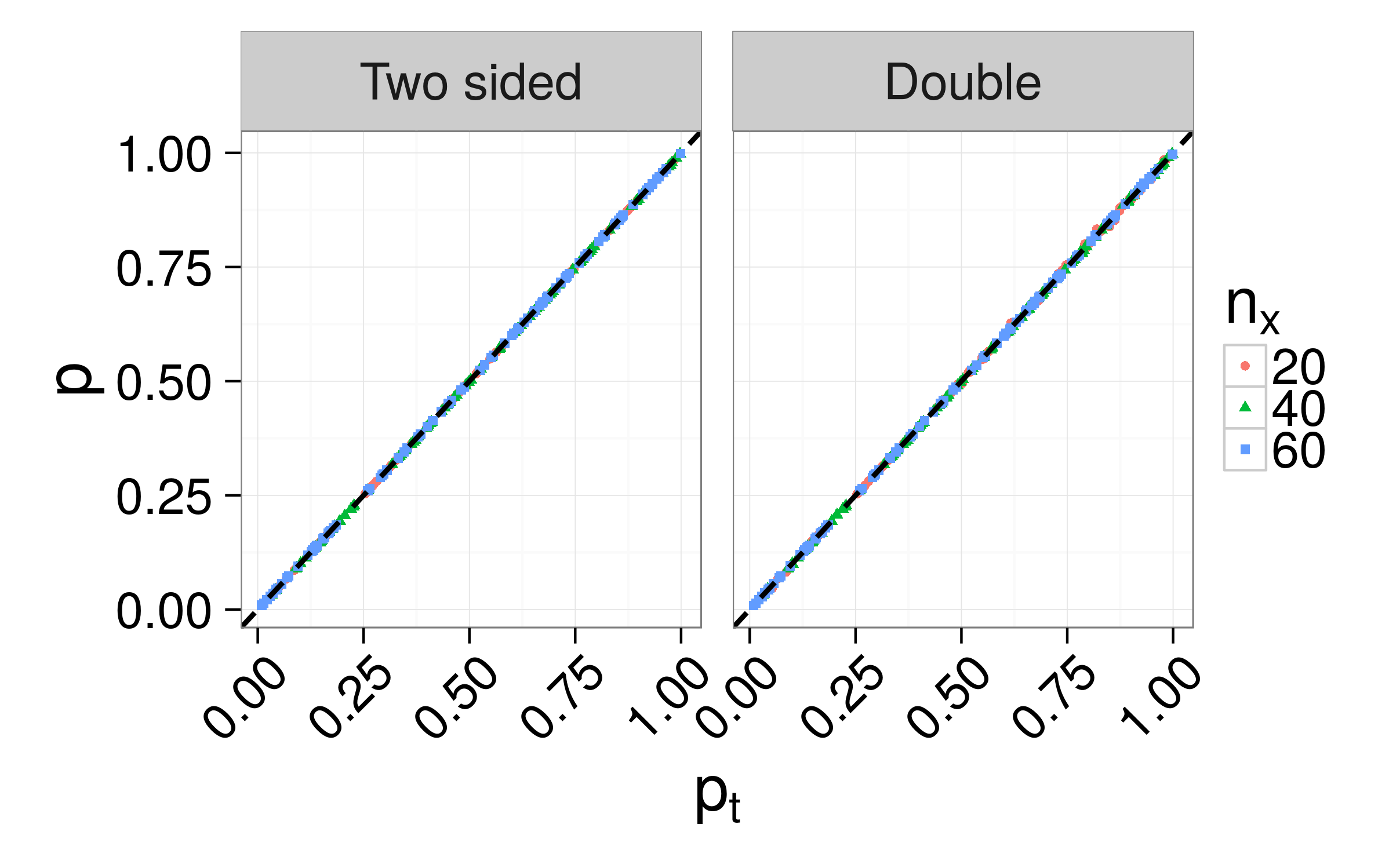}
  \caption{$n_x \ne n_y$}
  \end{subfigure}
\caption{MCC under the null hypothesis for $T = |\bar{x} - \bar{y}|$ with normal data for equal sample sizes of $n = n_x = n_y = 20, 40, 60$, and unequal sample sizes of $n_x = 20, 40, 60$ with $n_y = 100$. In both cases, data were simulated as normal random variables with $\mu_y = \mu_x = 0$ and $\sigma_x^2 = \sigma_x^2 = 1$}
\label{mcc_null}
\end{figure}

Figures \ref{mcc_gamma_smallN} and \ref{mcc_gamma_null} show simulation results for two-sided and doubled p-values for small sample sizes, and under the null, respectively, using the \verb|mcc| package \citep{mcc} under the same gamma data settings as in Section \ref{diffMean_gamma}. In Figure \ref{mcc_gamma_smallN}, we used $B=10^5$ iterations to obtain the Monte Carlo estimate $\tilde{p}$ of the true permutation p-value, and only show results for $\tilde{p} > 10^{-3}$ to ensure reliable estimates (1,019 values shown in Figure \ref{mcc_symGammaDiff_smallN}, and 705 values shown in Figure \ref{mcc_nonSymGammaDiff_smallN}).

\begin{figure}[htbp]
\centering
  \begin{subfigure}{0.48\textwidth}
  \centering
  \includegraphics[width=1\linewidth]{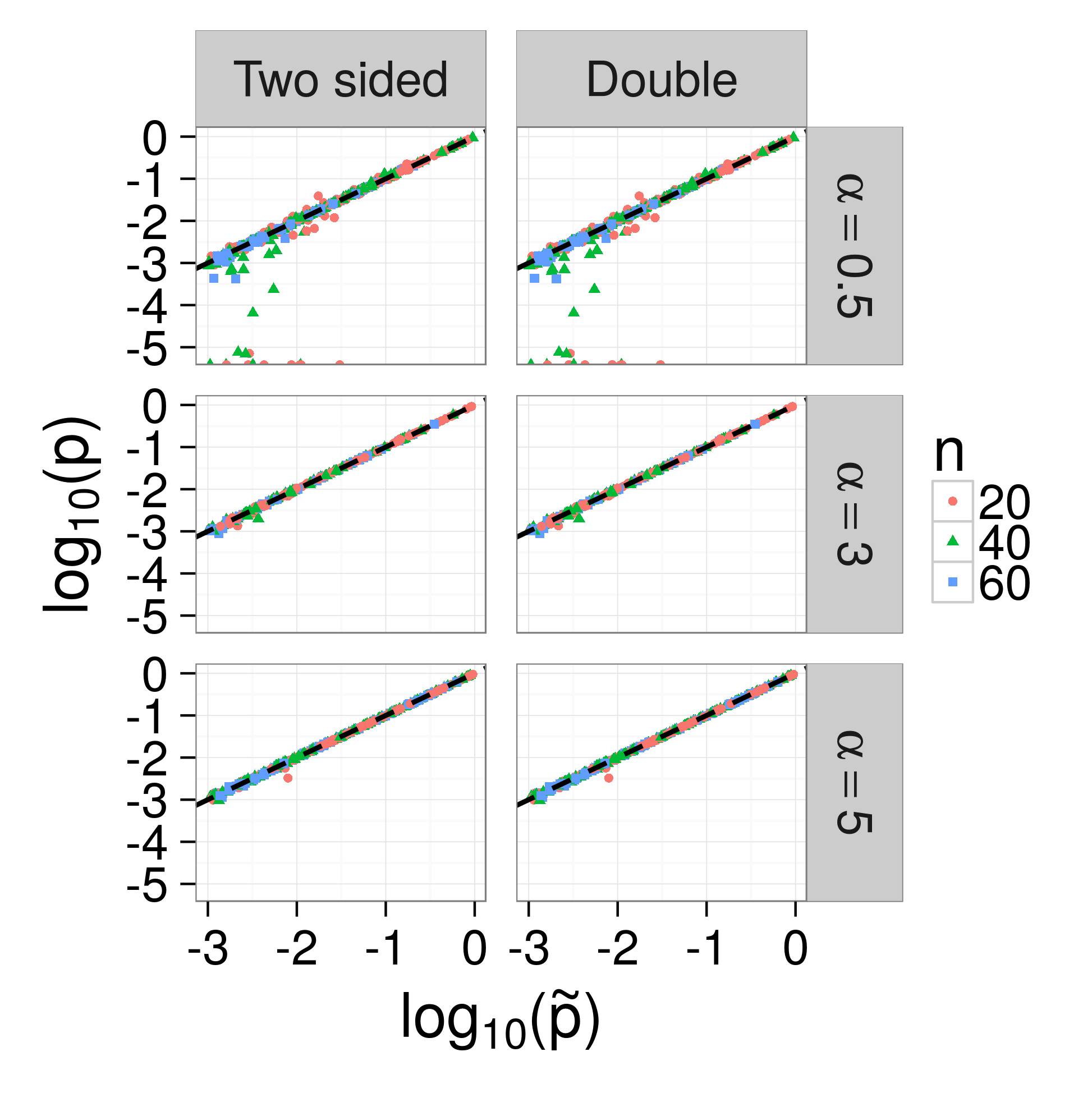}
  \caption{$n_x = n_y$}
  \label{mcc_symGammaDiff_smallN}
  \end{subfigure}
  \begin{subfigure}{0.48\textwidth}
  \centering
  \includegraphics[width=1\linewidth]{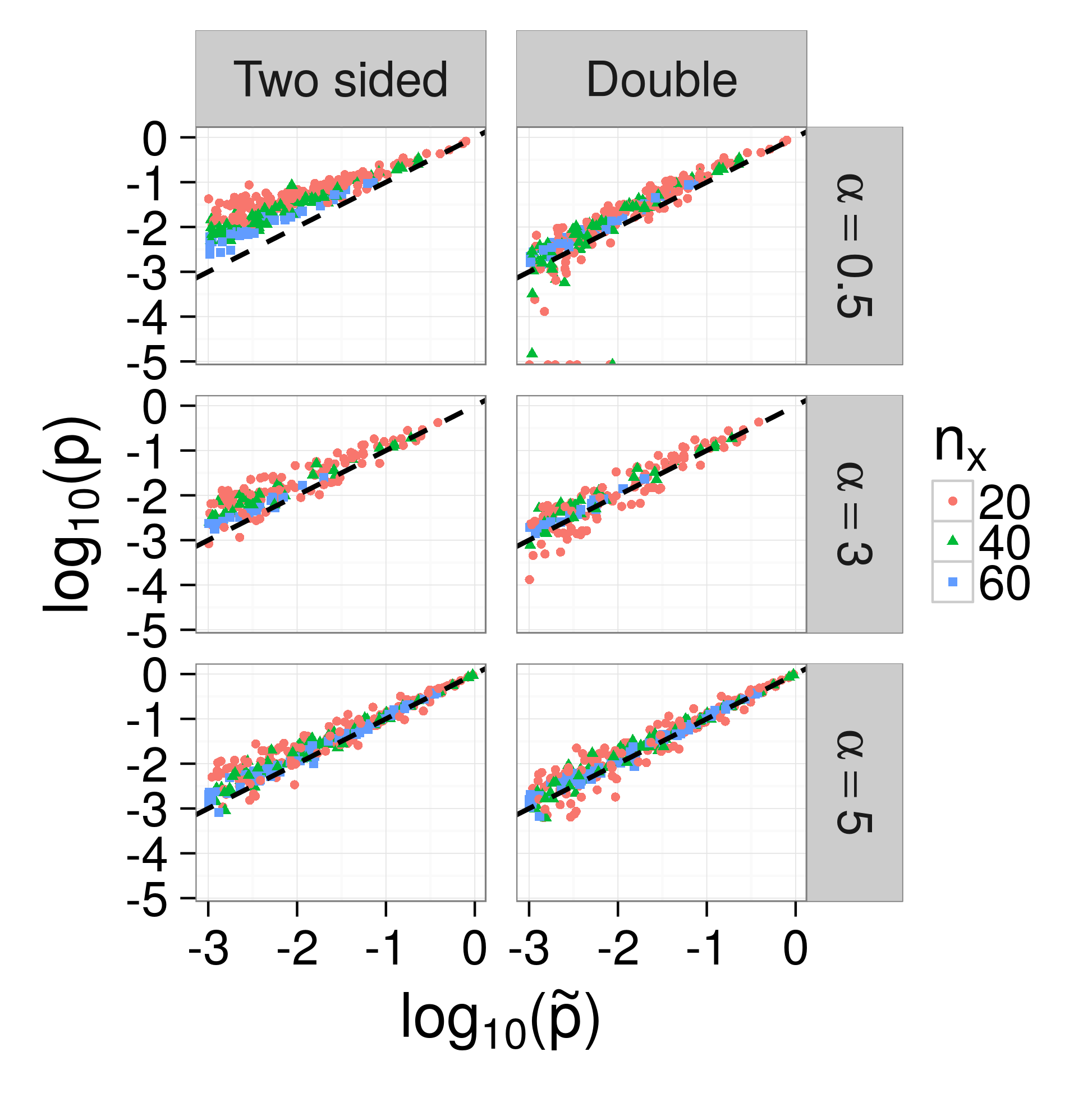}
  \caption{$n_x \ne n_y$}
  \label{mcc_nonSymGammaDiff_smallN}
  \end{subfigure}
\caption{MCC with small sample size for $T = |\bar{x} - \bar{y}|$ with gamma data and equal sample size $n = n_x = n_y = 20, 40, 60$, and unequal sample sizes of $n_x = 20, 40, 60$ with $n_y = 100$. In both cases, data were simulated as gamma random variables, as described in Section \ref{diffMean_gamma}}
\label{mcc_gamma_smallN}
\end{figure}

\begin{figure}[htbp]
\centering
  \begin{subfigure}{0.48\textwidth}
  \centering
  \includegraphics[width=1\linewidth]{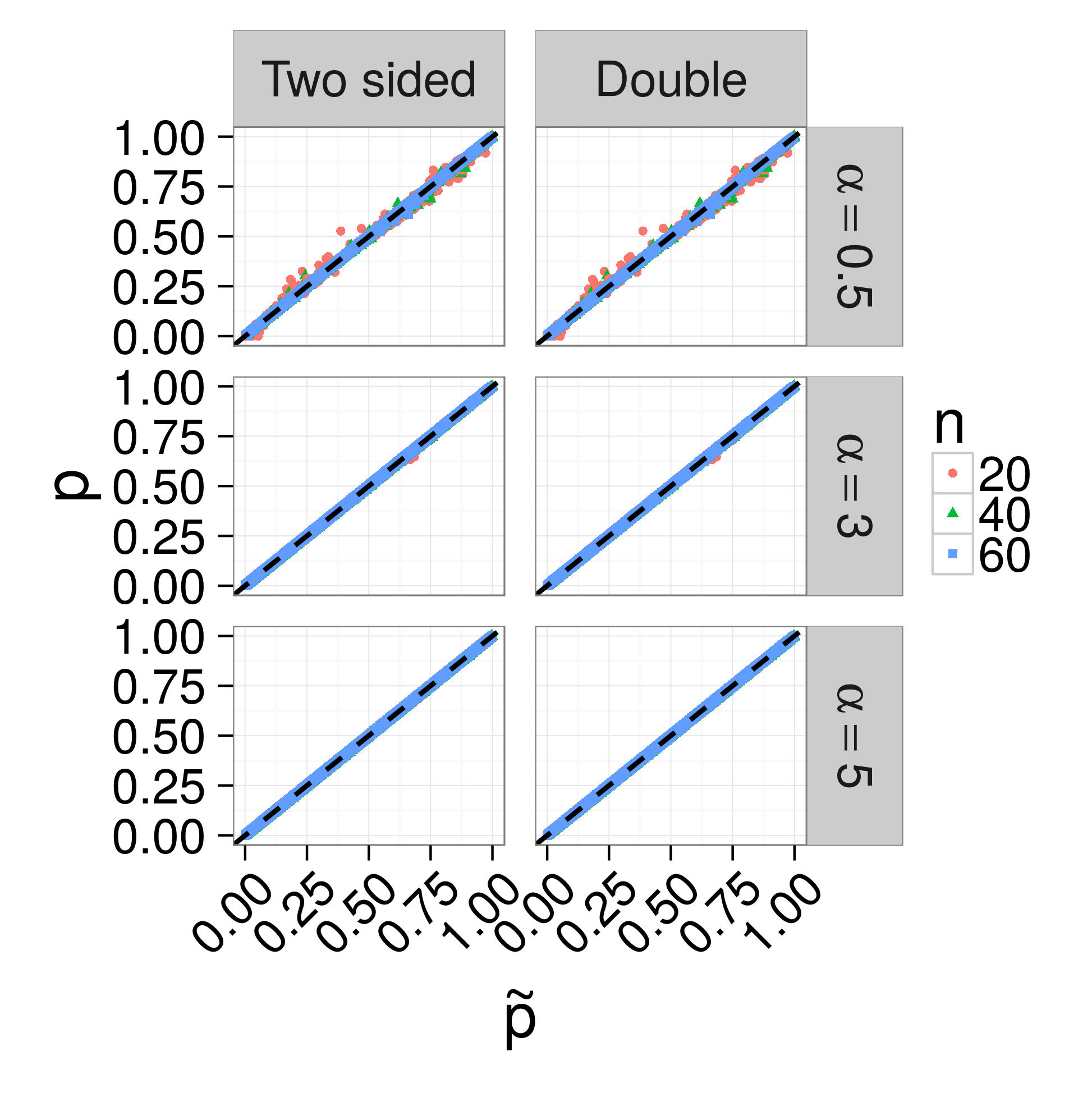}
  \caption{$n_x = n_y$}
  \end{subfigure}
  \begin{subfigure}{0.48\textwidth}
  \centering
  \includegraphics[width=1\linewidth]{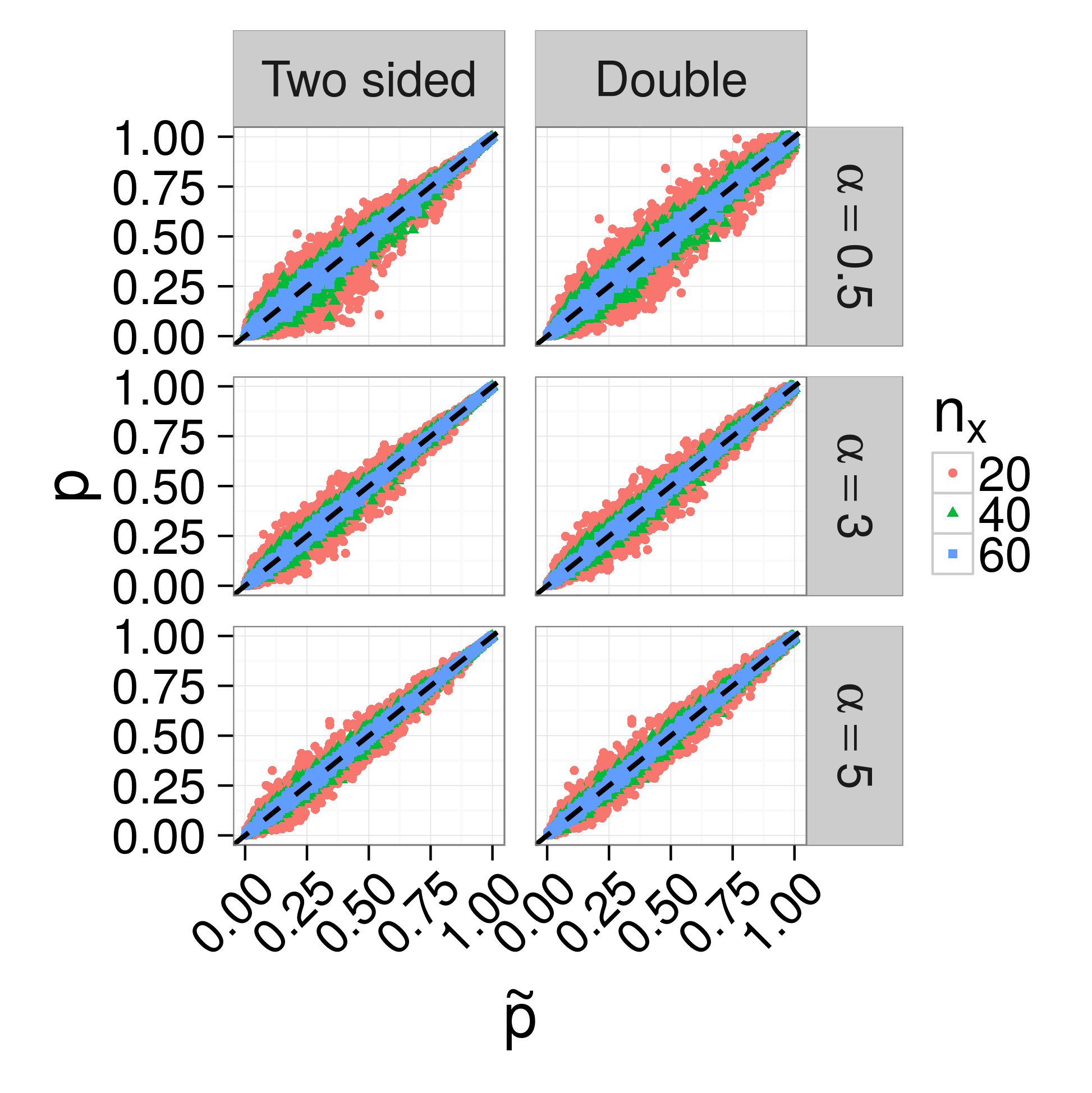}
  \caption{$n_x \ne n_y$}
  \end{subfigure}
\caption{MCC under the null hypothesis for $T = |\bar{x} - \bar{y}|$ with gamma data for equal sample sizes of $n = n_x = n_y = 20, 40, 60$, and unequal sample sizes of $n_x = 20, 40, 60$ with $n_y = 100$. In both cases, data were simulated as gamma random variables, as described in Section \ref{diffMean_gamma}}
\label{mcc_gamma_null}
\end{figure}

As seen in Figure \ref{mcc_gamma_smallN}, in many cases the MCC method substantially underestimated the permutation p-value for equal sample sizes $n_x = n_y$ and $\alpha = 0.5$. We did not observe this tendency with our resampling algorithm (see Figures \ref{simGammaDiff_sym_smallN} and \ref{simGammaDiff_nonSym_smallN}).

\subsection{Saddlepoint approximations}

Saddlepoint approximations can be used to estimate permutation p-values \citep{robinson1982}. As shown in Table \ref{saddlepoint}, estimates from our methods are comparable to those from saddlepoint approximations when using the statistic $T = |\bar{x} - \bar{y}|$. However, unlike saddlepoint approximations, our resampling algorithm requires no derivations.

\begin{table}[htbp]
\centering
\caption{Comparison with Saddlepoint approximations for $T = |\bar{x} - \bar{y}|$. Datasets are from \citet[][Table 2]{robinson1982}, who obtained them from \citet{lehman1975Nonparametric}. Dataset 1 pertains to hours of pain relief due to two different drugs ($n_x = n_y = 8$), and Dataset 2 pertains to the effect of an analgesia for two classes ($n_x = 7, n_y = 10$). The exact and saddlepoint p-values are from \citet{robinson1982}. The the p-value from our resampling algorithm ($\tilde{p}_{\text{pred}}$) is the mean from 100 runs; the first and third quantiles were (0.080, 0.088) for dataset 1, and (0.011, 0.012) for dataset 2.}
\begin{tabular}{ccc}
Method & Dataset 1 & Dataset 2 \\
\hline \hline 
Exact & 0.102 & 0.012 \\
First saddlepoint & 0.089 & 0.010 \\
Second saddlepoint & 0.101 & 0.011 \\
$\tilde{p}_{\text{pred}}$ & 0.083 & 0.012 \\
$\hat{p}_{\text{asym}}$ & 0.092 & 0.013
\end{tabular}
\label{saddlepoint}
\end{table}

\section{Simulations under null hypotheses for single parameters \label{E}}

Recent work, such as that by \citet{chung2013exact}, have extended permutation tests to be valid not only under the null $P_x= P_y$, but also under the more general null that $\theta(P_x) = \theta(P_y)$, where $\theta(P)$ is a single parameter. For example, for $X \sim N(\mu_x, \sigma_x^2), Y \sim N(\mu_y, \sigma_y^2)$, we might be interested in the alternative $H_1: \mu_x \ne \mu_y$, even if $\sigma_x^2 \ne \sigma_y^2$.

As described by \citet{chung2013exact}, in order to obtain a test procedure that is asymptotically valid in the above setting where $\sigma_x^2 \ne \sigma_y^2$, we need to replace $T = |\bar{x} - \bar{y}|$ with the studentized statistic
\begin{equation}
T = \frac{|\bar{x} - \bar{y}|}{\sqrt{s_x^2 / n_x + s_y^2 / n_y}}
\label{tstudent}
\end{equation}
where $s_x^2 = (n_x - 1)^{-1} \sum_i (x_i - \bar{x})^2$ and $s_y^2 = (n_y - 1)^{-1} \sum_j (y_j - \bar{y})^2$ are the sample variances. For each permutation, we compute the quantities $\bar{x}^*, \bar{y}^*, {s_x^*}^2$, and ${s_y^*}^2$ with the permuted datasets. In this section, we conduct simulations using (\ref{tstudent}) when $P_x \ne P_y$ under the null $H_0: \mu_x = \mu_y$ and alternative $H_1: \mu_x \ne \mu_y$.

We generated data $x_{i}, i=1,\ldots,n_x$ and $y_j, j=1,\ldots,n_y$ as realizations of the respective random variables $X_i\overset{\text{iid}}{\sim} N(0, \sigma_x^2)$ and $Y_{j} \overset{\text{iid}}{\sim} N(0, \sigma_y^2)$, where $\sigma_x^2 = 9$ and $\sigma_y^2 = 1$. For equal sample sizes, we set $n=n_x = n_y = 20, 40, 60$, and for unequal sample sizes we set $n_x = 20, 40, 60$ and $n_y = 100$. For both equal and unequal sample sizes, we simulated 1,000 datasets for each combination of parameters. Figures \ref{simDiff_sym_uneqVar_null} and \ref{simDiff_nonSym_uneqVar_null} show the results with equal and unequal sample sizes, respectively.

As seen in Figures \ref{simDiff_sym_uneqVar_null} and \ref{simDiff_nonSym_uneqVar_null}, the permutation test with the unstudentized statistic is relatively unaffected in our simulation under equal sample sizes, but is inaccurate for unequal sample sizes. By using a studentized statistic, our method is accurate even for unequal sample sizes. For comparison, Figures \ref{simDiff_sym_uneqVar_null} and \ref{simDiff_nonSym_uneqVar_null} also show the p-value from a t-test with unequal variance, as well as a Monte Carlo estimate using the unstudentized statistic $T = |\bar{x} - \bar{y}|$.

\begin{figure}[htbp]
\centering
\centering
\includegraphics[scale = 0.5]{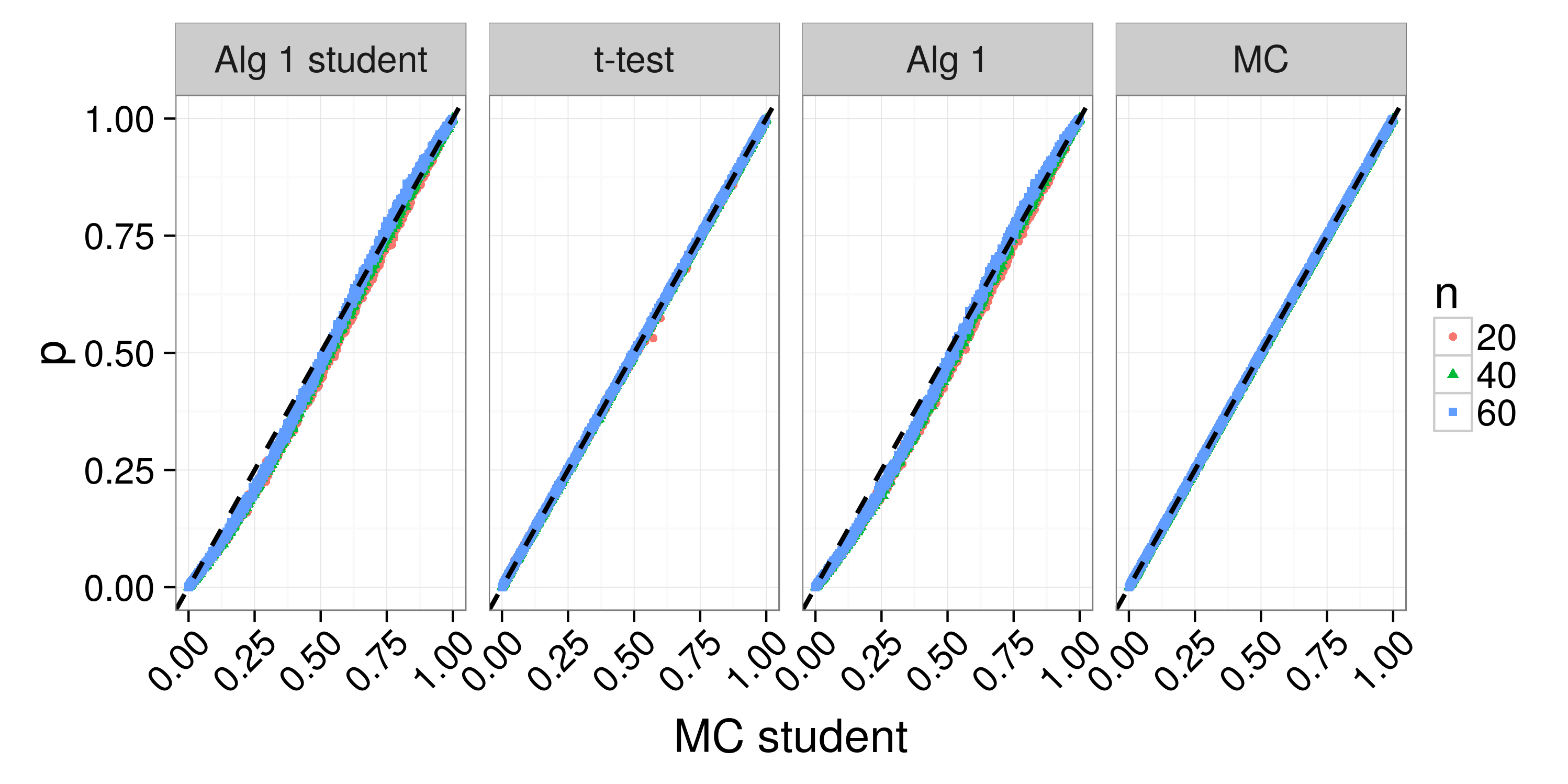}
\caption{Simulation results under the null $\mu_x = \mu_y$ (means $\mu_x = \mu_y = 0$) with normal data and unequal sample sizes of $n=n_x=n_y = 20, 40, 60$. \textit{Alg 1 Student} and \textit{Alg 1} are our resampling algorithm with the studentized (\ref{tstudent}) and unstudentized statistics, and with $B_{\text{pred}}=10^3$ iterations in each partition. \textit{t-test} is the p-value from a two-sided t-test with unequal variance. \textit{MC student} and \textit{MC} are Monte Carlo estimates with the studentized (\ref{tstudent}) and unstudentized statistics, and with $10^5$ iterations. The diagonal dashed line has slope of 1 and intercept of 0, and indicates agreement between methods.}
\label{simDiff_sym_uneqVar_null}
\end{figure}

\begin{figure}[htbp]
\centering
\centering
\includegraphics[scale = 0.5]{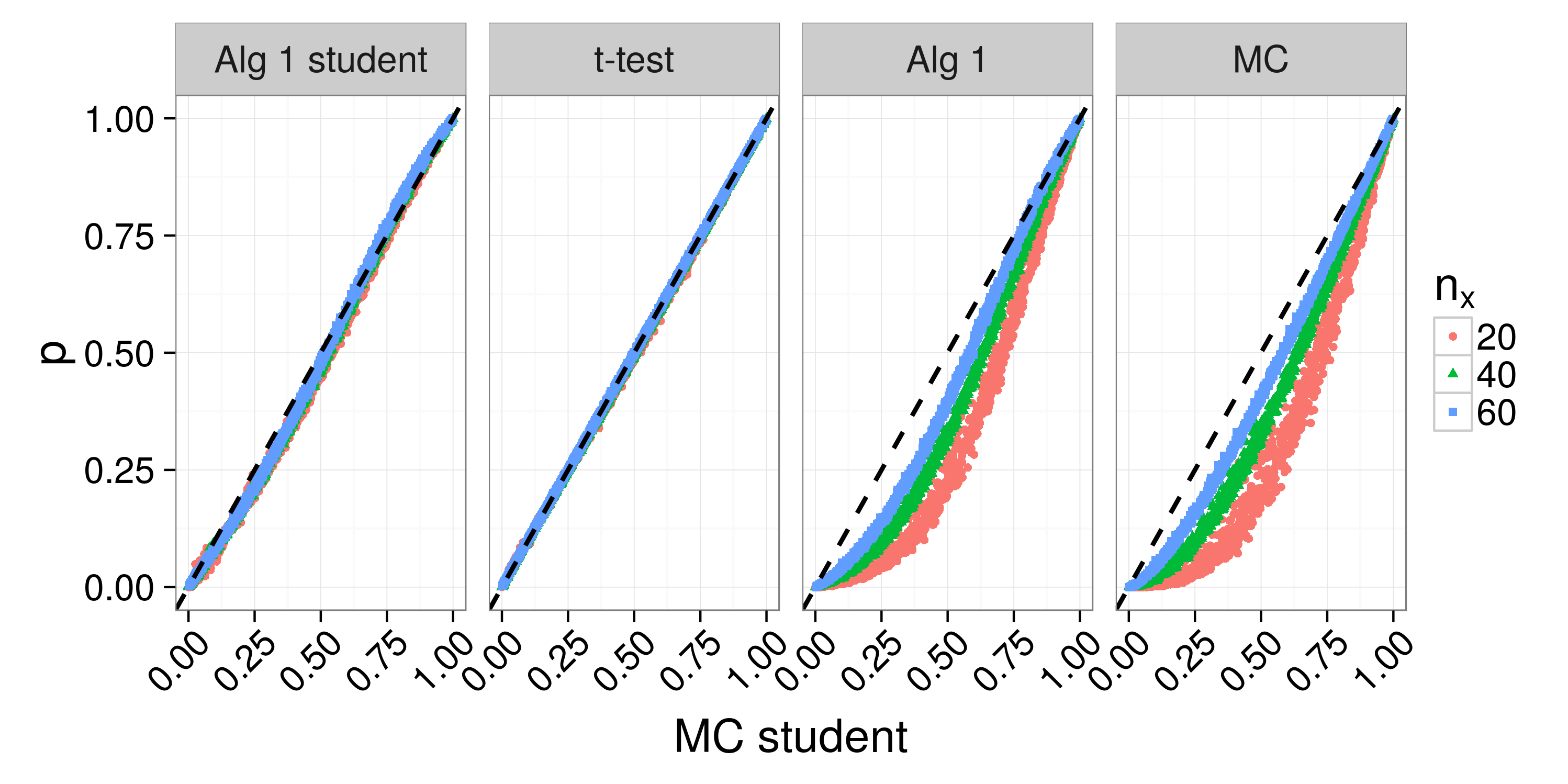}
\caption{Simulation results under the null $\mu_x = \mu_y$ (means $\mu_x = \mu_y = 0$) with normal data with unequal sample sizes of $n_x=20, 40, 60$ and $n_y = 100$. \textit{Alg 1 Student} and \textit{Alg 1} are our resampling algorithm with the studentized (\ref{tstudent}) and unstudentized statistics, and with $B_{\text{pred}}=10^3$ iterations in each partition. \textit{t-test} is the p-value from a two-sided t-test with unequal variance. \textit{MC student} and \textit{MC} are Monte Carlo estimates with the studentized (\ref{tstudent}) and unstudentized statistics, and with $10^5$ iterations. The diagonal dashed line has slope of 1 and intercept of 0, and indicates agreement between methods.}
\label{simDiff_nonSym_uneqVar_null}
\end{figure}

\section{Sufficient sample size \label{F}}

In this section, we provide guidance regarding the sample sizes necessary for our test to be reliable. Recall that $\msa \equiv \min_m \left\{ m \in \{1,\ldots, m_{\max} \} : \Phi^{-1}(1 - 1/B_\text{pred}) < \xi(m) \right\}$. $\msa$ is the expected number of data points available to the Poisson regression in our resampling algorithm for estimating the overall p-value. Large values of $\msa$ imply more reliable but slower estimates, and smaller values of $\msa$ imply less reliable but faster estimates. To ensure that the results of the sampling algorithm are reliable, we recommend that $\msa \ge c$ for some constant $c$. For example, we use $c=4$. Then for equal sample sizes $n = n_x = n_y$, we set
$$
\hat{n} = \min_n \{n \in \mathbb{N} : \msa \ge c \}.
$$

While not explicit in the above notation, we note that $\msa$, and thus $\hat{n}$, is a function of $\sigma_x^2, \sigma_y^2, \mu_x, \mu_y$, and $B_{\text{pred}}$. In Tables \ref{nHatRat} and \ref{nHatDiff}, we set $B_{\text{pred}} = 1,000$, and we also show $\hat{p}_{\text{asym}} = \hat{p}_{\text{asym}}(\hat{n}, \sigma_x^2, \sigma_y^2, \mu_x, \mu_y)$, the the p-value from our asymptotic approximation for the given set of parameter values and sample sizes. As in Figures 1 in Section 3 and Figure S1 in Appendix \ref{proofs}, to obtain $\hat{p}_{\text{asym}}$, we substituted parameter values for sample quantities, e.g. $\mu_x$ for $\bar{x}$ and $\sigma^2_x$ for $(n_x - 1)^{-1} \sum_{i=1}^{n_x} (x_i - \bar{x})^2$. As can be seen in Tables \ref{nHatRat} and \ref{nHatDiff}, $\hat{n}$ and $\hat{p}_{\text{asym}}$ have an inverse relationship.

In general, we recommend that researchers check the output from \verb|fastPerm| to ensure that $m_{\text{stop}} \ge 4$, and we note that the sample sizes required to achieve $m_{\text{stop}} \ge 4$ increase as the p-value decreases. Based on Tables \ref{nHatRat} and \ref{nHatDiff}, at least 15-20 observations in each group appears sufficient for p-values near $1 \times 10^{-6}$, and at least 70-90 observations in each group appears sufficient for p-values near $1 \times 10^{-30}$.

\begin{table}[htbp]
\centering
\caption{$\hat{n}$ for $T = \max(\bar{x} / \bar{y}, \bar{y} / \bar{x})$, equal samples sizes $n_x = n_y = \hat{n}$, $B_{\text{pred}} = 1,000$, and $c = 4$.}
\begin{tabular}{cccc}
\hline \hline
$\mu_y = \sigma^2_y$ & $\mu_x = \sigma^2_x$ & $\hat{n}$ & $\hat{p}_{\text{asym}}$ \\
\hline
\multirow{14}{*}{2} & 3 &   5 & $2.4 \times 10^{-1}$ \\
& 4 &   6 & $2.4 \times 10^{-2}$ \\
& 5 &  13 & $2.4 \times 10^{-5}$ \\
& 5.25 &  16 & $1.3 \times 10^{-6}$ \\
& 5.5 &  19 & $6.0 \times 10^{-8}$ \\
& 5.75 &  24 & $4.2 \times 10^{-10}$ \\
& 6 &  31 & $4.1 \times 10^{-13}$ \\
& 6.25 &  40 & $4.3 \times 10^{-17}$ \\
& 6.5 &  55 & $1.1 \times 10^{-23}$ \\
& 6.6 &  63 & $3.3 \times 10^{-27}$ \\
& 6.7 &  74 & $4.5 \times 10^{-32}$ \\
& 6.8 &  87 & $7.7 \times 10^{-38}$ \\
& 6.9 & 105 & $7.8 \times 10^{-46}$ \\
& 7 & 130 & $6.0 \times 10^{-57}$
\end{tabular}
\label{nHatRat}
\end{table}

\begin{table}[htbp]
\centering
\caption{$\hat{n}$ for $T = |\bar{x} - \bar{y}|$, $\sigma^2_x = \sigma^2_y = 1$, equal samples sizes $n_x = n_y = \hat{n}$, $B_{\text{pred}} = 1,000$, and $c = 4$.}
\begin{tabular}{cccc}
\hline \hline
$\mu_y$ & $\mu_x$ & $\hat{n}$ & $\hat{p}_{\text{asym}}$ \\
\hline
\multirow{11}{*}{0} & 1.5 &   5 & $5.4 \times 10^{-2}$ \\
 & 2 &   9 & $7.7 \times 10^{-4}$ \\
 & 2.2 &  13 & $2.1 \times 10^{-5}$ \\
 & 2.25 &  15 & $3.7 \times 10^{-6}$ \\
 & 2.3 &  18 & $3.1 \times 10^{-7}$ \\
 & 2.4 &  32 & $4.0 \times 10^{-12}$ \\
 & 2.45 &  53 & $2.3 \times 10^{-19}$ \\
 & 2.475 &  80 & $1.3 \times 10^{-28}$ \\
 & 2.48 &  89 & $1.1 \times 10^{-31}$ \\
 & 2.49 & 115 & $1.5 \times 10^{-40}$ \\
 & 2.5 & 165 & $1.4 \times 10^{-57}$
\end{tabular}
\label{nHatDiff}
\end{table}

\section{p-value for ratio of means via the delta method, and application to cancer genomic data \label{G}}

Let $\bar{x}$ and $\bar{y}$ be the sample means, and $s_x^2 = (n_x - 1)^{-1} \sum_i (x_i - \bar{x})^2$ and $s_y^2 = (n_y - 1)^{-1} \sum_i (y_i - \bar{y})^2$ be the sample estimates of variance. By the central limit theorem, for $n_x, n_y$ sufficiently large, and assuming independence between samples,
$$
\begin{pmatrix}
\bar{x} \\
\bar{y}
\end{pmatrix} \sim
N\left(
\begin{bmatrix}
\mu_x \\ 
\mu_y
\end{bmatrix},
\begin{bmatrix}
\sigma_x^2 / n_x & 0 \\
0 & \sigma_y^2 / n_y
\end{bmatrix}
\right).
$$
Let $g(\bar{x}, \bar{y}) = (\bar{x} / \bar{y})$. Then $\nabla g = (1 / \bar{y}, -\bar{x} / \bar{y}^2)'$, and by the delta method $\bar{x} / \bar{y} \rightarrow N(\theta, \tau_1^2)$, where $\theta = g(\mu_x, \mu_y) = \mu_x / \mu_y$ and 
$$
\tau_1^2 = \nabla g^T(\mu_x, \mu_y) \begin{bmatrix}
\sigma_x^2 / n_x & 0 \\
0 & \sigma_y^2 / n_y
\end{bmatrix} \nabla g(\mu_x, \mu_y) = \frac{\sigma_x^2}{n_x} \frac{1}{\mu_y^2} + \frac{\sigma_y^2}{n_y} \frac{\mu_x^2}{\mu_y^4}.
$$
Using unbiased estimates for the variance, we get
$$
\hat{\tau_1}^2 = \frac{s_x^2}{n_x} \frac{1}{\bar{y}^2} + \frac{s_y^2}{n_y} \frac{\bar{x}^2}{\bar{y}^4}.
$$
Similarly, we estimate the variance of $\bar{y} / \bar{x}$ as
$$
\hat{\tau_2}^2 = \frac{s_y^2}{n_y} \frac{1}{\bar{x}^2} + \frac{s_x^2}{n_x} \frac{\bar{y}^2}{\bar{x}^4}.
$$
Therefore, to test the null $H_0: \theta = 1$ versus the alternative $H_1: \theta \ne 1$, the two-sided p-value using the delta method and unbiased estimates of variance is
\begin{align*}
p_\Delta = 
\begin{cases}
\Pr(Z > \bar{x} / \bar{y}) + \Pr(U \le \bar{y}  / \bar{x}), \quad \bar{x} / \bar{y} \ge 1 \\
\Pr(U > \bar{y} / \bar{x}) + \Pr(Z \le \bar{x}  / \bar{y}), \quad \bar{x} / \bar{y} < 1
\end{cases},
\end{align*}
where $Z \sim N(1, \hat{\tau_1}^2)$ and $U \sim N(1, \hat{\tau_2}^2)$. We use the $\Delta$ subscript in $p_\Delta$ to emphasize that the p-value is from the delta method. We note that $p_\Delta$ is potentially problematic, particularly if $\hat{\tau_1}^2$ or $\hat{\tau_2}^2$ are large, because the ratio is bounded below by zero, but the normal distribution is not.

Figure \ref{pDelta} compares estimates of the permutation p-values from our resampling algorithm ($\tilde{p}_{\text{pred}}$) to $p_\Delta$ for the cancer genomic data in Section 6. The dashed lines have an intercept of zero and slope of one, and indicate agreement. As seen in Figure \ref{pDelta}, $p_\Delta$ tends to be an overestimate for small p-values, which is the same trend observed in the simulations. None of the 15 genes with the smallest $p_\Delta$ were identified by \citet{zhan2015} as strongly distinguishing between LUAD and LUSC. Out of the 100 genes with the smallest $p_\Delta$, three were identified by \citet{zhan2015} as strongly distinguishing between LUAD and LUSC (\textit{ATP1B3}, \textit{PVRL1}, and \textit{PERP}).
\begin{figure}[htbp]
\centering
\begin{subfigure}{0.48\textwidth}
  \includegraphics[width = \linewidth]{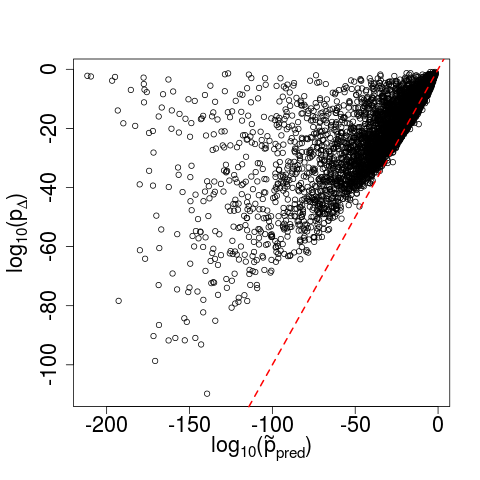}
  \caption{Genes with $\tilde{p} \le 1 \times 10^{-3}$ ($10,302$ genes)}
  \label{pDelta_fp}
\end{subfigure}
\begin{subfigure}{0.48\textwidth}
  \includegraphics[width = \linewidth]{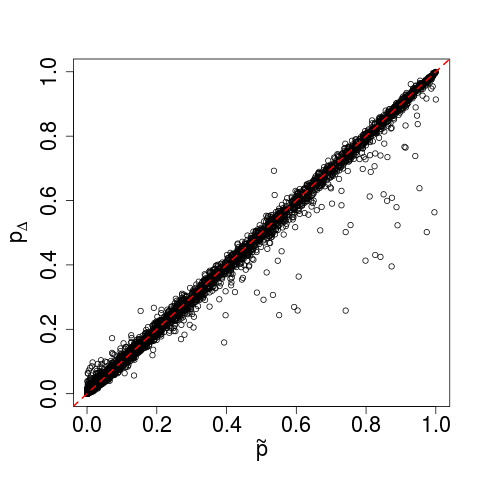}
  \caption{Genes with $\tilde{p} > 1 \times 10^{-3}$ ($5,084$ genes)}
  \label{pDelta_MC}
\end{subfigure}
  \caption{p-values for cancer genomic data: Comparison of results with the delta method ($p_\Delta$) and our resampling algorithm ($\tilde{p}_\text{pred}$) with $B_{\text{pred}} = 1,000$ iterations within each partition, or with simple Monte Carlo ($\tilde{p}$) with a total of $B=1,000$ iterations (see Section 6). the dashed lines have intercept of zero and slope of one, and indicate agreement between the methods.}
\label{pDelta}
\end{figure}

\end{appendices}
\end{document}